\documentclass[11pt]{article}

\usepackage[matrix,arrow]{xy}

\usepackage{epsfig}

\usepackage{amssymb,amsthm}
\usepackage{amsmath}
\usepackage{amscd}
\usepackage{amsfonts}
\usepackage{latexsym}
\usepackage{graphics}
\usepackage{color}
\usepackage{bbm}
\usepackage{bm}

\usepackage{xypic}
\usepackage{mathrsfs}

\usepackage{graphicx}

\input{xy}
\xyoption{all}

\newtheorem{theorem}[equation]{Theorem}
\newtheorem{lemma}[equation]{Lemma}
\newtheorem{proposition}[equation]{Proposition}
\newtheorem{definition}[equation]{Definition}

\newtheorem{corollary}[equation]{Corollary}

\catcode`@=11
\renewcommand{\section}
{\@startsection{section}{1}{0pt}{\medskipamount}{\medskipamount}{\large\bf}}
\makeatletter\renewcommand{\subsection}
{\@startsection{subsection}{2}{\z@}{-3.25ex plus -1ex minus -.2ex}
{1.5ex plus .2ex}{\it }}
\makeatletter\renewcommand{\subsubsection}
{\@startsection{subsubsection}{3}{\z@}{-3.25ex plus -1ex minus -.2ex}
{1.5ex plus .2ex}{\noindent\underline}}

\numberwithin{equation}{section}
\catcode`@=12

\def\={\ =\ }
\def\dd{{\rm d}}

\newcommand{\Tr}[1]{\:{\rm Tr}\,#1}
\def\e{{\,\rm e}\,}
\newcommand{\mbf}[1]{{\boldsymbol {#1} }}

\newcommand{\hol}{{\sf hol}}

\newcommand{\Tor}{{\sf Tor}}
\newcommand{\Fred}{{\tt Fred}}
\newcommand{\Heis}{{\tt Heis}}
\newcommand{\Pic}{{\tt Pic}}
\newcommand{\ort}{{\tt or}}
\newcommand{\ttt}{{\tt t}}
\newcommand{\ttc}{{\tt c}}

\newcommand{\CS}{{\sf CS}}
\newcommand{\rank}{{\tt rank}}

\renewcommand{\Re}{\ensuremath{\mathfrak{Re}}}

\newcommand{\nablatr}{{{\Large\blacktriangledown}}}

\DeclareMathOperator{\K}{\rm K}
\DeclareMathOperator{\E}{\rm E}

\DeclareMathOperator{\Hom}{\sf Hom}
\DeclareMathOperator{\RH}{\rm H}

\DeclareMathOperator{\End}{\sf End}
\DeclareMathOperator{\Aut}{\sf Aut}
\DeclareMathOperator{\Ext}{\sf Ext}
\DeclareMathOperator{\Comp}{\sf Comp}
\DeclareMathOperator{\Symm}{\sf Symm}

\DeclareMathOperator{\ch}{ch}

\DeclareMathOperator{\Index}{\sf Index}

\def\ii{{\,{\rm i}\,}}

\newcommand{\cG}{\mathcal{G}}

\newcommand{\cale}{\mathcal{E}}

\newcommand{\cals}{\mathcal{S}}

\newcommand{\calu}{\mathcal{U}}

\newcommand{\Vcal}{\mathcal{V}}

\newcommand{\fM}{\mathfrak{M}}

\newcommand{\fD}{\mathfrak{D}}

\newcommand{\unit}{\mathbbm{1}}   			% identity map/matrix

\newcommand{\scrF}{\mathscr{F}}

\newcommand{\scrH}{\mathscr{H}}
\newcommand{\scrG}{\mathscr{G}}
\newcommand{\scrE}{\mathscr{E}}
\newcommand{\scrK}{\mathscr{K}}
\newcommand{\scrAb}{\mathscr{A}b}
\newcommand{\Bord}{\mathscr{B}ord}
\newcommand{\Vect}{\mathscr{V}ect}
\newcommand{\scrZ}{\mathscr{Z}}

\newcommand{\Ocal}{\mathcal{O}}

\def\alg{{\mathcal A}}

\def\hil{{\mathcal H}}
\def\bun{{\mathcal E}}

\newcommand{\IZ}{\mathbb{Z}}
\newcommand{\IC}{\mathbb{C}}

\newcommand{\IN}{\mathbb{N}}
\newcommand{\IR}{\mathbb{R}}

\newcommand{\calH}{{\cal H}}

\newcommand{\cO}{{\cal O}}

\def\Id{{\rm id}}
\def\pt{{\rm pt}}

\def\ch{{\sf ch}}

\def\PD{{\sf Pd}}

\newcommand{\bbt}{{\mathbb T}}

\newcommand{\calc}{{\mathcal C}}

\def\Pic{{\sf Pic}}
\def\Dirac{{D\!\!\!\!/\,}} % Dirac operator
\newcommand{\torus}{{\mathbb{T}}}

\def\e{\epsilon}

\def\beq{\begin{equation}}
\def\eeq{\end{equation}}
\def\bea{\begin{eqnarray}}
\def\eea{\end{eqnarray}}
\def\beqa{\begin{eqnarray*}}
\def\eeqa{\end{eqnarray*}}

\renewcommand{\e}{\,\mathrm{e}\,}
\newcommand{\im}{\,\mathrm{i}\,}

\newcommand{\R}{{\mathbb{R}}}

\newcommand{\uo}{{{\rm U}(1)}}

\def\Dirac{{D\!\!\!\!/\,}} % Dirac operator

\def\Hom{{\sf Hom}}
\def\Ext{{\sf Ext}}
\def\End{{\sf End}}

\def\>{\rangle}
\def\<{\langle}
\def\+{\dagger}
\def\={\ =\ }

\begin{document}

%\begin{titlepage}
\setcounter{page}{1}
\begin{flushright}
HWM--12--13\\
EMPG--12--20\\
ESI--2385
\end{flushright}

\vskip 0.5cm

\begin{center}

{\Large\bf Quantization of Higher Abelian Gauge Theory \\[6pt] in
  Generalized Differential Cohomology}\footnote{Invited Contribution
  to the Proceedings of the 7th International Conference on
  Mathematical Methods in Physics, Rio de Janeiro, Brazil, April 16--20, 2012; to be
  published in {\sl Proceedings of Science}.}

\vspace{5mm}

{\large Richard J. Szabo}
\\[3mm]
\noindent {\em Department of Mathematics\\ Heriot--Watt University\\
Colin Maclaurin Building, Riccarton, Edinburgh EH14 4AS, U.K.} \\ and
\\ \emph{Maxwell Institute
  for Mathematical Sciences, Edinburgh, U.K.}
\\[3mm]
{Email: {\tt R.J.Szabo@hw.ac.uk}}

\bigskip

\begin{abstract}

\noindent
We review and elaborate on some aspects of the quantization of
certain classes of higher abelian gauge theories using techniques of generalized
differential cohomology. Particular emphasis is placed on the examples
of generalized Maxwell theory and Cheeger--Simons cohomology, and of
Ramond--Ramond fields in Type~II superstring theory and differential K-theory.

\end{abstract}

%% AMS subject classification (MSC2000):  19-02 , 19M05

\end{center}
%\end{titlepage}

\bigskip

{\baselineskip=12pt
\tableofcontents
}

%\newpage

\bigskip

\section{Introduction\label{Intro}}

\noindent
This paper is devoted to a survey of some topics in the mathematical
formulation of generalized abelian gauge theories as they arise in
string theory and M-theory. It has been realised
that the proper mathematical treatment of the quantization of such
systems involves techniques from generalized differential cohomology
theories, and this has sparked a wealth of activity in both physics
and mathematics in recent years, as well as intensive interactions
between the two disciplines. Our presentation is neither complete nor
is it exhaustive in the treatment of topics we have chosen to
cover. Instead, we broadly overview various topics, presenting some
aspects with a different emphasis compared to other treatments of the
subject, and also elucidate
certain calculational details which are probably well-known to the
experts but which we have not found in the
literature. Our presentation is intentionally mathematical but with
all motivations, discussions and prejudices towards certain features
inspired by physical considerations from string theory; we have
attempted to define all pertinent mathematical concepts for the
non-experts whilst describing physical concepts more formally.

Before describing the precise contents of this article, let us explain the
general setting we shall consider and which important omissions the
reader can anticipate in the following sections. The different theories we are
interested in all have the same low-energy limit, which is ten-dimensional
supergravity defined on a manifold $M$; we can generalize $M$ to
include groupoids which enables us to formulate the theory on e.g. orbifolds. The relevant field content includes a riemannian metric $g$, a
dilaton field $\phi:M\to\IR$, a closed three-form $H$, and the
Ramond--Ramond gauge fields $F$ which are inhomogeneous differential forms on
$M$ that are $H$-twisted closed, i.e. $(\dd+H\wedge)F=0$. The relevant
string theories are quantized by certain generalized differential
cohomology theories $\check\E{}^\bullet(M)$ depending on the physical
constraints imposed, which gives a suitable lattice in de~Rham
cohomology $\RH^\bullet(M;\IR)$ where the fields should live in order
to enforce charge quantization; in this paper these
theories will always be certain smooth refinements of ordinary
cohomology and K-theory. Examples of theories with the correct
low-energy limit include bosonic string theory which is
quantized by ordinary differential cohomology
$\check\RH{}^3(M)$, heterotic string theory where $H$ is not closed
and which is thereby described by differential cochains, and
Type~I string theory which involves self-dual fields that are quantized by
differential KO-theory $\check{\rm KO}(M)$. For Type~II superstring
theory, the quantization of the $B$-field is somewhat more involved and
has been studied thoroughly in~\cite{DFM,DFM2}; the issue is to
reconcile the worldsheet and target space descriptions of the $B$-field,
which also serves as a differential twisting of K-theory in which the
suitable refinements of the Ramond--Ramond fields $F$ to classes or
cocycles should live. This latter construction includes Type~I string
theory as a special case of the Type~II theory via an orientifold construction.
Various generalizations of this theme can also be considered: For
instance in
Kaluza--Klein compactifications one sets $M=X\times K$ with $K$ a compact
``internal'' manifold and regards the gauge fields as forms $F\in\Omega^\bullet(X;{\tt
  harm}^\bullet(K))$ where ${\tt harm}^\bullet(K)$ is the vector space of harmonic
forms on $K$, while if $M$ is unoriented then the dual gauge fields
live in $\Omega^\bullet(M;{\tt or}(M))$ where ${\tt or}(M)$ is the
orientation bundle of $M$. With this setting in mind, we can now spell out a few of our main
omissions of topics.

Firstly, we do not discuss in detail the appearence or role of
anomalies, which are a physical driving force for the uses of K-theory
in string theory and should be properly addressed within some
framework of categorified index theory; see~\cite{Freed:2000ta,Freed:2000tt,Freed:2004yc}
for extensive discussions on the importance of generalized
differential cohomology theories in this context.

Secondly, we do not
consider twistings of our cohomology theories. The ingredients of
twisted generalized differential cohomology theories are described
in~\cite{Kahle}. Twisted differential K-theory is defined in~\cite{Carey}
using the twisting groupoid of abelian bundle gerbes with
2-connection; this construction is then based on sections of bundles
of Fredholm operators.

Finally, we do not consider orbifolds. They represent a broad class of consistent
string backgrounds with rich interesting features for
which equivariant geometric versions of K-theory seem to have the
appropriate features; here one should use the suggestion~\cite{Sharpe} that string orbifolds are best regarded as
quotient stacks, at least from the perspective of the
worldsheet sigma-model. Differential equivariant complex
K-theory is developed in~\cite{SV,Ortiz,bunke-20092}, while an account of orbifold Ramond--Ramond fields in the path integral
framework and with more general backgrounds can be found
in~\cite{DFM,DFM2}. The original model of~\cite{SV} is defined
using classifying spaces for equivariant K-theory and Bredon cohomology with coefficients in
the real representation ring of the orbifold group; this
approach nicely captures salient features of Ramond--Ramond fields on
Type~II orbifolds, such as flux quantization. The subsequent
construction of~\cite{Ortiz} is rooted in homotopy theory and utilizes
delocalized de~Rham equivariant cohomology; in this approach a ring
structure and push-forwards are readily constructed by taking
real-valued forms as fixed points of a suitable real
structure. Finally, the model of~\cite{bunke-20092} is set in the framework of equivariant local index theory
and differentiable \'etale stacks, and uses delocalized de~Rham
cohomology; this approach is powerful enough to construct all necessary
ingredients in gauge theory such as products, push-forwards, an intersection pairing,
and an $\IR$-valued subfunctor, at the price of using broad
classes of cocycles (geometric families of Dirac operators), some of
which have no immediate physical interpretation in terms of
Ramond--Ramond gauge theory.

The outline of the remainder of this paper, together with directions
for complementary reading, is as follows.

In \S\ref{Abgauge}, we start from the Maxwell gauge theory of
electromagnetism and generalize it to higher abelian gauge
theories. We explain how Dirac charge quantization is naturally rooted
in a description involving Cheeger--Simons cohomology; foundational
aspects of this line of reasoning can be found in~\cite{Freed:2000ta},
while a more
pedagogical introduction geared at physicists
is~\cite{Freed:2006yc}. We also describe the configuration space of
abelian gauge fields in details, and how to properly incorporate
currents. Gauge fields are
modelled by groupoids, and equivalent groupoids are equally good for
the purposes of defining the functional integral of the quantum gauge
theory; we explore one such model which is contained in the seminal
paper~\cite{Hopkins:2002rd}. A complementary review can be found
in~\cite{Valentino:2008xd}, while more detailed mathematical aspects of 
differential cohomology are reviewed in~\cite{Bunke-rev}. 

In \S\ref{RRfields}, we begin by briefly reviewing the relationship
between D-branes and K-theory: We introduce the notion of D-brane, describe
D-branes as K-cycles for geometric K-homology, explain the physical
relevance of the Atiyah--Bott--Shapiro
construction, and give the formula for D-brane charges. This leads us
into the notion of Ramond--Ramond charges and how they are related to the
semi-classical quantization of Ramond--Ramond fields in Type~II
superstring theory. We survey various models for
differential K-theory from a physical perspective, describe the mathematical properties of flat
Ramond--Ramond fields, and formulate the gauge theory of
Ramond--Ramond fields, focusing in particular on their realization as
self-dual fields and the use of differential K-theory in the
description of their holonomies on D-branes. A review of certain
aspects of D-branes and K-theory from a mathematical perspective can be found
in~\cite{Szabo:2008hx}, while various aspects of differential K-theory in string
theory is presented in~\cite{Freed:2002qp}. A mathematical survey of differential K-theory
is found in~\cite{Bunke:2010mq}. 

Finally, in \S\ref{HamQuant} we begin with a quick overview of the
mathematical formulation of functional integral quantization of
generalized abelian gauge theories, together with many explicit examples. We
then consider in some detail the hamiltonian quantization of self-dual generalized
abelian gauge fields using the concept of Pontrjagin self-duality. We
apply this formalism to the quantization of
Ramond--Ramond fields, through the theory of Heisenberg groups and their
representations; the seminal work on this approach to quantization is~\cite{Freed:2006ya}. We work only
with the simplest backgrounds that contain no $H$-flux or D-branes, and
explicitly carry out the hamiltonian quantization of the
Ramond--Ramond gauge theory by constructing the pertinent Heisenberg
groups along the lines of~\cite{Freed:2006ya}. By choosing the natural
polarization on the configuration space of the self-dual
Ramond--Ramond gauge theory, the Heisenberg group admits a unique
irreducible unitary representation which is identified as the quantum
Hilbert space of the gauge theory.

\bigskip

\section{Abelian gauge theory and differential cohomology
  \label{Abgauge}}

\subsection{Maxwell theory}

In undergraduate physics courses on classical electromagnetic theory
one learns about perhaps the most fundamental set of equations in
physics, the Maxwell equations; these equations govern all classical
electromagnetic phenomena and are responsible for much of modern
technology. In this section we begin with a mathematical
introduction to the classical Maxwell theory.

Let $M=\IR\times N$ be a four-manifold with lorentzian signature
metric $\dd t\otimes \dd t-g$, where $t\in\IR$ parametrizes the
``time'' direction and $(N,g)$ is a connected riemannian
three-manifold which we will think of as ``space''. Classical
electromagnetism takes place in \emph{Minkowski spacetime} where $N$
is taken to be the vector space $\IR^3$.

Introduce a pair of differential forms
\beqa
F\in\Omega^2(M) \qquad \mbox{and} \qquad j_e\in\Omega_c^3(M) \ ,
\eeqa
where the two-form $F$ is called the gauge field strength or flux, while
$j_e$ is a differential form of compact spatial support called the
electric current. Maxwell's equations then read
\beq
\dd F=0 \qquad \mbox{and} \qquad \dd \star F=j_e \ ,
\label{Maxwelleqs}\eeq
where $\star$ is the Hodge duality operator associated to the
lorentzian metric on $M$. Consistency of the second equation in (\ref{Maxwelleqs}) requires the conservation
law
\beqa
\dd j_e=0 \ ,
\eeqa
which defines the electric charge
\beqa
q_e:=\big[j_e\big|_N\,\big] \ \in \ \RH_c^3(N;\IR)\cong \IR \ .
\eeqa

We can generalize the first equation of (\ref{Maxwelleqs}) to include
a magnetic current three-form $j_m$ with
\beqa
\dd F=j_m \ ,
\eeqa
so that $\dd F=0$ outside the support of $j_m$. This defines the
magnetic charge
\beqa
q_m:=\big[j_m\big|_N\,\big] \ \in \ \RH_c^3(N;\IR)\cong \IR \ .
\eeqa
The magnetic current vanishes in the ``classical theory'', but the
``quantum theory'' allows for it; as we discuss below, this leads to the quantization of
electric \emph{and} magnetic charge. Since the two-form $F$ is not
required to have compact support, the charges $q_{e}$ and $q_m$ are non-zero generally; in
fact, they live in the kernel of the natural forgetful map which
forgets about the compact support condition,
\beqa
q_{e},q_{m} \ \in \ \ker\big(\RH_c^3(N;\IR) \ \longrightarrow \ \RH^3(N;\IR)
\big) \ ,
\eeqa
since the currents $j_{e}$ and $j_{m}$ are trivialized by the flux $F$ on the
interior of the three-manifold $N$.

Absence of magnetic charge in the classical theory implies that the
de~Rham cohomology class of the flux is trivial,
\beq
[F]_{\rm dR}=0 \qquad \mbox{in} \quad \RH^2(M;\IR) \ .
\label{FdR0}\eeq
Hence there exists a one-form $A\in\Omega^1(M)$ such that
\beq
F=\dd A \ .
\label{FdA}\eeq
The \emph{gauge potential} $A$ is only defined up to
gauge transformations $A\mapsto A+\alpha$ by closed differential one-forms
$\alpha\in\Omega_{\rm cl}^1(M)$. If $M$ is contractible (in particular
if $N=\IR^3$), then the two conditions (\ref{FdR0}) and (\ref{FdA})
are automatically satisfied and equivalent to each other. In general,
the space of classical electromagnetic fields modulo gauge
transformations is the infinite-dimensional abelian Lie group
\beq
\scrF_{\rm class}(M)=\Omega^1(M)\,\big/\, \Omega_{\rm cl}^1(M) \ .
\label{calfclassM}\eeq
Alternatively, we can take $A$ to be a connection on a principal
$\IR$-bundle over $M$ (with the additive group structure on the real
line $\IR$);
the space of such connections is an affine space based on
$\Omega^1(M)/\dd \Omega^0(M)$, and the quotient by gauge equivalence classes
of flat connections is an affine space modelled on~(\ref{calfclassM}). This is the correct configuration space of fields
for classical electromagnetism; as we discuss below, the story is
rather different for the quantum theory. 

It is also possible to derive these results from an action
principle. Maxwell's equations (\ref{Maxwelleqs}) in this instance are the variational
equations for the action functional
\beq
S_M[A]=\int_M\, \Big(-\frac12\, \dd A\wedge\star \dd A+A\wedge j_e\Big)
\label{Maxwellaction}\eeq
with respect to compactly supported variations of $A$ such that
$\int_M\, \dd(A\wedge\star F)=0$. Then $S_M[A+\alpha]= S_M[A]$ for
$\alpha\in\Omega^1_{\rm cl}(M)$ up to an exact term $-\int_M\,
\dd(\alpha \wedge \star F)$; in this sense the action functional
(\ref{Maxwellaction}) is classically well-defined on the quotient
space (\ref{calfclassM}).

\subsection{Semi-classical quantization\label{Maxwellsemi}}

If $\RH^2(N;\IR)$ is non-trivial, we may well have $[F]_{\rm dR}\neq
0$ in $\RH^2(M;\IR)$, e.g. outside the support of a magnetic current
$j_m$ in $M$. The \emph{Dirac quantization condition} states that the
de~Rham cohomology class of the flux sits in a lattice
\beqa
\mbox{$\frac1{2\pi}$}\, [F]_{\rm dR} \ \in \ \Lambda\subset\RH^2(M;\IR) \ ,
\eeqa
where
\beqa
\Lambda = \RH^2(M;\IZ) \, \big/ \, \Tor\, \RH^2(M;\IZ)
\eeqa
is the full lattice $\RH^2(M;\IZ)\hookrightarrow \RH^2(M;\IR)$ induced
in cohomology by the inclusion of abelian groups $\IZ\hookrightarrow\IR$, whose kernel
consists of torsion classes in $\RH^2(M;\IZ)$.

Let us pause to briefly explain how this is related to the usual
notion of Dirac charge quantization in physics (see
e.g.~\cite{Freed:2000ta} for further details). Let us take $N=\IR^3$,
and the
electric and magnetic currents to be of the form
\beqa
j_e=q_e\, \delta_W \qquad \mbox{and} \qquad j_m=q_m\,\delta_\IR \ ,
\eeqa
where $\delta_W=\PD_M(W)$ is Poincar\'e dual to an oriented
one-manifold $W\subset M$ (the ``worldline'' of a charged particle), while $\delta_\IR$ is a distributional
three-form on $M=\IR\times\IR^3$ dual to the one-manifold
$\IR\times0\subset M$. The global obstruction to the local
representation (\ref{FdA}) on $\IR\times(\IR^3-0)$ is then
\beq
\int_{S^2}\, F=\int_{B^3}\, \dd F=\int_{B^3}\, q_m\, \delta_\IR = q_m \ ,
\label{globalobstr}\eeq
where $S^3=\partial B^3$ is the unit sphere in $t\times(\IR^3-0)$ for
all $t\in \IR$. This obstruction is due to the non-trivial cohomology
$\RH^2(\IR^3-0;\IR)\neq0$ and it may be thought of as originating through
the \emph{Dirac string}, which is a semi-infinite solenoid represented
by a ray from the origin $0\in\IR^3$; requiring the string to be
physically invisible then yields the global obstruction
(\ref{globalobstr}). The Dirac quantization law then ensures that in the quantum theory
the exponentiated charge coupling $\exp\big(\im \int_M\, A\wedge
j_e\big)=\exp\big(\ii q_e\, \int_W\, A\big)$ is
well-defined. Since $\int_W\, A\in \IR/q_m\, \IZ$ by
(\ref{globalobstr}), the coupling is well-defined if
\beq
q_e\, q_m \ \in \ 2\pi\,\IZ \ .
\label{qeqmquant}\eeq
In the following we will interpret the quantization condition
(\ref{qeqmquant}) geometrically. The wavefunction of a
non-relativistic quantum mechanical particle of charge $q_e$ on
$\IR^3-0$ is a section of a line bundle associated to the
representation
\beqa
\IR\,\big/\, q_m\, \IZ \ \longrightarrow \ \uo \ , \qquad x \
\longmapsto \ \e^{\ii q_e\, x} \ ,
\eeqa
and is therefore well-defined if and only if (\ref{qeqmquant})
holds. In this sense the quantization of charge is a consequence of
the compactness of the gauge group $\IR/q_m\, \IZ$.

There is an elegant model in differential geometry which combines
locality of the gauge field $F$ with global obstructions, including
Dirac charge quantization. For this, we take $F$ to be the curvature
of a connection $A$ on a principal $\bbt$-bundle $\pi:L\to M$, $\bbt=\IR/\IZ$, with
first Chern class $c_1(L)\in\RH^2(M;\IZ)$, i.e. $A$ is a
right-invariant one-form on $L$ such that $\dd A=\pi^*F$. Then the
classical configuration space of fields (\ref{calfclassM}) is replaced
by the quantum groupoid of fields $\scrF_{\rm qu}(M)$. Recall that a \emph{groupoid} is a small category in which all morphisms
(viewed as arrows between objects) are invertible. The objects of the
category $\scrF_{\rm qu}(M)$ are principal $\bbt$-bundles with
connection, while its morphisms are connection-preserving bundle
isomorphisms (gauge transformations). It has the structure of an
infinite-dimensional abelian Lie group, with unit the trivial bundle, under tensor product of circle
bundles with connection. The set of isomorphism classes
$\pi_0{\scrF_{\rm qu}(M)}$ is an infinite-dimensional abelian Lie
group that fits into an exact sequence
\beqa
0 \ \longrightarrow \ \RH^1(M;\bbt) \ \longrightarrow \
\pi_0{\scrF_{\rm qu}(M)} \ \xrightarrow{ \ F \ } \ \scrF_{\rm
  class}(M) \ \longrightarrow \ 0
\eeqa
which describes the quantum configuration space as an extension of the
classical one (\ref{calfclassM}) by gauge equivalence classes of flat
connections; whence such connections are detectable quantum
mechanically, but not classically. Symbolically,
\beqa
\pi_0{\scrF_{\rm qu}(M)} = \bigsqcup_{c_1\in\RH^2(M;\IZ)}\, \alg(L_{c_1}) \,
\big/ \, \cG \ ,
\eeqa
where $\alg(L_{c_1})$ is the affine space of smooth connections on a
line bundle of first Chern class $c_1$ while
$\cG=\Omega^0(M;\bbt)$ is the gauge group. The group $\RH^1(M;\bbt)$ of flat fields can be
described as follows. The short exact sequence of abelian groups
\beq
0 \ \longrightarrow \ \IZ \ \hookrightarrow \ \IR \ \longrightarrow \
\bbt=\IR/\IZ \ \longrightarrow \ 0
\label{ZRTseq}\eeq
induces a short exact
sequence in cohomology
\beqa
0 \ \longrightarrow \ \RH^1(M;\IZ)\otimes\bbt \ \longrightarrow \
\RH^1(M;\bbt) \ \xrightarrow{ \ \beta \ } \ \Tor\, \RH^2(M;\IZ) \
\longrightarrow \ 0 \ \,
\eeqa
where the group $\RH^1(M;\IZ)\otimes\bbt$ is the identity component of
$\RH^1(M;\bbt)$, the torsion subgroup of $\RH^2(M;\IZ)$ is the group of components of
$\RH^1(M;\bbt)$, and $\beta$ is the Bockstein homomorphism.

Another point of view of the space of quantum fields $\pi_0{\scrF_{\rm
  qu}(M)}$ is as the group of holonomies
\beqa
\chi_A \,:\, Z_1(M) \ \longrightarrow \ \uo \ , \qquad
\chi_A(\gamma)=\exp\Big( \im \oint_\gamma\, A \Big) \ ,
\eeqa
where $Z_p(M)$ is the group of smooth $p$-cycles on $M$. Such group
homomorphisms are characterized by the feature that there is a unique integral
two-form $F\in 2\pi\, \Omega_\IZ^2(M)$ with
\beqa
\chi_A(\partial D)=\exp\Big(\im \int_D\, F\Big) \qquad \mbox{for}
\quad D\in C_2(M) \ ,
\eeqa
where $\Omega_\IZ^p(M)$ is the lattice of closed $p$-forms on $M$ with
integer periods, and $C_p(M)$ denotes the group of smooth $p$-chains on $M$. Then the group $\pi_0{\scrF_{\rm
  qu}(M)}$ may be characterized by the short exact sequence
\beq
0 \ \longrightarrow \ \Omega^1(M)\, \big/ \, \Omega^1_\IZ(M) \
\longrightarrow \ \pi_0{\scrF_{\rm
  qu}(M)} \ \xrightarrow{ \ c_1 \ } \ \RH^2(M;\IZ) \ \longrightarrow \ 0 \ ,
\label{FquMshort}\eeq
where the kernel of the characteristic class map $c_1$
consists of connections on the trivial line bundle over $M$ modulo
gauge equivalence; the quotient in (\ref{FquMshort}) is the group of
topologically trivial one-form fields on $M$. More generally, the
space of quantum fields completes the pullback square
\beqa
\xymatrix{
\pi_0{\scrF_{\rm
  qu}(M)} \ \ar[r] \ar[d] & \ \Omega_{\rm cl}^2(M) \ar[d] \\
\RH^2(M;\IZ) \ \ar[r] & \ \RH^2(M;\IR)
}
\eeqa
so that $\pi_0{\scrF_{\rm
  qu}(M)}\subset \RH^2(M;\IZ)\times \Omega_{\rm cl}^2(M)$; thus a quantum
field is a representative in $\Omega_{\rm cl}^2(M)$ of a first Chern
class $c_1\in\RH^2(M;\IZ)$
in de~Rham cohomology.

\subsection{Higher abelian gauge theory\label{GAGT}}

There are various generalizations of Maxwell theory motivated
from string theory which use higher degree differential form fields,
for instance the two-form $B$-field of superstring theory, the
three-form $C$-field of M-theory, and other supergravity fields. The
higher Maxwell equations on a lorentzian $n+1$-manifold $M=\IR\times N$, with
$N$ a riemannian $n$-manifold, read again
\beq
\dd F=0 \qquad \mbox{and} \qquad \dd\star F=j_e \ ,
\label{GAGTeom}\eeq
but now generally $F\in\Omega^p(M)$ and
$j_e\in\Omega_c^{n-p+2}(M)$ with $\dd j_e=0$. In the absence of electric
current $j_e=0$, the classical flux group is
\beqa
\big([F]_{\rm dR}\,,\,[\star F]_{\rm dR}\big) \ \in \
\RH^p(M;\IR)\oplus \RH^{n-p+1}(M;\IR) \ ,
\eeqa
which possesses ``electric-magnetic duality'' interchanging
(magnetic) $p$-forms with (electric) $n-p+1$-forms. On the other hand, the group of
classical electric charges is
\beqa
\big[j_e\big|_N\,\big] \ \in \ Q_e:= \ker\big(\RH^{n-p+2}_c(N;\IR) \ \longrightarrow \ \RH^{n-p+2}(N;\IR)
\big) \ .
\eeqa
From the exact sequence
\beqa
\RH^{n-p+1}(M;\IR) & \xrightarrow{ \ i^* \ } &
\RH^{n-p+1}\big(M-{\rm supp}(j_e)\,;\,\IR\big) \ \longrightarrow \\ &
\longrightarrow &
\RH^{n-p+1}\big(M\,,\, M-{\rm supp}(j_e)\,;\,\IR\big) \ \xrightarrow{
  \ \delta \ } \ \RH^{n-p+2}(M;\IR)
\eeqa
one identifies the charge group as
\beqa
Q_e\cong \RH^{n-p+1}\big(M-{\rm supp}(j_e)\,;\,\IR\big)\, \big/\,
\RH^{n-p+1}(M;\IR) \ .
\eeqa
We interpret this as the group of ``charges measured by the flux at
infinity'', which are given by integrating the form $\star F$
over a gaussian sphere $S_\infty^{n-p+1}$.

Again these equations can be obtained from
  a variational principle for the action functional
\beq
S_M[A]:= \int_M\, \Big(-\frac12\, F\wedge\star F+A\wedge j_e\Big) \ ,
\label{GAGTaction}\eeq
where $F=\dd A$. Since the current form $j_e$ is closed and
compactly supported on $M$, by Poincar\'e duality there is a dual class in
the real homology $[W_e]:= \PD_M(j_e)\in \RH_{p-1}(M;\IR)$ which is
represented by a compact oriented submanifold $W_e \subset M$ such that
\beq
\int_M\, a\wedge j_e = -\int_{W_e} \, a\big|_{W_e}
\label{PDint}\eeq
for any closed $p-1$-form $a$; if $a$ is not closed, then the
formula (\ref{PDint}) still holds but $j_e$ must be now regarded as a
de~Rham current, i.e. a distributional form
supported on $W_e \subset M$. We think of the $p_e+1$-dimensional
submanifold $W_e$ as the worldvolume of an
``electrically charged $p_e$-brane'' where $p_e=p-2$, with 
uniform charge density $q_e=\big[\star j_e\big|_{W_e} \big]
\in\RH^0(W_e;\IR)$ induced by the current which the electrically
charged brane produces. Alternatively, given $q_e\in\RH^0(W_e;\IR)$,
the current $[j_e]=i_!(q_e)\in \RH^{n-p+2}(M,M- W_e)$ is the
pushforward of $q_e$ induced by the embedding $i:W_e\hookrightarrow M$.

We can also
introduce a magnetic current $j_m \in\Omega^{p+1}_c(M)$, which modifies
the first equation of motion in (\ref{GAGTeom}) to $\dd F=j_m$, and
the corresponding magnetic $p_m$-brane $W_m\subset M$ with
$p_m=n-p-1$. Then the classical dyonic charge group is
\beq
\big([j_m]\,,\, [j_e]\big) \ \in \ \RH^{p+1}(M,M- W_m;\IR)\oplus
\RH^{n-p+2}(M,M- W_e;\IR)  \ .
\label{classchargegp}\eeq
When $p$ is even the lattice of charges is symplectic,
while for $p$ odd the lattice is symmetric.

Everything we said before concerning the semi-classical Maxwell theory
has an analogue for these higher abelian gauge fields. In particular,
Dirac quantization implies the quantization of classical charges and
the quantum charge group is the real image of the lattice
\beqa
\RH^{p+1}(M,M- W_m;\IZ) \ \oplus \ 
\RH^{n-p+2}(M,M- W_e;\IZ)
\eeqa
in de~Rham cohomology. From this one might expect that the quantum
flux group is the real image of the lattice $\RH^p(M;\IZ)\oplus
\RH^{n-p+1}(M;\IZ)$ in de~Rham cohomology, but we shall see that there
are some subtleties with this naive guess. 
The proper geometric
interpretation of the quantum theory of higher abelian gauge fields, which is a higher generalization of the description of
electromagnetic fields in terms of line bundles with connection, is
provided by studying isomorphism classes of fields in Cheeger--Simons
differential cohomology~\cite{Alexander:1985aa,Brylinski}.

\begin{definition}
The \emph{$p$-th Cheeger--Simons differential cohomology group} of $M$
is the subgroup
\beqa
\check\RH{}^p(M)\subset \Hom_{\scrAb}\big(Z_{p-1}(M)\,,\, \uo\big)
\eeqa
in the category $\scrAb$ of abelian groups consisting of homomorphisms $\chi$, called \emph{differential characters}, such
that there exists a unique closed integral $p$-form
$F_\chi\in 2\pi\, \Omega_\IZ^p(M)$, called the \emph{curvature} of the
differential character $\chi$, with
\beqa
\chi(\partial B) = \exp\Big(\im \int_B\, F_\chi \Big) \qquad
\mbox{for} \quad B\in C_p(M) \ .
\eeqa
\label{CheegerSimonsdef}\end{definition}

In the following we use a multiplicative notation for characters
$\chi$ which are valued in the circle group $\uo$, and an additive
notation for their classes $[\check A]$ which are valued in
the abelian Lie algebra $\bbt=\IR/\IZ$; when we wish to utilize both
descriptions simultaneously we will also write $\chi_{\check A}$. For reasons that will eventually become clear, the field theories
we are interested in all fall into the following
characterization.

\begin{definition}
A \emph{higher abelian gauge theory} is a field theory on a
smooth manifold $M$ whose (semi-classical) configuration space of
gauge inequivalent field configurations is given by the differential
cohomology group $\check\RH{}^p(M)$ for some $p\in\IZ$, and whose
\emph{charge group} is the integer cohomology $\RH^p(M;\IZ)$.
\end{definition}

\subsubsection*{Properties}

\begin{enumerate}
\item The map $M\mapsto\check\RH{}^p(M)$ is a contravariant functor,
  with $\check\RH{}^p(M)$ an infinite-dimensional abelian Lie group
  whose connected components are labelled by the charge group
 $\pi_0\check\RH{}^p(M)= \RH^p(M;\IZ)$.
\item There is a \emph{graded ring structure}
  $\check\RH{}^{p_1}(M)\otimes\check\RH{}^{p_2}(M)\to \check\RH{}^{p_1+p_2}(M)$, denoted $\chi_1\smile \chi_2$
  for $\chi_1,\chi_2\in \check\RH{}^\bullet(M)$, and an \emph{integration map}
  $\int^{\check\RH{}}\!\!\!\int_M \, :\check\RH{}^{n+2}(M)\to
  \check\RH{}^1(\pt)\cong\bbt$ where $\pt$ denotes a one-point space; the existence of this integration
  requires a suitable notion of $\check\RH{}$-orientation on $M$. More
  generally, given an $\check\RH{}$-oriented bundle of manifolds
  $M\hookrightarrow X\to P$, there is an integration over the fibres
  $\int^{\check\RH{}}\!\!\!\int_{X/P} \, :\check\RH{}^s(X)\to \check\RH{}^{s-n-1}(P)$.
\item There is a surjective \emph{field strength map} defined by a
  natural transformation
\beqa
F\,:\, \check\RH{}^p(M) \ \longrightarrow \ \Omega_\IZ^p(M) \ , \qquad
F(\chi):= \mbox{$\frac1{2\pi}$}\, F_\chi
\eeqa
which is a graded ring homomorphism, i.e.
\beqa
F(\chi_1\smile\chi_2)= F_{\chi_1} \wedge F_{\chi_2} \ .
\eeqa
Then integration obeys a version of Stokes' theorem
\beqa
\int^{\check\RH{}}\!\!\!\!\!\int_{\partial N} \, [\check A] =  \,
\int^{\check\RH{}}\!\!\!\!\!\int_N \, F\big([\check A]\big)
\eeqa
for $[\check A] \in\check\RH{}^p(N)$.
\item There is a surjective \emph{characteristic class map} defined by a
  natural transformation
\beqa
c\,:\, \check\RH{}^p(M) \ \longrightarrow \ \RH^p(M;\IZ)
\eeqa
which is a ring homomorphism, i.e.
\beqa
c(\chi_1\smile\chi_2)= c(\chi_1)\smile c(\chi_2) \ ,
\eeqa
and which is compatible with the field strength map, i.e.
\beqa
\big[(c\otimes\IR)(\chi)\big]= \mbox{$\frac1{2\pi}$}\, \big[F_\chi\big]_{\rm dR} \ .
\eeqa
Together the maps $c$ and $F$ define the pullback square
\beqa
\xymatrix{
\check\RH{}^p(M) \ \ar[r] \ar[d] & \ \Omega_{\rm cl}^p(M) \ar[d] \\
\RH^p(M;\IZ) \ \ar[r] & \ \RH^p(M;\IR)
}
\eeqa
which leads to the exact sequence
\beqa
0 \ \longrightarrow \ \RH^{p-1}(M;\IZ)\otimes \IR/\IZ \
\longrightarrow \ \check\RH{}^p(M) \ \longrightarrow \ \Omega_\IZ^p(M)
\times_{[-]} \RH^p(M;\IZ)
\eeqa
where $\Omega_\IZ^p(M)
\times_{[-]} \RH^p(M;\IZ):= \big\{(\omega,\xi) \ \big| \ [\omega]_{\rm
  dR}= \xi \big\}$.
\item Topologically trivial or flat fields $\chi_1=\chi_{A_1}$
  correspond to the class of a globally defined differential form $A_1$ on $M$; its
  product with any other character $\chi_2$ is also topologically
  trivial and given by
\beqa
\chi_{A_1}\smile \chi_2= \chi_{A_1\wedge F_{\chi_2}} \ .
\eeqa
More generally, the products of $\phi\in\ker F\subset \check\RH{}^p(M)$
and $\xi\in\ker c\subset \check\RH{}^p(M)$
with any character $\chi\in\check\RH{}^l(M)$ correspond respectively to the classes $(-1)^l\,
c(\chi)\smile\phi$ and $(-1)^l\,
F(\chi)\smile\xi$.
\item The Cheeger--Simons groups are completely characterized by
  two short exact sequences, which can be summarised in the diagram
\beq
\xymatrix{
0 \ar[dr] & & & & 0 \\
 & \RH^{p-1}(M;\bbt) \ar[dr] & & \RH^p(M;\IZ) \ar[ur] & \\
 & & \check\RH{}^p(M) \ar[ur]^c \ar[dr]^F & & \\
 & \Omega^{p-1}(M) \,\big/ \, \Omega_\IZ^{p-1}(M) \ar[ur] & &
 \Omega^p_\IZ(M) \ar[dr] & \\
0 \ar[ur] & & & & 0
}
\label{CSexactseqs}\eeq
The sequence running from top to bottom is the field strength
sequence, where, by Poincar\'e duality,
$\RH^{p-1}(M;\bbt)\cong\Hom_{\scrAb}( \RH_{p-1}(M;\IZ),\bbt)$ is the group of \emph{flat
fields} $\chi$ with $F_\chi=0$; it defines a torus
$\torus^p(M) \subset\check\RH{}^p(M)$ with fundamental group
\beqa
\pi_1\check\RH{}^p(M)\cong \RH^{p-1}(M;\IZ)\,\big/\, \Tor\,
\RH^{p-1}(M;\IZ)
\eeqa
based at the identity $0$. The sequence from bottom to top is the
characteristic class sequence, with $\Omega^{p-1}(M) /
\Omega_\IZ^{p-1}(M)$ the torus of \emph{topologically trivial fields} whose
classes $[A]$ have curvature $F([A])=\dd A$. One also has the exact
sequence
\beq
0 \ \longrightarrow \ \RH^{p-1}(M;\IR)\,\big/\, \RH^{p-1}(M;\IZ) \
\longrightarrow \ \RH^{p-1}(M;\torus) \ \xrightarrow{ \ \beta \ } \ \Tor\,
\RH^p(M;\IZ) \ \longrightarrow \ 0 \ ,
\label{CSflatseq}\eeq
where the torsion subgroup is the group of \emph{discrete Wilson
  lines} and $\beta$ is the Bockstein homomorphism.
This
yields a geometric picture of $\check\RH{}^p(M)$ as consisting of
infinitely many connected components $\check\RH{}^p_c(M)$ labelled by the charges $c\in\RH^p(M;\IZ)$, with each
topological sector a torus fibration over a vector space whose fibres are
finite-dimensional tori $\Omega^{p-1}_{\rm cl}(M) /
\Omega_\IZ^{p-1}(M)$ represented by topologically trivial flat
fields, called \emph{Wilson lines}. In particular, there is a non-canonical splitting
$\check\RH{}^p(M)= \bigsqcup_{c\in\RH^p(M;\IZ)}\,
\check\RH{}^p_c(M)\cong T\times \Gamma\times V$ where $T$ is the torus of
Wilson lines, $\Gamma$ is the subgroup of topologically trivial flat
fields, and $V\cong{\rm im}(\dd^\dag)$ is the vector space of
\emph{oscillator modes}; there is an isomorphism
$\check\RH{}^p_0(M)\, /\, \torus^p(M)\cong \dd\Omega^{p-1}(M)$ of vector
spaces. Crucially, in contrast to ordinary cohomology groups, the
differential cohomology contains information about
both flat and topologically trivial fields on $M$, as they generally
have non-zero classes in $\check\RH{}^p(M)$.
\item The Cheeger--Simons groups satisfy \emph{Pontrjagin--Poincar\'e
    duality}
\beqa
\Hom_{\scrAb}\big(\check\RH{}^p(M)\,,\, \bbt\big) \cong \check\RH{}^{n+2-p}(M) \ .
\eeqa
This duality is a consequence of the fact that integration defines a
perfect bilinear pairing
\beqa
\check\RH{}^p(M)\times \check\RH{}^{n+2-p}(M) \ \longrightarrow \ 
\check\RH{}^1(\pt)= \bbt
\eeqa
by
\beqa
\langle \chi_1,\chi_2\rangle:= \int^{\check\RH{}}\!\!\!\!\!\int_M \, \chi_1\smile \chi_2
\ .
\eeqa
On cocycles we also denote this pairing by $([\check A_1],[\check
A_2])\mapsto \int^{\check\RH{}}\!\!\!\int_M\, [\check A_1]\smile[\check
A_2] =: \langle[\check A_1],[\check A_2]\rangle$. If the characteristic class of $\check A_1$ is zero then the
pairing is $\int_M\, A_1\wedge F_{\check A_2} \ {\rm mod}\ \IZ$, and if also
$c(\check A_2)=0$ then this becomes $\int_M\, A_1\wedge\dd A_2 \ {\rm
  mod}\ \IZ$; note that both of these pairings are given by integrals
of forms. On the other hand, if the curvature $F_{\check A_1}=0$ with
$[\check A_1]=\alpha_1\in\RH^{p-1}(M;\torus)$, then the pairing is
$\int^{\RH{}}\!\!\!\!\int_M\, \alpha_1\smile c({\check A_2}) \ {\rm mod}\
\IZ$.
\item A differential
character $\chi\in \check
\RH{}^p(M)$ defines a \emph{holonomy}
\beq
\hol_\Sigma(\chi):= \exp\Big(\ii\oint_\Sigma\, A_\chi
\Big) \ \in \ \uo
\label{charhol}\eeq
for any $p-1$-cycle $\Sigma\in
Z_{p-1}(M)$, where the potential $A_\chi\in\Omega^{p-1}(\Sigma)$ is defined by
$F_{\chi|_\Sigma}=\dd A_\chi$ and we have used $\RH^p(\Sigma;\IZ)=0$. For
flat fields $F_\chi=0$, the holonomy defines a class
$\big[\hol(\chi)\big] \in \RH^{p-1}(M;\uo)$.
\end{enumerate}

\subsubsection*{Examples}

The groups $\check\RH{}^p(M)$ vanish for all $p<0$. For the first few
non-vanishing groups we have the following identifications:
\begin{itemize}
\item For $p=0$, $\check\RH{}^0(M)\cong \RH^0(M;\IZ)$ is identified
  via the characteristic class map $c$ as the group of
  connected components of $M$; the field strength map $F$ assigns an
  integer to each component.
\item For $p=1$, $\check\RH{}^1(M)\cong\Omega^0(M;\uo)$ is the space
  of differentiable circle-valued maps $f:M\to\uo\cong S^1$. The field strength and 
  characteristic class maps are given in this case
  by
\beqa
f \
\stackrel{F}{\longmapsto} \ \mbox{$\frac1{2\pi\im}$}\, \dd \log f \qquad \mbox{and} \qquad f \ \stackrel{c}{\longmapsto} \ f^*[\dd\theta]  \ ,
\eeqa
where $[\dd\theta]$ is the fundamental class of $S^1$; these maps
describe how the function $f$ acts on cohomology. The holonomy is
the evaluation $\hol_x(f)=f(x)$ of $f$ at $x\in M$.
\item For $p=2$, $\check\RH{}^2(M)\cong \Pic_\nabla(M)$ is the Picard group of gauge
  equivalence classes of line bundles with connection $(L,\nabla)$ on
  $M$ and gauge group generated by $\check\RH{}^1(M)$. The field strength and characteristic class
  maps are given in this case by
\beqa
(L,\nabla) \ \stackrel{F}{\longmapsto} \ \mbox{$\frac1{2\pi\im}$}\, \nabla^2 \qquad \mbox{and}
\qquad (L,\nabla) \ \stackrel{c}{\longmapsto} \ c_1(L) \ .
\eeqa
The connection $\nabla$ determines a holonomy $\hol_\gamma(\nabla)$ for $\gamma\in Z_1(M)$
which coincides with
the holonomy (\ref{charhol}) of the corresponding differential
character.
\item For $p=3$, $\check\RH{}^3(M)$ is isomorphic to the group of gauge
  equivalence classes of $\uo$ gerbes $\scrG\downarrow M$ with
  2-connection $(A,B)$ and gauge group generated by the differential
  cohomology $\check\RH{}^2(M)$. The
  field strength map gives the curvature $H=\dd B$ of the $B$-field,
  while the characteristic class map returns the Dixmier--Douady class of
  $\scrG$. The 2-connection $(A,B)$ determines a holonomy
  $\hol_\Sigma(B) \in\uo$ for $\Sigma\in Z_2(M)$.
\item When $M$ is a point, one can use the exact sequences in
  (\ref{CSexactseqs}) to explicitly compute
\beqa
\check\RH{}^p(\pt)= \left\{ \begin{matrix} \IZ \ , & \quad p=0 \ , \\
    \IR/\IZ \ , & \quad p=1 \ , \\ 0 \ , & \quad p>1 \ ,
  \end{matrix} \right.
\eeqa
where for $p=0$ only the characteristic class contributes while for
$p=1$ only topologically trivial flat fields contribute. From the
field strength exact sequence in (\ref{CSexactseqs}) one also finds
\beqa
\check\RH{}^{n+2}(M)\cong \RH^{n+1}(M;\bbt) \qquad \mbox{and} \qquad 
\check\RH{}^p(M)=0 \quad \mbox{for} \quad p>n+2 \ .
\eeqa
\end{itemize}

\subsubsection*{Deligne cohomology}

An explicit cochain model for the Cheeger--Simons groups is provided
by Deligne cohomology, see e.g.~\cite{Brylinski}, which makes explicit
the
previous properties and examples together with their higher generalizations. The
degree $p$ smooth Deligne cohomology is the $p$-th \v{C}ech
hypercohomology of the truncated sheaf complex
\beqa
0\ \longrightarrow \ \uo_M \ \xrightarrow{\dd\log} \ \Omega_M^1 \
\xrightarrow{ \ \dd \ } \ \Omega_M^2 \ 
\xrightarrow{ \ \dd \ } \ \cdots \ \xrightarrow{ \ \dd \ } \
\Omega_M^p \ ,
\eeqa
where $\uo_M$ is the sheaf of smooth $\uo$-valued functions on
$M$ and $\Omega_M^p$ is the sheaf of differential $p$-forms on
$M$. The degree $p$ Deligne cohomology group $\RH_{\mathfrak{D}}^p(M)$
can be calculated as the cohomology of the total complex of the double
complex with respect to a good open cover
$\fM=\{M_\alpha\}_{\alpha\in I}$ of $M$ given by
\beqa
\xymatrix{
\vdots & \vdots & & \vdots \\
C^2\big(\fM\,;\,\uo_M\big) \ \ar[u]^\delta \ \ar[r]^{\ \dd\log} & \ 
C^2\big(\fM\,;\, \Omega_M^1\big) \ \ar[u]^\delta \ar[r]^{ \ \  \ \  \ \dd }  & \ 
\cdots \ \ar[r]^{\!\!\!\!\!\!\!\!\!\!\!\!\dd } & \ 
C^2\big(\fM\,;\, \Omega_M^p\big) \ \ar[u]^\delta \\ 
C^1\big(\fM\,;\,\uo_M\big) \ \ar[u]^\delta \ \ar[r]^{\ \dd\log} & \ 
C^1\big(\fM\,;\, \Omega_M^1\big) \ \ar[u]^\delta \ar[r]^{  \ \  \ \ \ \dd } & \ 
\cdots \ \ar[r]^{\!\!\!\!\!\!\!\!\!\!\!\!\dd } & \ 
C^1\big(\fM\,;\, \Omega_M^p\big) \ \ar[u]^\delta \\ 
C^0\big(\fM\,;\,\uo_M\big) \ \ar[u]^\delta \ \ar[r]^{\ \dd\log} & \ 
C^0\big(\fM\,;\, \Omega_M^1\big) \ \ar[u]^\delta \ar[r]^{  \ \  \ \ \ \dd } & \ 
\cdots \ \ar[r]^{\!\!\!\!\!\!\!\!\!\!\!\!\dd } & \ 
C^0\big(\fM\,;\, \Omega_M^p\big) \ \ar[u]^\delta
}
\eeqa
which is the quotient of the abelian group of Deligne
$p$-cocycles by the subgroup of Deligne $p$-coboundaries; here
$\delta$ is the usual \v{C}ech coboundary operator, and
$C^p(\fM;\uo_M)$ and $C^p(\fM;\Omega^k_M)$ denote the
\v{C}ech $p$-cochains. Below we describe these groups explicitly in
low degree; contractible $k$-fold intersections of open sets of the cover
$\fM$ are
denoted by
$M_{\alpha_1\dots\alpha_k}:= M_{\alpha_1}\cap\cdots\cap
M_{\alpha_k}$.

A degree~$0$ smooth Deligne class in $C^0(\fM;\uo_M)$ is just a
smooth map $g:M\to\uo$. This case applies to the sigma-model on $M$
of a scalar field $g$ compactified on a circle; the charge group is the
group of winding numbers $c\in\RH^1(M;\IZ)$ of the field $g$ around
$S^1$. In
  particular, $\check\RH{}^1(S^1)$ is the loop group $L\uo$ and if
  $\sigma\sim\sigma+1$ is the coordinate on the circle $M=S^1$, then
  Fourier series expansion gives the explicit decomposition
\beqa
\log g(\sigma)= 2\pi\ii g_0+2\pi\ii c\, \sigma+ \sum_{k\neq0}\,
\frac{g_k}k \, \e^{2\pi\ii k\, \sigma} \ \in \ \bbt \oplus \Gamma_c \oplus V
\ .
\eeqa

A Deligne one-cocycle is a pair
\beqa
(g_{\alpha\beta},A_\alpha) \ \in \
C^1\big(\fM\,;\,\uo_M\big)\oplus C^0\big(\fM\,;\,
\Omega_M^1\big)
\eeqa
satisfying the cocycle conditions
\beqa
g_{\alpha\beta}\, g_{\beta\gamma}\, g_{\gamma\alpha}&=& 1 \qquad
\mbox{on} \quad M_{\alpha\beta\gamma} \ , \nonumber \\[4pt]
A_\alpha-A_\beta &=& \dd \log g_{\alpha\beta} \qquad \mbox{on}
\quad M_{\alpha\beta} \ .
\eeqa
A Deligne one-coboundary defines a gauge transformation and is a pair of the type
\beqa
\big(h \,,\,
\dd\log h\big)
\eeqa
for a
smooth map $h:M\to\uo$. The $\uo$ \v{C}ech
one-cocycle $g_{\alpha\beta}:M_{\alpha\beta}\to
\uo$ determines smooth transition functions on overlaps for a hermitian
line bundle $L\to M$. This cocycle represents the first Chern class $c_1(L)=
[g_{\alpha\beta}]\in \RH^1(M;\uo)\cong \RH^2(M;\IZ)$, where
the canonical isomorphism follows from the exponential sequence
\beqa
0 \ \longrightarrow \ \IZ \ \longrightarrow \ \IR \ \xrightarrow{\exp}
\ \uo \ \longrightarrow \ 1 \ ;
\eeqa
it is 
the obstruction to triviality of the line bundle $L\to M$. The local one-forms
$A_\alpha\in\Omega^1(M_\alpha)$ define a unitary connection
$\nabla=\dd+A$ on $L$. This is the case that arose in Maxwell theory.

A Deligne two-cocycle is a triple
\beqa
(g_{\alpha\beta\gamma},A_{\alpha\beta},B_\alpha) \ \in \
C^2\big(\fM\,;\, \uo_M\big)\oplus
C^1\big(\fM\,;\,\Omega_M^1\big)\oplus
C^0\big(\fM\,;\,\Omega^2_M\big)
\eeqa
satisfying the cocycle conditions
\bea
g_{\alpha\beta\gamma}\,
g^{-1}_{\beta\gamma\delta}\, g_{\gamma\delta\alpha}\,
g^{-1}_{\delta\alpha\beta} &=& 1 \qquad \mbox{on} \quad
M_{\alpha\beta\gamma\delta} \  , \nonumber \\[4pt]
A_{\alpha\beta}+A_{\beta\gamma}+ A_{\gamma\alpha} &=& 
\dd \log g_{\alpha\beta\gamma} \qquad \mbox{on} \quad
M_{\alpha\beta\gamma} \ , \nonumber \\[4pt]
B_\alpha-B_\beta &=& \dd A_{\alpha\beta} \qquad \mbox{on} \quad
M_{\alpha\beta} \ .
\label{2cocycleconds}\eea
A Deligne two-coboundary defines a gauge transformation and is a triple of the type
\beqa
\big(h_{\alpha\beta}\, h_{\beta\gamma}\, h_{\gamma\alpha}\,,\,
\dd\log
h_{\alpha\beta}+a_\alpha-a_\beta\,,\, \dd a_\alpha\big)
\eeqa
for $(h_{\alpha\beta},a_\alpha)\in C^1(\fM;\uo_M) \oplus C^0(\fM;
\Omega_M^1)$.
The $\uo$ \v{C}ech
two-cocycle $g_{\alpha\beta\gamma}:M_{\alpha\beta\gamma}\to \uo$ specifies a hermitian
``transition'' line bundle $L_{\alpha\beta}$ over each overlap
$M_{\alpha\beta}$, an isomorphism
$L_{\alpha\beta}\cong L^*_{\beta\alpha}$, and a trivialization
of the line bundle
$L_{\alpha\beta}\otimes L_{\beta\gamma}\otimes L_{\gamma\alpha}$ on
each triple overlap $M_{\alpha\beta\gamma}$; the pair
$\scrG=(L_{\alpha\beta},g_{\alpha\beta\gamma})$ defines a {gerbe}
on $M$. This cocycle
represents the Dixmier--Douady class $dd(\scrG)= [g_{\alpha\beta\gamma}]\in
\RH^2(M;\uo)= \RH^3(M;\IZ)$; it is the obstruction to triviality of
the gerbe $\scrG\downarrow M$.
The \v{C}ech one-cochain $A_{\alpha\beta}$ defines connection one-forms on each line bundle
$L_{\alpha\beta}\to M_{\alpha\beta}$ such that the section
$g_{\alpha\beta\gamma}$ is covariantly constant with respect to the
induced connection on $L_{\alpha\beta}\otimes L_{\beta\gamma}\otimes
L_{\gamma\alpha}$; it defines a {0-connection} (or
{connective structure}) on the gerbe $\scrG$. The collection of two-forms
$B_\alpha\in\Omega^2(M_\alpha)$ defines a {1-connection} (or
{curving}) on $\scrG$. The pair $(A_{\alpha\beta},B_\alpha)$ defines a
{$2$-connection} on the gerbe $\scrG=(L_{\alpha\beta},g_{\alpha\beta\gamma})$.
The {gauge group} of the gerbe is generated by line bundles
$\ell \to M$ with connection $\nabla=\dd+a$ and curvature
$f_\nabla=\dd a$ with $F([\ell,\nabla]) = \frac1{2\pi
  \ii}\,f_\nabla\in\Omega_\IZ^2(M)$ through the gauge transformations
\beqa
L_{\alpha\beta} \ \longmapsto \ L_{\alpha\beta}\otimes \ell\big|_{M_{\alpha\beta}} \ ,
  \qquad A_{\alpha\beta} \ \longmapsto \ A_{\alpha\beta}+
  a\big|_{M_{\alpha\beta}} \qquad \mbox{and} \qquad B_\alpha \
  \longmapsto \ B_\alpha+f_\nabla \ .
\eeqa
The Deligne cohomology $\RH^2_{\fD}(M)$ applies to the
higher abelian gauge theory of the
Kalb--Ramond $B$-field of superstring theory, with $M$ a
ten-dimensional manifold. See~\cite{Mathai:2005sw} for a description
of the instanton moduli space in the bundle gerbe version of
generalized Maxwell theory.

For $p=3$, a Deligne class is represented by a quadruple
\beqa
(g_{\alpha\beta\gamma\delta},A_{\alpha\beta\gamma},B_{\alpha\beta}, C_\alpha) \ \in \
C^3\big(\fM\,;\, \uo_M\big)\oplus
C^2\big(\fM\,;\,\Omega_M^1\big)\oplus
C^1\big(\fM\,;\,\Omega^2_M\big)\oplus
C^0\big(\fM\,;\,\Omega^3_M\big)
\eeqa
satisfying the equations
\bea
g_{\alpha\beta\gamma\delta}\, g_{\alpha\beta\delta\kappa}\, 
g_{\beta\gamma\delta\kappa}\, g^{-1}_{\alpha\gamma\delta\kappa}\, g^{-1}_{\alpha\beta\gamma\kappa}
 &=& 1 \qquad \mbox{on} \quad
M_{\alpha\beta\gamma\delta\kappa} \  , \nonumber \\[4pt]
A_{\alpha\beta\gamma}+ A_{\alpha\gamma\delta}-A_{\alpha\beta\delta}-A_{\beta\gamma\delta} &=& 
\dd \log g_{\alpha\beta\gamma\delta} \qquad \mbox{on} \quad
M_{\alpha\beta\gamma\delta} \ , \nonumber \\[4pt]
B_{\alpha\beta}+B_{\beta\gamma}+ B_{\gamma\alpha} &=& \dd A_{\alpha\beta\gamma} \qquad \mbox{on} \quad
M_{\alpha\beta\gamma} \ , \nonumber \\[4pt]
C_\alpha-C_\beta&=& \dd B_{\alpha\beta} \qquad \mbox{on} \quad
M_{\alpha\beta} \ . \nonumber 
\eea
Gauge transformations are generated by the Deligne three-coboundaries
\beqa
\big(h_{\alpha\beta\gamma}\, h^{-1}_{\alpha\gamma\delta}\,
h_{\beta\gamma\delta}\, h^{-1}_{\alpha\beta\delta}\,,\, \dd \log
h_{\alpha\beta\gamma} +a_{\alpha\beta}+a_{\beta\gamma}+
a_{\gamma\alpha}\,,\, \dd a_{\alpha\beta}+b_\alpha-b_\beta\,,\, \dd
b_\alpha\big)
\eeqa
for $(h_{\alpha\beta\gamma},a_{\alpha\beta},b_\alpha)\in C^2(\fM;\uo_M)\oplus
C^1(\fM;\Omega_M^1)\oplus
C^0(\fM;\Omega^2_M)$. This cocycle represents a 2-gerbe
$\underline{\scrG}=(\scrG_{\alpha\beta\gamma},g_{\alpha\beta\gamma\delta})$ with
3-connection
$(A_{\alpha\beta\gamma},B_{\alpha\beta},C_\alpha)$~\cite{Stevenson-thesis,Johnson,Brylinski2}.
The Deligne cohomology $\RH^3_{\fD}(M)$ is the one
appropriate to the abelian gauge theory of the three-form $C$-field of
M-theory, with $M$ an 11-dimensional manifold.

\subsubsection*{Holonomy and curvature}

The construction of holonomy and curvature of a Deligne class defines
an isomorphism~\cite{Brylinski} $$ \RH_{\mathfrak{D}}^{p}(M) \ \xrightarrow{ \
  \approx \ } \ \check \RH{}^{p+1} (M) \ . $$

Given a degree~$1$ smooth Deligne class $[(g_{\alpha\beta},A_\alpha)]$
represented by a hermitian line bundle $L\to M$ with unitary
connection $\nabla=\dd+A$, the
curvature is the globally defined two-form given by $F_\nabla=\dd
A_\alpha$ on $M_\alpha$ with $F([L,\nabla]) = \frac1{2\pi
  \ii}\,F_\nabla\in\Omega_\IZ^2(M)$.
By Stokes' theorem, the holonomy of $\nabla$ around any one-cycle
$\gamma\subset M$ is then obtained from the product formula~\cite{Kiyonori:2001aa}
\beqa
\hol_\gamma(A) =\prod_{\alpha\in I} \, \exp\Big(\ii
\int_{\gamma_\alpha}\, A_\alpha\Big) \ \prod_{\alpha,\beta\in I}\,
g_{\alpha\beta}(\gamma_{\alpha\beta}) \ ,
\eeqa
where $\gamma_\alpha\subset M_\alpha$ is a path in a subdivision of
the loop $\gamma$ into segments and $\gamma_{\alpha\beta}=\gamma_\alpha\cap
\gamma_\beta$ is a point in $M_{\alpha\beta}$. This
definition agrees with the general definition of holonomy
(\ref{charhol}) in terms of
differential characters.

Given a degree~$2$ smooth Deligne class
$[(g_{\alpha\beta\gamma},A_{\alpha\beta},B_\alpha)]$ represented by a
gerbe $\scrG\downarrow M$ with 2-connection $(A,B)$, the curvature
of the corresponding differential character is the globally defined
closed three-form $H=H_{(A,B)}$ given by $H=\dd B_\alpha$ on
$M_\alpha$ with $F([\scrG,A,B]) = \frac1{2\pi\ii}\, H\in\Omega_\IZ^3(M)$,
while its characteristic class is the
Dixmier--Douady class $dd(\scrG) \in \RH^3(M;\IZ)$ of the gerbe $\scrG$.
Its holonomy around a two-cycle $\Sigma\subset M$ is obtained
by choosing a triangulation $\{\Sigma_\alpha\}_{\alpha\in I}$ of $\Sigma$ subordinate to the open cover
$\Sigma\cap\fM$. Keeping careful track of orientations, by repeated
application of Stokes' theorem one arrives at the product
formula~\cite{Brylinski,Kiyonori:2001aa}
\beqa
\hol_\Sigma(B)=\prod_{\alpha\in I}\, \exp\Big(\ii\int_{\Sigma_\alpha}\,
B_\alpha\Big) \ \prod_{\alpha,\beta\in I}\,
\exp\Big(\ii\int_{\Sigma_{\alpha\beta}}\, A_{\alpha\beta}\Big) \
\prod_{\alpha,\beta,\gamma\in I}\,
g_{\alpha\beta\gamma}(\Sigma_{\alpha\beta\gamma})
\eeqa
where $\Sigma_{\alpha\beta}$ is the common boundary edge of the surfaces $\Sigma_\alpha$ and
$\Sigma_\beta$, and $\Sigma_{\alpha\beta\gamma}=\Sigma_{\alpha\beta}\cap
\Sigma_{\beta\gamma}\cap\Sigma_{\gamma\alpha}$ are vertices of the
triangulation of $\Sigma$. The coincidence between this expression and the
general formula in terms of differential characters (\ref{charhol}) is shown
explicitly in~\cite{Carey:2002xp}. This construction first appeared in
the context of the Wess--Zumino--Witten model in~\cite{Gawedzki:1987ak}.

For $p=3$, the smooth Deligne class
$[(g_{\alpha\beta\gamma\delta},A_{\alpha\beta\gamma},B_{\alpha\beta},
C_\alpha)]$ of a 2-gerbe $\underline{\scrG}\downarrow\downarrow M$
with 3-connection $(A,B,C)$ has curvature $G$ given by $G=\dd C_\alpha$ on
$M_\alpha$ with $F([\,\underline{\scrG}\,,A,B,C]) = \frac1{2\pi\ii}\, G
\in\Omega_\IZ^4(M)$, and a degree~$4$ characteristic class
$[g_{\alpha\beta\gamma\delta}]\in \RH^4(M;\IZ)$. For a smooth three-cycle
$\Sigma\subset M$, choose a triangulation subordinate to the open
cover $\fM$ used to define the Deligne class; it consists of tetrahedra
$\Sigma_\alpha$, faces $\Sigma_{\alpha\beta}$, edges
$\Sigma_{\alpha\beta\gamma}$, and vertices
$\Sigma_{\alpha\beta\gamma\delta}$. Then the holonomy around $\Sigma$ is given
by~\cite{Johnson}
\beqa
\hol_\Sigma(C)&=&\prod_{\alpha\in I}\, \exp\Big(\ii\int_{\Sigma_\alpha}\,
C_\alpha\Big) \ \prod_{\alpha,\beta\in I}\,
\exp\Big(\ii\int_{\Sigma_{\alpha\beta}}\, B_{\alpha\beta}\Big) \
\prod_{\alpha,\beta,\gamma\in I}\,
\exp\Big(\ii\int_{\Sigma_{\alpha\beta}\gamma}\,
A_{\alpha\beta\gamma}\Big) \\ && \times \ \prod_{\alpha,\beta,\gamma, 
  \delta \in I}\,
g_{\alpha\beta\gamma\delta}(\Sigma_{\alpha\beta\gamma\delta}) \ .
\eeqa

A general Deligne $p$-cocycle is represented by
\beqa
\big(g_{\alpha_1\dots\alpha_{p+1}} , A_{\alpha_1\dots\alpha_p}^1,\dots,
  A^p_\alpha\big) \ \in \ C^p\big(\fM\,;\, \uo_M\big) \ \oplus \
  \bigoplus_{k=1}^p\, 
  C^{p-k}\big( \fM\,;\, \Omega_M^k\big) \ .
\eeqa
To compute its holonomy, triangulate a smooth $p$-cycle $\Sigma\subset
M$ by a $p$-dimensional simplicial complex $\mathfrak{S}(\Sigma)$; the
$k$-simplices of $\mathfrak{S}(\Sigma)$ are denoted $\sigma^k$ for
$k=0,1,\dots,p$. Let $\rho:\mathfrak{S}(\Sigma)\to I $ be the
index map for the triangulation~\cite{Kiyonori:2001aa}. Then the
holonomy around $\Sigma$ is given by~\cite{Johnson}
\beqa
\hol_\Sigma(A^p)= \prod_{k=1}^p \
\prod_{\underline{\sigma}^k\in\mathfrak{S}(\Sigma)} \,
\exp\Big(\ii\int_{\sigma^k}\, A^k_{\rho(\sigma^p)\dots\rho(\sigma^k)}
\Big) \ \prod_{\underline{\sigma}^0 \in\mathfrak{S}(\Sigma)} \,
g_{\rho(\sigma^p)\dots\rho(\sigma^1) \rho(\sigma^0)}(\sigma^0) \ ,
\eeqa
where the products are taken over \emph{flags} of simplices
$\underline{\sigma}^k:= \big\{(\sigma^k,\sigma^{k+1},\dots,\sigma^p) \
\big| \ \sigma^k\subset\sigma^{k+1}\subset\cdots \subset
\sigma^p\big\}$ for $k=0,1,\dots,p$.

\subsection{Configuration space and gauge transformations\label{Confspace}}

As in Maxwell theory (or more generally in Yang--Mills theory),
locality forces us to work with \emph{gauge potentials} $A$, rather
than with isomorphism classes of gauge fields $F$. The most convenient
mathematical framework for dealing with local quantum field theory is
through categorification, wherein we work directly at the level of
cochain complexes; this enables one to build a quantum field theory by
``gluing'' elementary constituents together. In particular, following the treatment of the
configuration space of semi-classical Maxwell theory, we seek a suitable groupoid of
higher abelian gauge fields which plays the role of the
semi-classical configuration space. This framework also leads to a quantum
definition of charge.

\subsubsection*{Quantum groupoid of fields}

We seek a groupoid
$\check\scrH{}^p(M)$ whose objects are gauge potentials $\check A$,
whose morphisms are ``gauge transformations'', and whose set of
isomorphism classes are gauge equivalence classes, so that
\beqa
\pi_0{\check\scrH{}^p(M)} \cong\check\RH{}^p(M) \ .
\eeqa
In
particular, every object has a group of automorphisms $\RH^{p-2}(M;\bbt)$. The particular
sort of groupoid that we want is an example of an \emph{action
  groupoid}, obtained by the
action of a gauge group on a set of objects: Given an action $G\times X\to X$ of a group $G$ on a set
$X$, there is a groupoid $G\rightrightarrows X$ whose objects are the points
$x\in X$ and whose morphisms are the group actions $x\xrightarrow{ \ g
  \ } x'=g\triangleright x$ for $g\in G$ and $x\in X$. When $X$ is
e.g. a smooth manifold, this groupoid
is sometimes denoted $[X/G]$ and called a \emph{quotient stack}; it is
naturally related to the orbifold $X/G$. 
The desired
properties of the category $\check\scrH{}^p(M)$ can be described as
follows.

Firstly, the gauge transformation law on corresponding
characters $\chi_{\check A}$ is given on $p$-chains $\Sigma$ with
$\partial\Sigma\neq0$ as
\beqa
\chi_{\check A}(\Sigma) \ \longmapsto \ \check\chi(\partial\Sigma)\,
\chi_{\check A}(\Sigma) \qquad \mbox{with} \quad
\check\chi(\partial\Sigma) = \exp\Big(\im \int_\Sigma\, F_{\check\chi}
\Big) \ .
\eeqa
Hence gauge transformations are given by differential characters
$\check\chi\in\check\RH{}^{p-1}(M)$; the action of $[\check
C] \in\check\RH{}^{p-1}(M)$ on a gauge potential $\check A$ is given by
\beq
g_{[\check C]}\triangleright \check A = \check A+ F\big([\check C] \big) \ .
\label{gcheckCcheckA}\eeq
The groupoid of morphisms is thus
the gauge group $\check\RH{}^{p-1}(M)$.

Secondly, the connected components $\check Z_c^p(M)$ of the space of
objects $\check Z^p(M)$ are labelled
by the charge group $c\in \RH^p(M;\IZ)$, each taken as a torsor for $\Omega^{p-1}(M)$. From
(\ref{gcheckCcheckA}) it follows that the flat characters in
$\RH^{p-2}(M;\bbt)$ act trivially on the space of gauge fields
$\check\scrH{}^p(M)$, and hence the group of automorphisms of any object
$\check A$ is
\beqa
\Aut_{\check\scrH{}^p(M)}(\check A)=\RH^{p-2}(M;\bbt) \ .
\eeqa
On $M=\IR\times N$, an automorphism $\alpha\in\RH^{p-2}(N;\bbt)$ acts
on ``wavefunctionals'' $\psi(\check A)$ in the quantum Hilbert space
$\hil(N)$ of the gauge theory via elements of the quantum electric
charge group $Q\in\RH^{n-p+2}(N;\IZ)$ as $\alpha\triangleright\psi(\check
A)=\e^{2\pi\ii\langle\alpha,Q\rangle}\, \psi(\check A)$; the precise
definition of the Hilbert space $\hil(N)$ is the subject of \S\ref{QuantHAGT} In each topological sector
$c\in\RH^p(M;\IZ)$ we choose a field strength $F_c$. Then an arbitrary
field strength in this sector can be written as $F=F_c+\dd a$ for
$a\in\Omega^{p-1}(M)$; these are the oscillator modes. Gauge transformations act through $a\mapsto
a+\omega$ with $\omega\in\Omega_\IZ^{p-1}(M)$; if
$\omega=\dd\varepsilon$ is exact then
this is called a \emph{small gauge transformation} and there are also
gauge transformations of the gauge transformations given by
$\varepsilon\mapsto\varepsilon + \eta$ with
$\eta\in\Omega^{p-2}_\IZ(M)$.

The requisite quantum groupoid then has connected components given by the action groupoids $\check\scrH{}_c^p(M)=
\big[\check Z{}^p_c(M)\, \big/\, \check\RH{}^p(M)\big]$.

\subsubsection*{Category of differential cocycles}

An explicit model for the category $\check\scrH{}^p(M)$ is described in~\cite{Hopkins:2002rd}. In this formulation the space of
objects $\check Z{}^p(M)$ of
$\check\scrH{}^p(M)$ are \emph{cocycles}; these are triples
\beqa
(c,h,\omega) \ \in \ \check C{}^p(M):= C^p(M;\IZ)\times C^{p-1}(M;\IR) \times
\Omega^p(M)
\eeqa
satisfying the cocycle condition
\beqa
\dd_{\check\RH}(c,h,\omega) := (\delta c,\omega-c-\delta h, \dd\omega) = (0,0,0)
\eeqa
with $\dd_{\check\RH}^2=0$. Here we think of $c$ as the characteristic class
$c(\chi)$ of
a differential character~$\chi$, $\omega$ as its field
strength $F_\chi$, and $h$ as the ``monodromy'' $\log\chi$. Note that the
connected components of the space of objects of the category are indeed
labelled by the charge group $\RH^p(M;\IZ)$. Two objects
are connected
by a morphism $(c_1,h_1,\omega_1)\longrightarrow (c_2,h_2,\omega_2)$ if and only if
\beqa
(c_1,h_1,\omega_1)=(c_2,h_2,\omega_2)+\dd_{\check\RH} (b,k,0)
\eeqa
for some $(b,k)\in C^{p-1}(M;\IZ)\times C^{p-2}(M;\IR)$ subject to the
equivalence relation
\beqa
(b,k,0)\sim (b,k,0)-\dd_{\check\RH} (b',k',0) \qquad \mbox{for} \quad (b',k'\,)
\in C^{p-2}(M;\IZ)\times C^{p-3}(M;\IR) \ .
\eeqa

Then the set of isomorphism classes of objects in the category
$\check\scrH{}^p(M)$ is isomorphic to the Cheeger--Simons differential cohomology
group $\check\RH{}^p(M)$. Moreover, the automorphism group of the trivial
object is given by
\beqa
\Aut_{\check\scrH{}^p(M)}(0,0,0)\cong \RH^{p-2}(M;\bbt) \ ,
\eeqa
and whence the flat characters in $\RH^{p-2}(M;\bbt)$ act trivially on
the configuration space of abelian gauge fields
$\check\scrH{}^p(M)$. However, in this model there is no groupoid
decomposition $\check Z{}^p(M)=\bigsqcup_{c\in\RH^p(M;\IZ)}\, \check
Z_c^p(M)$ into topological sectors.

For $p=0$, $\check\scrH{}^0(M)$ is the category of maps $M\to\IZ$ and
identity morphisms.

For $p=1$, $\check\scrH{}^1(M)$ is the category of
smooth maps $M\to\uo$ and identity morphisms; to each object
$(c,h,\omega)\in C^1(M;\IZ)\times C^0(M;\IR)\times \Omega^1(M)$ we
associate its differential character.

For $p=2$, $\check\scrH{}^2(M)$ is the monoidal category of
$\uo$-bundles with connection, with groupoid structure of tensor
product. An object $\check A=(c,h,\omega)\in \check\scrH{}^2(M)$ determines a
$\uo$-bundle with connection in the following way: For each open set
$U\subset M$, a principal homogeneous space $\Gamma(U)$ for the group of
isomorphism classes in the group $\check\scrH{}^1(U)$ of smooth maps
$U\to\uo$ is given by the set of trivialising ``scalar potentials''
$\check\sigma= (b,k,\theta)\in C^1(U;\IZ)\times C^0(U;\IR)\times
\Omega^1(U)$ with $\dd_{\check\RH}\check\sigma=\check A$ modulo the
equivalence relation $\check\sigma\sim
\check\sigma+\dd_{\check\RH}\check\tau$ for $\check\tau\in
C^1(U;\IZ)\times \{0\} \times
\Omega^0(U)$; any two potentials in $\Gamma(U)$ differ by an object of
the category $\check\scrH{}^1(M)$. The connection
$\nabla:\Gamma(U)\to\Omega^1(U)$ then sends $\check\sigma\mapsto\theta$.

\subsubsection*{Quantum 2-groupoid of fields}

The model of differential cocycles illustrates that the construction of the category $\check\scrH{}^p(M)$ may be iterated to give higher categories as well,
by the well-known process of categorification, which works for any
cochain complex: We replace the equivalence relation on morphisms
with 2-morphisms, and so on. These higher morphisms are necessary to
obtain a fully local gauge theory with local gauge symmetries as well;
they typically appear in the
quantization of reducible gauge systems with large symmetries, for example in
Batalin--Vilkovisky quantization the off-shell higher gauge symmetries
require introduction of both ghost fields and ghosts-for-ghosts. We
recall here
the pertinent category theory definitions; see e.g.~\cite{Baez} for
further details and references.

Here we shall only consider
the first member of this higher categorical hierarchy. A
\emph{2-category} consists of a set of objects $\check A$ (gauge potentials), a set of morphisms
$g:\check A\longrightarrow \check B$ from a source potential $\check
A$ to a target potential $\check B$ (gauge transformations), and a
set of 2-morphisms $\alpha:g\Longrightarrow h$ from a source gauge
transformation $g:\check A\longrightarrow \check B$ to a target gauge
transformation $h:\check A\longrightarrow \check B$ (gauge-for-gauge transformations). The composition of gauge transformations $g:\check B\longrightarrow
\check C$ and $h:\check A\longrightarrow \check B$ is written $g\circ h: \check A\longrightarrow
\check C$; it defines an associative multiplication with unit the identity
gauge transformation $\unit_{\check A}:\check
A\longrightarrow \check A$ for all gauge potentials $\check A$. For the gauge-for-gauge transformations there are two
types of composition. The ``vertical'' composition of 2-morphisms
$\alpha:g\Longrightarrow g'$ and $\alpha':g'\Longrightarrow
g^{\prime\prime}$ between gauge transformations $g,g',g^{\prime\prime}
:\check A\longrightarrow
\check B$ with the same source and target is denoted $\alpha\cdot\alpha':g\Longrightarrow
g^{\prime\prime}$; it defines an associative multiplication with unit
the identity vertical gauge-for-gauge transformation
$\unit_g:g\Longrightarrow g$ for all gauge transformations $g:\check A\longrightarrow
\check B$. The ``horizontal'' composition of 2-morphisms
$\alpha_1:g_1\Longrightarrow g_1'$ and $\alpha_2:g_2\Longrightarrow
g_2'$ between gauge transformations $g_1,g_1':\check B\longrightarrow
\check C$ and $g_2,g_2':\check A\longrightarrow \check B$ is written
$\alpha_1\circ\alpha_2: g_1\circ g_2\Longrightarrow g_1'\circ g_2'$, where
$g_1\circ g_2,g_1'\circ g_2':\check A\longrightarrow \check C$; again it
gives an associative multiplication with unit the identity horizontal
gauge-for-gauge transformation $\unit_{\unit_{\check A}}:\unit_{\check
  A}\Longrightarrow\unit_{\check A}$ for all gauge potentials $\check
A$. The two compositions of 2-morphisms obey the interchange law
\beq
(\alpha_1'\cdot\alpha_1)\circ(\alpha_2'\cdot\alpha_2)=
(\alpha_1'\circ\alpha_2')\cdot (\alpha_1\circ\alpha_2) \ .
\label{interchangelaw}\eeq

A \emph{2-groupoid} structure on this 2-category is obtained by
requiring that every gauge transformation $g:\check A\longrightarrow \check
B$ has an inverse $g^{-1}: \check B\longrightarrow \check A$ such that
$g^{-1}\circ g=\unit_{\check A}$ and $g\circ g^{-1}=\unit_{\check B}$, while
every
gauge-for-gauge transformation $\alpha:g\Longrightarrow h$ has a vertical inverse
$\alpha_v^{-1}:h\Longrightarrow g$ such that $\alpha_v^{-1}\cdot
\alpha=\unit_g$ and $\alpha\cdot \alpha_v^{-1}=\unit_h$, and a
horizontal inverse $\alpha_h^{-1}:g^{-1} \Longrightarrow h^{-1}$ such that
$\alpha_h^{-1}\circ \alpha=\unit_{\unit_{\check A}}$ and
$\alpha\circ\alpha_h^{-1}= \unit_{\unit_{\check B}}$.

The corresponding set of 2-isomorphism classes now has the structure
of a \emph{2-group} $\scrG$, which is a 2-groupoid with one object;
the morphisms $g$ form a group $G$ under composition whose unit
element is the identity morphism, while the 2-morphisms
$\alpha:g\Longrightarrow g'$ for $g,g' \in G$ form a group under horizontal
composition and can be composed vertically with the two compositions
tied together via the interchange law
(\ref{interchangelaw}). Equivalently, a 2-group is a groupoid $\scrG$
equipped with a monoidal structure that obeys the usual
group axioms; in many cases of interest we require that the axioms
(and in general all associativity conditions on the category) hold
only in a ``weak'' sense, i.e. up to natural isomorphism imposed
e.g. by equivalence relations. 

An
important classification result states that 2-groups are the same things as \emph{crossed
  modules}~\cite{Baez}, which consist of pairs of groups 
$G,H$ together with an
action of $G$ on $H$ by automorphisms and a group homomorphism $\ttt:H\to G$ which is
$G$-equivariant, i.e. 
\beqa
\ttt(g\triangleright h)=g\, \ttt(h)\, g^{-1}
\eeqa
for all $g\in G$ and $h\in H$, and which satisfies the Peiffer
identity
\beqa
\ttt(h)\triangleright h' = h\, h'\, h^{-1}
\eeqa
for all $h,h'\in H$. Given a 2-group $\scrG$, we construct a crossed
module by taking $G$ to be the group of morphisms of $\scrG$,
$H$ as the group of 2-morphisms of $\scrG$ whose source is the
identity morphism, the homomorphism $\ttt:H\to G$ is defined by sending a
2-morphism to its target morphism, and the $G$-action on $H$ is defined
by $g\triangleright h =\unit_g\circ h\circ \unit_{g^{-1}}$. In this classification
the 2-morphisms of $\scrG$ form the crossed-product group $G\ltimes H$ under
horizontal composition: 2-morphisms $g\Longrightarrow
g'$ are equivalent to pairs $(g,h)\in G\times H$ with $g'=\ttt(h)\,
g$. Moreover, the vertical composition of $g\Longrightarrow \tilde
g=\ttt(h)\, g$ and $g'\Longrightarrow \tilde g'=\ttt(h'\,)\, g'$ is given by
\beqa
(g,h)\cdot(g',h'\,) = (g,h'\, h)
\eeqa
for composable 2-morphisms, i.e. when $g'=\ttt(h)\, g$ so that $\tilde
g'=\ttt(h'\,)\, \ttt(h)\, g=\ttt(h'\, h)\, g$, while horizontal composition
can be represented as
\beqa
(g,h)\circ(g',h'\,)= \big(g\, g'\,,\, h\, (g\triangleright h'\,)\big) \ .
\eeqa
In the ``weak'' case one replaces the homomorphism $\ttt:H\to G$ with
an element $[a]$ of the group cohomology $\RH^3(G,H)$ which comes from
the associator isomorphism $a$ on $\scrG\times\scrG\times \scrG$~\cite{Baez}. 

For example, define the \emph{shifted} group ${\rm b}\uo$ as the Lie
2-group where $G=\unit$ is the trivial group, $H=\uo$, and $\ttt$ is the trivial
homomorphism. Then a principal ${\rm b}\uo$-2-bundle with
2-connection, i.e. an object in the 2-category of $\uo$-bundles with connection, is a gerbe with 2-connection~\cite{Baez1}. On a trivial
gerbe, a 2-connection is a globally defined two-form
$B\in\Omega^2(M)$; in this case the holonomies of paths, regarded
as morphisms between points $x\in M$ in the path 2-groupoid of the smooth
manifold $M$, are all trivial while the 2-holonomies over
two-chains $\Sigma\subset M$ with $\partial\Sigma\neq0$, regarded as 2-morphisms in the path
2-groupoid, give elements $\exp\big(\ii\int_\Sigma\, B\big)\in\uo$.

\subsubsection*{Currents}

In the non-vacuum theory, i.e. when there are currents
present, the description of the configuration space changes somewhat. Consider a differential cocycle
\beqa
\check j= (c,h,\omega) \ \in \ \check\scrH{}^p(M) \ ,
\eeqa
so that $c=c(\check j)$, $j:=\omega=F(\check j)$, and $h$ is the
``holonomy'' of $\check j$. A \emph{trivialization} of $\check j$ is a
\emph{differential cochain}
\beqa
\check A \ \in \ \check C{}^{p-1}(M)= C^{p-1}(M;\IZ)\times C^{p-2}(M;\IR)\times
\Omega^{p-1}(M) \qquad \mbox{with} \quad \dd_{\check\RH}\check A=\check j \ .
\eeqa
Associated to $\check A$ is a differential form $A\in\Omega^{p-1}(M)$
such that $\dd A=F(\check j)=j\in \Omega_{\rm cl}^p(M)$.

Starting from
a current $j$, we consider a \emph{refinement} of $j$ to a
differential cocycle $\check j$ such that $F(\check j)=j$, and
regard the configuration space of higher abelian gauge fields as
the set of trivializations of $[\check j]\in\check\RH{}^p(M)$. The
set of all such trivializations is a torsor for the group
$\check\RH{}^p(M)$, and the group of gauge transformations is
isomorphic to $\check\RH{}^{p-1}(M)$.

\subsection{Generalized abelian gauge theory\label{Gendiffcoh}}

A \emph{generalized differential cohomology theory} $\check\E{}^\bullet$ is a
geometric refinement of a generalized cohomology theory $\E^\bullet$
on the category of smooth manifolds. For any manifold $M$, it
completes the pullback square
\beq
\xymatrix{
\check\E{}^\bullet(M) \ \ar[r] \ar[d] & \ \Omega_{\rm
  cl}(M;E^\bullet)^\bullet \ar[d]^{[-]_{\rm dR}} \\
\E^\bullet(M) \ \ar[r]_{\!\!\!\!\!\!\varphi} & \ \RH(M;E^\bullet)^\bullet
}
\label{Epullback}\eeq
where $E^\bullet := \E^\bullet(\pt)\otimes_\IZ\IR$ is a graded real
vector space, the image of the morphism $\varphi$ is a full lattice
while its kernel is the torsion subgroup $\ker\varphi= \Tor\,
\E^\bullet(M)$ (in particular $\varphi\otimes\IR$ is a group isomorphism), and the grading on ordinary cohomology is given by the
total degree
\beqa
\RH(M;E^\bullet)^d := \bigoplus_{p+q=d}\, \RH^p(M;E^q) \ .
\eeqa

\begin{theorem}
Differential cohomology theories exist for any generalized cohomology theory.
\label{HStheorem}\end{theorem}

The proof of Thm.~\ref{HStheorem} can be found
in~\cite{Hopkins:2002rd} and it consists in replacing differential
cocycles by certain maps to the corresponding classifying spaces
$B\E$, the collection of which are called ``differential function
spaces'': One now considers triples $(c,h,\omega)$ where $c:M\to B\E$
is a smooth map,
$h\in C^\bullet(M)$ is a cochain on $M$, and $\omega\in\Omega_{\rm
  cl}^\bullet(M)$ is a closed differential form on $M$ such that $\delta h=\omega-c^*i$ for some cocycle $i\in
Z^\bullet(B\E)$, together with homotopy equivalence relations. Existence and
uniqueness of generalized differential cohomology theories is also
investigated by~\cite{BunkeSchick2007,BunkeSchick2009}. The groups
$\check\E{}^d(M)$ have analogous properties to those of the
Cheeger--Simons groups, including:
\begin{itemize}
\item[(i)] The connected components of $\check\E{}^d(M)$ are labelled by
  the charge group $\pi_0\check\E{}^d(M)=\E^d(M)$.
\item[(ii)] There is a field strength map $F:\check\E{}^d(M)\to
  \Omega_{\rm cl}(M;E^\bullet)^d$.
\item[(iii)] There is a torus $\torus_{\E}^d(M)=\E^{d-1}(M;\torus)$ of
  flat fields in $\check\E{}^d(M)$, where $\E$-cohomology with
  coefficients in $\IR/\IZ$ is constructed using stable homotopy
  theory in terms of a spectrum in such a way that it fits into natural long exact sequences
\beqa
\cdots \ \longrightarrow \ \E^{d-1}(M) \ \longrightarrow \ \E^{d-1}(M)\otimes\IR \ \longrightarrow \
\E^{d-1}(M;\torus) \ \longrightarrow \ \E^d(M) \ \longrightarrow \
\cdots \ .
\eeqa
See~\cite{BunkeSchick2009} and~\cite[\S2.3]{Bunke:2010mq} for
conditions under which there is an isomorphism $\torus_{\E}^d(M)\cong
\ker(F)$ of cohomology theories, and~\cite{Hopkins:2002rd} for an
explicit realization of the isomorphism.
\end{itemize}

Most of what we have said before concerning (higher) abelian gauge theories
have analogs in this more general setting. A \emph{generalized abelian
  gauge theory} is determined by a multiplicative
generalized cohomology theory $\E^\bullet$ with invertible closed
differential forms
\beq
\omega_M=\sqrt{F\big([\ort_{\E}(M)]\big)}
\label{squareroot}\eeq
depending
functorially and locally on $M$ that normalize the morphism $\varphi$
in (\ref{Epullback}) such that the intersection pairing
$\E^\bullet(M)\otimes\E^\bullet(M)\to \IZ$ becomes compatible with the integration of
curvatures. Here $[\ort_{\E}(M)]\in\check\E{}^\bullet(M)$ is a smooth
$\E$-orientation on $M$, and we use invertibility to regard cup products with the cohomology class of $F\big([\ort_{\E}(M)]\big)$ as an invertible linear operator on $\RH(M,E^\bullet)^\bullet$; if the cohomology is finitely-generated, then the square root (\ref{squareroot}) can be defined using the usual Jordan normal form. A gauge potential $\check A\in\check
C{}^d(M)$ is a non-flat trivialization of a current $\check
j\in\check Z{}^{d+1}(M)$ for some $d\in\IZ$. If $\check j=0$, then the potential has
a class $[\check
A]\in\check \E{}^d(M)$; we will mostly consider this source-free case
in this article, as the gauge
theory in this instance involves only free fields and so its
quantization can be carried out in a rigorous way, e.g. along the
lines discussed in~\cite{Freed:2006ya,Freed:2006yc}. The gauge field $F=F([\check A])$ is a differential form of total
degree $d$ on spacetime, i.e. an element
\beqa
F \ \in \ \Omega(M;E^\bullet)^d=\bigoplus_{k\in\IZ}\,
\Omega^k(M;E^{d-k}) \ .
\eeqa

In this paper we are primarily interested in two particular physical
applications of this general construction.
\begin{itemize}
\item When $\E^\bullet=\RH^\bullet$ is ordinary cohomology, this
  construction gives the differential cohomology $\check\RH{}^\bullet$ considered in
  this section. In this case $H^\bullet=\IR$ and $\varphi$ is the map $\RH^\bullet(M;\IZ)\hookrightarrow
  \RH^\bullet(M;\IR)$ induced by the inclusion $\IZ\hookrightarrow
  \IR$ of abelian groups with $F\big([\ort_{\RH}(M)]\big)=1$. This is the model that is appropriate to
  higher abelian gauge theory with extended $p-1$-brane charges,
  e.g.~the fluxes of supergravity.
\item When $\E^\bullet= {\rm R}^{-1}$ is the version of connective
  KO-theory defined by its truncated Postnikov spectrum of real vector
  spaces at degree four, which sits
inside an exact sequence
\beqa
0 \ \longrightarrow \ \RH^3(M;\IZ) \ \longrightarrow \ {\rm
R}^{-1}(M) \ \longrightarrow \ {\begin{matrix} \RH^0(M;\IZ_2)\\ \oplus
\\ \RH^1(M;\IZ_2) \end{matrix}} \ \longrightarrow \ 0 \ ,
\eeqa
the corresponding smooth refinement $\check{\rm R}{}^{-1}$ can be
modelled by a 2-groupoid
consisting of invertible elements of a symmetric monoidal 2-category
whose objects are $\IC$-algebras, morphisms are $\IZ_2$-graded algebra
bimodules, 2-morphisms are intertwiners between bimodules, and
monoidal structure provided by tensor product of
$\IC$-algebras. This model consistently reconciles the target space and
worldsheet field theories of the $B$-field; the three-form $H$-flux in Type~II
  superstring theory is
Dirac quantized by a ``differential $B$-field'' which is an object of
$\check{\rm R}{}^{-1}(M)$, whose charge group classifies the
twistings of the complex
K-theory of orbifolds.
 See~\cite{DFM,DFM2} for details.
\item When $\E^\bullet=\K^\bullet$ is complex K-theory, this construction
  gives the differential K-theory $\check\K{}^\bullet$ of vector bundles
  with connection on $M$. Here $K^\bullet=\IR(u)$ with $u^{-1}$ the
  Bott generator of degree $\deg(u^{-1})=-2$, and $\varphi$ is the
  modified Chern character from K-theory to real cohomology with
  normalising differential form
  $F\big([\ort_{\K}(M)]\big) = \widehat{A}(M)$ the Atiyah--Hirzebruch class
  of the tangent bundle of $M$. This is
  the model that is appropriate to Ramond--Ramond gauge theory with
  D-brane charges in Type~II superstring theory. In this case the $p$-forms
  $F$ of \S\ref{GAGT}
correspond to Ramond--Ramond fields and
the submanifolds $W_e$ to the worldvolumes of D-branes; now the charges $q_e$ are classes in the complex K-theory
group $\K^0(W_e)$ and their pushforwards under the embedding $W_e\hookrightarrow M$
gives the Ramond--Ramond charge $[j_e]$ of a D-brane
wrapping $W_e$. These considerations are the subject
of the next section.
\end{itemize}

\bigskip

\section{Ramond--Ramond fields and differential K-theory\label{RRfields}}

\subsection{D-branes and K-cycles\label{DBranes}}

We begin with a brief mathematical introduction to D-branes in Type~II
superstring theory;
see~\cite{Szabo:2008hx} for further details and references. A D-brane
is a suitable boundary condition for the Euler--Lagrange equations in
a two-dimensional superconformal field theory on an open oriented
Riemann surface $\Sigma$. It is realized as a submanifold $W\subset M$
of spacetime onto which ``open strings attach''. Topologically, open
strings are the relative maps
\beqa
(\Sigma,\partial\Sigma) \ \longrightarrow \ (M,W) \ .
\eeqa
Compatibility with superconformal invariance constrains the allowed
``worldvolumes'' $W$, e.g. in the absence of $H$-flux, 
Freed--Witten anomaly cancellation requires that $W$ be a spin$^c$
manifold (and hence K-oriented).

D-branes actually have more
structure: The worldvolume $W$ carries a complex
``Chan--Paton vector bundle'' $E$ with connection $\nabla$; here the
rank of $E$ is the ``number of D-branes wrapping'' $W\subset M$. Hence
the ``charge'' of a D-brane may be naturally considered as the complex
K-theory class $[E]\in\K^0(W)$.

However, in analogy with the charges in higher abelian gauge
theory, a homology theory is more natural in the description of
D-brane charges (see e.g.~\cite{Szabo:2002jv}); here we shall use the Baum--Douglas geometric
formulation of
K-homology~\cite{BaumDouglas}, which encodes important physical
aspects of D-branes in Type~I and
Type~II string theory~\cite{Reis:2005pp,Szabo:2008hx,Reis:2006th}. Most importantly, D-branes
are dynamical objects that ``interact with'' or ``couple to'' the
Ramond--Ramond fields of supergravity; this requires considerations
from differential K-theory which we explain later on. For the
moment, we briefly highlight a few of the salient features of this topological
description.

Let $M$ be a ten-dimensional spin manifold; a choice of spin structure
enables us to introduce fermions. Sometimes the manifold $M$ can also refer only to a
subspace of spacetime, for example when one considers string
compactifications; in that case $m=\dim(M)<10$. We suppose that there is no
background $H$-flux, which means that we can use untwisted cohomology
theories.

\begin{definition}
A \emph{D-brane} in $M$ is a Baum--Douglas K-cycle $(W,E,f)$, where
$f:W\hookrightarrow M$ is a closed spin$^c$ submanifold called the
\emph{worldvolume}, and $E\to W$ is a complex vector bundle with
connection called the \emph{Chan--Paton gauge bundle}.
\end{definition}

Note that the Chan--Paton bundle defines a \emph{stable} element of
$\K^0(W)$. The collection of D-branes described in this way forms an
additive category under disjoint union of K-cycles. The quotient by Baum--Douglas
``gauge equivalence'' is isomorphic to the analytic K-homology of
$M$; the isomorphism classes of K-cycles generate the geometric
K-homology $\K_\bullet(M)$. In this way D-branes naturally provide K-homology classes on $M$
which are dual to K-theory classes $f_![E]\in\K^d(M)$, where $f_!$ is
the Gysin map in K-theory and $d$ is the codimension of $W$ in
$M$. There is a natural $\IZ_2$-grading on K-cycles given by the
parity of the dimension of the worldvolume; the K-cycle $(W,E,f)$ is odd in Type~IIA
string theory and even in Type~IIB string theory.

Here we have used the usual physical intuition of a D-brane as an embedded
submanifold $S\subset M$ of spacetime. However, not all D-branes (regarded as consistent
boundary conditions in the underlying boundary conformal field theory)
admit such a geometric description. When they do we say that the
D-brane is ``representable''. The description of D-branes in terms of
Baum--Douglas K-cycles for K-homology naturally requires
non-representable branes~\cite{Reis:2005pp} with arbitrary maps
$f:W\to M$. A D-brane $(W,E,f)$ is said to \emph{wrap}
$S\subset M$ if $f(W)\subset S$, and to \emph{fill} $S$ if
$f(W)=S$. 

As a consequence of the gauge equivalence relation of vector bundle modification,
together with the K-theory Thom isomorphism,
the class of any D-brane on $W$ may be expressed in terms of virtual
K-cycles on $M$ through
\beqa
[M,\cals_E^+,\Id_M]-[M,\cals_E^-,\Id_M]= \pm\, [W,E,f] \ ,
\eeqa
where $\cals_E^\pm \to M$ are twisted spinor bundles. This is the
K-homology version of the \emph{Atiyah--Bott--Shapiro construction}; the
reduction of classes from the left-hand side to the right-hand side of this
equation represents the standard Sen--Witten construction of D-branes through \emph{tachyon condensation} on a brane-antibrane system
wrapping the whole spacetime $M$~\cite{Witten:1998cd,Olsen:1999xx}. More generally, for a given worldvolume $S$ the Sen--Witten
construction naturally assigns D-brane charges to classes of wrapped
branes in the K-homology $\K_\bullet(S)$~\cite{Reis:2005pp}.

The cohomological formula for the charge of a D-brane arises
physically through ``anomaly cancellation'' arguments. Mathematically,
it arises very naturally through the \emph{modified Chern character}
$\ch\mapsto \ch\smile \sqrt{\widehat{A}(M)}$ which, by the
Atiyah--Singer index theorem, is an isometric isomorphism between
K-theory and cohomology groups over $\IR$ with respect to their
natural bilinear pairings. On K-theory this pairing is given by the index
of the twisted Dirac operator, which coincides with the natural
intersection form on boundary states and computes the chiral fermion
anomaly on D-branes; on cohomology it is given by
evaluation on the fundamental class. Then the \emph{charge}
of a D-brane $(W,E,f)$ is given by
\beq
Q(W,E,f)=\ch\big(f_![E]\big) \smile \sqrt{\widehat{A}(M)} \ \in \
\RH^\bullet(M;\IR) \ .
\label{Dbranecharge}\eeq
This form of the charge vector is called the \emph{Minasian--Moore
  formula}~\cite{Minasian:1997mm}. It respects the Baum--Douglas gauge equivalence relations.

\subsection{Ramond--Ramond gauge theory\label{RRGT}}

Let us now look at D-branes and their charges as currents in a
suitable generalized abelian gauge theory. As before, let $M$ be a
ten-dimensional spin manifold. The generalized abelian gauge theory of Ramond--Ramond fields arises as a low-energy limit
of Type~II superstring theory from the quantum Hilbert space of states of
closed superstrings on the manifold $M$; we are interested in the
corresponding equations of motion in Type~II supergravity. We begin with a description
of topologically trivial Ramond--Ramond fields as elements of the
differential complex $\Omega^\bullet(M)$.

For this, let $K^\bullet$ be the $\IZ$-graded real
vector space of Laurent polynomials
$\IR(u)$, where $u^{-1}$ is the Bott generator of
$\K^{-2}(\pt)\cong\IZ[u^{-1}]$ with $\deg(u)=2$. Let
$\Omega(M;K^\bullet)^j$ denote the vector space of differential forms
on $M$ of \emph{total} degree $j$. In this article we are interested
mostly in the cases $j=0,-1$; an element $F\in\Omega(M;K^\bullet)^j$
then admits an expansion in even degree forms
\beq
F=\sum_{k=0}^5\, u^{-k}\otimes F_{2k} \qquad \mbox{for} \quad j=0
\label{IIARR}\eeq
and in odd degree forms
\beq
F=\sum_{k=1}^5\, u^{-k}\otimes F_{2k-1} \qquad \mbox{for} \quad j=-1 \
,
\label{IIBRR}\eeq
where $F_p\in\Omega^p(M)$. We call an element
$F\in\Omega(M;K^\bullet)^j$ the total \emph{Ramond--Ramond field
  strength}, where $j=0$ for the Type~IIA string theory while $j=-1$
for the Type~IIB string theory. In the IIA theory, the standard
supergravity Bianchi identity is
\beqa
\dd F=0 \ .
\eeqa

The space of forms $\Omega(M;K^\bullet)^j$ is a symplectic vector
space with symplectic form given by
\beqa
\omega_j=\Big[\, \frac12\, \int_M\, \delta F\wedge\Psi^{-1}_j(\delta F)\,
\Big]_{u^0} \ ,
\eeqa
where the operation $[-]_{u^0}:= u^{-{\rm deg}/2}\, (-)$ projects out the constant coefficient
in $\Omega^j(M)$
of a Laurent polynomial in $\IR(u)$, and
\beqa
\Psi^{-1}_j\, :\, \Omega(M;K^\bullet)^j \ \longrightarrow \
\Omega(M;K^\bullet)^{10-j}
\eeqa
is the map which essentially sends $F$ to its complex conjugate
$\overline{F}$, where $\overline{u}:=-u$; on decomposable elements
$F=u^k\otimes f$ it is given by
\beqa
\Psi^{-1}_j(u^k\otimes f):= (-1)^{j\, (j-1)/2}\, u^{5-j}\, (-u)^k
\otimes f \
.
\eeqa
In the IIA/IIB theories, one has the respective expansions
\beqa
\omega_0(F,G)=\int_M \ \sum_{k=0}^5\, (-1)^{k+1}\, F_{2k}\wedge
G_{10-2k}
\eeqa
and
\beqa
\omega_{-1}(F,G)=\int_M \ 
\sum_{k=1}^5\, (-1)^{k+1}\, F_{2k-1}\wedge G_{11-2k} \ .
\eeqa

Now we let $(M,g)$ be a ten-dimensional lorentzian spin manifold. Then
there is a metric of indefinite signature on the vector space
$\Omega(M;K^\bullet)^j$ given by
\beqa
g_j(F,G):=\Big[\, \int_M\, F\wedge \check\iota(G)\,
\Big]_{u^0} \ ,
\eeqa
where the map
\beqa
\check\iota \,:\, \Omega(M;K^\bullet)^j \ \longrightarrow \
\Omega\big(M \,;\,
(K^\bullet)^*\big)^{10-j}
\eeqa
is essentially the Hodge duality operator $\star$ associated to the lorentzian
metric $g$ on $M$, with the convention $\star:u\mapsto u^{-1}$; on decomposable elements
$F=u^k\otimes f$ it is given by
\beqa
\check\iota(u^k\otimes f)= u^{-k}\otimes \star f \ .
\eeqa
In the IIA/IIB theories, one has the respective expansions
\beqa
g_0(F,G)=\int_M \ \sum_{k=0}^5\, F_{2k}\wedge \star
G_{2k} \qquad \mbox{and} \qquad g_{-1}(F,G)=\int_M \ 
\sum_{k=1}^5\, F_{2k-1}\wedge\star G_{2k-1} \ .
\eeqa

The metric and symplectic form together define an involution
\beq
I_j\, :\, \Omega(M;K^\bullet)^j \ \longrightarrow \
\Omega(M;K^\bullet)^{j} \ , \qquad
I_j(F):=-\big((\Psi^{-1}_j)^{-1}\circ\check\iota\big)(F)
\label{Ijdef}\eeq
which has the property
\beqa
g_j(F,G)=\omega_j\big(I_j(F)\,,\, G\big) \ .
\eeqa
In the IIA/IIB theories, one has the respective expansions
\beq
I_0(F)= \sum_{k=0}^5\, (-1)^{k+1} \, u^{-k}\otimes \star F_{10-2k} \qquad
\mbox{and} \qquad I_{-1}(F)= \sum_{k=1}^5\, (-1)^k\, u^{-k}\otimes
\star F_{11-2k} \ .
\label{I0I-1}\eeq
This involution defines a $\IZ_2$-grading on the vector space of
differential forms through the decomposition into $\pm\,1$ eigenspaces
\beqa
\Omega(M;K^\bullet)^j= \Omega_+(M;K^\bullet)^j\oplus
\Omega_-(M;K^\bullet)^j 
\eeqa
with $I_j$ acting as multiplication by $\pm\,1$ on $\Omega_\pm(M;K^\bullet)^j$.

Forms $F^+\in \Omega_+(M;K^\bullet)^j$, i.e. forms $F^+$ which satisfy
\beq
I_j(F^+)=+\, F^+ \ ,
\label{IjF+}\eeq
are called \emph{self-dual forms}. An alternative
formulation of the self-duality equation (\ref{IjF+}) can be given using Clifford
algebras~\cite{Belov:2006xj}. On any lorentzian manifold $(M,g)$ of dimension $4k+2$,
$k\in\IN$, the associated volume form ${\rm vol}_g$ defines a Clifford algebra
involution $\Gamma:=\ttc({\rm vol}_g)$, where $\ttc(\omega)$ denotes
Clifford multiplication by the form $\omega$. We may then define the
involution $I_j$ on $\Omega(M;K^\bullet)^j$ by
\beqa
\ttc\big(I_j(F)\big):= \Gamma\, \ttc(F)\ .
\eeqa
This definition agrees with (\ref{I0I-1}) by using the property
\beqa
\Gamma\, \ttc(F_p) = (-1)^{p\,(p+1)/2}\, \ttc(\star F_p)
\eeqa
of Clifford multiplication for $F_p\in\Omega^p(M)$. Note that if we
work in euclidean signature instead, then $I_j$ would define a complex
structure which is compatible with the euclidean metric $g_j$ and the
symplectic form $\omega_j$; see~\cite{Belov:2006xj} for further details.

All the supergravity equations of motion for the Ramond--Ramond fields
in Type~II string theory can be rewritten using only the self-dual
form $F^+$. In particular, the equation
\beqa
\dd F^+=0
\eeqa
contains both the Bianchi identity and the Ramond--Ramond equation of
motion
\beqa
\dd F=0 \qquad \mbox{and} \qquad \dd\star F=0 \ .
\eeqa
These equations are identical to the equations of motion that we
encountered in \S\ref{Abgauge} for abelian gauge theory, like the
Maxwell theory. This suggests that the Ramond--Ramond gauge theory
should be modelled on some sort of generalized differential
cohomology theory; we shall elaborate on this model soon.

For the
Type~IIA theory, the self-duality equations are
\beqa
F_6=-\star F_4 \ , \qquad F_8=\star F_2 \qquad \mbox{and} \qquad
F_{10}=-\star F_0 \ ,
\eeqa
and hence from (\ref{IIARR}) the Type~IIA self-dual Ramond--Ramond
field has an expansion
\beqa
F^+=u^0\otimes F_1+u^{-1}\otimes F_2+u^{-2}\otimes F_4+u^{-3}\otimes\star
F_4-u^{-4}\otimes \star F_2+u^{-5}\otimes\star F_0
\eeqa
with the equations of motion
\beqa
\dd F_0=\dd F_2=\dd F_4=0=\dd\star F_4=\dd\star F_2 \ .
\eeqa
Similarly, from (\ref{IIBRR}) the Type~IIB self-dual Ramond--Ramond
field has an expansion
\beqa
F^+= u^{-1}\otimes F_1+u^{-2}\otimes F_3+u^{-3}\otimes
F_5^++u^{-4}\otimes \star F_3-u^{-5}\otimes\star F_1 \ ,
\eeqa
with $\star F_5^+=-F_5^+$ and the equations of motion
\beqa
\dd F_1=\dd F_3=0=\dd\star F_3=\dd\star F_1 \qquad \mbox{and} \qquad
\dd F_5^+=0 \ .
\eeqa

Note that these expressions are just particular parametrizations of the self-dual
field $F^+$: It can always be written as
\beqa
F^+=F+I_j(F) \qquad \mbox{with} \quad F\in \Omega(M;K^\bullet)^j \ ,
\eeqa
but the form $F$ is \emph{not} unique. The self-duality condition on the total
Ramond--Ramond field strength is a feature that we will have to deal
with carefully later on. Together with the constraint on the parity of
differential form degree, it is the requirement of the GSO projection
in the underlying closed superstring theory on the manifold $M$. Note also
that, generally, a necessary condition for the existence of solutions
to the self-duality equations on a lorentzian manifold $(M,g)$ is that $\dim(M)=4k+2$ with $k\in\IN$;
only under these conditions is the map (\ref{Ijdef}) an involution.

\subsubsection*{Ramond--Ramond currents}

At the level of topologically trivial fields, the Ramond--Ramond
fields are \emph{sourced} by D-branes on $M$, just like the
electromagnetic fields of Maxwell theory are \emph{sourced} by
electrically charged particles. ``Anomaly cancellation'' arguments in
the Ramond--Ramond abelian gauge theory requires that, in the presence
of a D-brane $(W,E,f)$, the Ramond--Ramond fields obey the equations
of motion
\beq
\dd F=0 \qquad \mbox{and} \qquad \dd\star F=j(W,E,f) \ ,
\label{RReoms}\eeq
where $j(W,E,f)$ is the ``Ramond--Ramond current'' whose cohomology
class is given by the D-brane charge vector (\ref{Dbranecharge}) and
as before we use a distributional representative of the Poincar\'e
dual class $\PD_M(f(W))$. Note the formal equivalence of (\ref{RReoms}) to the equations of
motion (\ref{GAGTeom}) of (higher) abelian gauge theories. In particular, the
Ramond--Ramond field is a trivialization of the Ramond--Ramond
current. However, these equations are incompatible with the
self-duality of $F$.

\subsection{Semi-classical quantization}

Since the Ramond--Ramond charges of D-branes are classified by the complex
K-theory $\K_c^{j}(M)$ of spacetime with compact support (with
$j=0/-1$ in Type~IIB/IIA theory), it is natural to expect that the
Ramond--Ramond fields are also classified in some way via
K-theory. This statement about the K-theory classification of
Ramond--Ramond fields on a locally compact spacetime $M$ follows from
the relation between D-brane charges and the group of Ramond--Ramond
fluxes ``measured at infinity''.

For this, let $M$ be a non-compact manifold, e.g. $M=\IR\times N$ with
$N$ non-compact. If the worldvolume $W$ is compact inside the K-cycles
$(W,E,f)$, then the current $j(W,E,f)$ is supported in the interior
$\mathring{M}\subset M$. Let $M_\infty$ be the ``boundary of $M$ at
infinity'', e.g. $(\IR^n)_\infty=S^{n-1}$. Since $j(W,E,f)$ is
trivialized by $F$ in $\mathring{M}$, the D-brane charge lives in the
kernel of the natural forgetful homomorphism
\beqa
\mathfrak{q}^\bullet \,:\, \K^\bullet_c(M) \ \longrightarrow \
\K^\bullet(M)
\eeqa
induced by the inclusions
\beqa
(M,\emptyset) \ \hookrightarrow \ (M,M_\infty) \qquad \mbox{and}
\qquad i\,:\, M_\infty \ \hookrightarrow \ M
\eeqa
of (relative pairs of) locally compact spaces; this morphism forgets about the compact support condition, with
$\K_c^\bullet(M)\cong \K^\bullet(M,M_\infty)$ given by relative K-theory. By Bott periodicity,
the long exact sequence for the pair $(M,M_\infty)$ in K-theory
truncates to the six-term exact sequence
\beqa
\xymatrix{\K^{-1}(M_\infty)~\ar[r]&~{\K^{0}(M,M_\infty)}~
\ar[r]^{ \ \ \mathfrak{q}^0}& ~\K^{0}(M)\ar[d]^{i^*}\\
\K^{-1}(M)~\ar[u]^{i^*}&\ar[l]^{\!\!\mathfrak{q}^{-1}}~\K^{-1}(M,M_\infty)~&
\ar[l]~\K^{0}(M_\infty) }
\eeqa
It follows that the D-brane charge groups are given by
\beqa
\ker(\mathfrak{q}^0)\cong \K^{-1}(M_\infty)\,\big/\,
i^*\big(\K^{-1}(M)\big) \qquad \mbox{and} \qquad
\ker(\mathfrak{q}^{-1})\cong \K^0(M_\infty)\,\big/\,
i^*\big(\K^0(M)\big) \ .
\eeqa
We interpret these formulas in the following way. D-brane charge,
regarded as the total charge of a Ramond--Ramond current
$j(W,E,f)$, can only be detected by classes of fields at infinity
$M_\infty$ which are \emph{not} the restrictions of fields defined on
$\mathring{M}$, i.e. which cannot be extended to all of spacetime
$M$. Moreover, the group $\K^j(M)$ for $j=0,-1$ classifies fields $F$
which do not contribute to the D-brane charge, i.e. $\K^j(M)$
topologically classifies gauge equivalence classes of Ramond--Ramond
fields in the absence of D-branes, where $j=0,-1$ for
Type~IIA/IIB string theory. In the following we will \emph{assume}
that this relation between Ramond--Ramond fields and K-theory holds
for arbitrary spacetime manifolds $M$.

Let us now make explicit the relation between the Ramond--Ramond
fields and cohomology, i.e. the de~Rham cohomology class
$[F(\xi)]_{\rm dR}$ associated to an element $\xi\in\K^j(M)$ that
determines (the class of) the Ramond--Ramond field
$F$. In~\cite{Moore2000} it is argued using the Ramond--Ramond
equation of motion from (\ref{RReoms}) that
\beqa
[F]_{\rm dR} \ \in \ \Lambda_{\K^j}\subset \RH(M;K^\bullet)^j \ ,
\eeqa
where the full lattice $\Lambda_{\K^j}$ is the image of the modified
Chern character map
\beqa
\ch\wedge\sqrt{\widehat{A}(M)} \,:\, \K^j(M) \ \longrightarrow \ 
\RH(M;K^\bullet)^j
\eeqa
which is a group isomorphism over $\IR$. This means that the cohomology class of
the Ramond--Ramond field $F$ associated to the K-theory class
$\xi\in\K^j(M)$ is
\beq
\big[F(\xi)\big]=\ch(\xi)\smile \sqrt{\widehat{A}(M)} \ .
\label{RRfieldKclass}\eeq
This formula is interpreted as the Dirac quantization condition for
Ramond--Ramond fields in Type~II string theory. It follows that the
Ramond--Ramond field is given by a representative for an element of
\emph{differential K-theory}, which we now proceed to describe in some
detail.

\subsection{Differential K-theory\label{DiffK}}

The set of gauge inequivalent Ramond--Ramond fields (or equivalently
gauge equivalence classes of Ramond--Ramond currents) lives inside an
infinite-dimensional abelian Lie group $\check\K{}^j(M)$, the
differential K-theory group of spacetime $M$; its connected components
are labelled by the complex K-theory $\K^j(M)$, the group of D-brane charges. The group $\check\K{}^j(M)$
extends the setwise fibre product
\beqa
\Omega_\IZ(M;K^\bullet)^j\times_{[-]} \K^j(M):=\Big\{(F,\xi)\ \Big| \ \ch(\xi)\wedge \sqrt{\widehat{A}(M)} =
[F]_{\rm dR} \Big\} 
\eeqa
by the torus of topologically trivial flat Ramond--Ramond fields,
i.e. there is an exact sequence
\beqa
0 \ \longrightarrow \ \K^{j-1}(M)\otimes\bbt \ \longrightarrow \
\check\K{}^j(M) \ \longrightarrow \ \Omega_\IZ(M;K^\bullet)^j\times_{[-]} \K^j(M) \ \longrightarrow \ 0 \ .
\eeqa
In general, a Ramond--Ramond \emph{potential} is a representative for
a class $[\check C]$ 
in the differential K-theory $\check\K{}^j(M)$. As before, this group
can be characterised by two short exact sequences which are summarised
in the diagram
\beq
\xymatrix{
0 \ar[dr] & & & & 0 \\
 & \K^{j-1}(M;\bbt) \ar[dr] & & \K^j(M) \ar[ur] & \\
 & & \check\K{}^j(M) \ar[ur]^c \ar[dr]^F & & \\
 & \Omega(M;K^\bullet)^{j-1} \,\big/ \, \Omega_\IZ(M;K^\bullet)^{j-1} \ar[ur] & &
 \Omega_\IZ(M;K^\bullet)^j \ar[dr] & \\
0 \ar[ur] & & & & 0
}
\label{diffKexactseqs}\eeq
where the kernel of the field strength map $F$ is the group of flat
Ramond--Ramond fields on $M$, while the kernel of the characteristic
class map $c$ is the torus of topologically trivial Ramond--Ramond
gauge potentials. Let us pause to briefly describe the structures of each of
these classes of fields, as they will play a prominent role in our
applications.

\subsubsection*{Flat Ramond--Ramond fields}

The group $\K^j(M;\bbt)$ of flat fields can be described as explained
in \S\ref{Gendiffcoh}, by using the
short exact sequence of abelian groups (\ref{ZRTseq}) to write the
corresponding long exact sequence in complex K-theory
\beq
\cdots \ \longrightarrow \ \K^j(M) \ \longrightarrow \ \K^j(M)\otimes\IR
\ \longrightarrow \ 
\K^j(M;\bbt) \ \xrightarrow{ \ \delta \ } \ \K^{j+1}(M) \ \longrightarrow \
\cdots \ .
\label{Ktheorylongexact}\eeq
Using the Chern character homomorphism, the flat Ramond--Ramond fields
are thus described by the short exact sequence
\beqa
0 \ \longrightarrow \ \RH^j(X;\IR)\,\big/ \, \Lambda_{\K^j} \
\longrightarrow \ \K^j(M;\bbt) \ \xrightarrow{ \ \beta \ } \
\Tor\, \K^{j+1}(M) \ \longrightarrow \ 0 \ ,
\eeqa
where the Bockstein homomorphism $\beta$ is induced from the
connecting homomorphism $\delta$ in the long exact sequence
(\ref{Ktheorylongexact}). If the K-theory group $\K^{j+1}(M)$ is pure
torsion, then the flat torsion Ramond--Ramond fields can be
represented in terms of virtual flat vector bundles over $M$ as
\beq
\K^j(M;\bbt)\cong \Tor\, \K^{j+1}(M) \ .
\label{KjMTTor}\eeq
By Pontrjagin duality of K-theory, we have
\beq
\K^j(M;\bbt)\cong\Hom_{\scrAb}\big(\K_j(M)\,,\, \bbt\big) \ ,
\label{PontrdualityKth}\eeq
where $\K_\bullet$ is geometric K-homology. The
isomorphism (\ref{PontrdualityKth}) follows by using the fact that
$\bbt=\IR/\IZ$ is a divisible group along with the universal coefficient
theorem for K-theory to find $$\Ext_{\scrAb}\big(\K_j(M)\,,\,\bbt \big)=0$$ for all $j\in\IZ$;
this implies that the contravariant functor $G\mapsto \Hom_{\scrAb}(G,\bbt)$
from the category of abelian groups into itself takes exact sequences
into exact sequences.

The Ramond--Ramond flux couplings implied by
(\ref{KjMTTor})--(\ref{PontrdualityKth}) can be made explicitly to
``background D-branes'', which are \emph{K-chains} $\big( \,
\widetilde{W}\,,\, \widetilde{E}\,,\, \widetilde{f}\ \big)$ whose boundary
\beq
\partial_{\K} \big( \,
\widetilde{W}\,,\, \widetilde{E}\,,\, \widetilde{f}\ \big):= \big( \,
\partial\widetilde{W}\,,\, \widetilde{E}\,\big|_{\partial{\widetilde{W}}}\,,\,
\widetilde{f}\, \big|_{\partial{\widetilde{W}}}\, \big) = (W,E,f)
\label{Kchainbdry}\eeq
is a Baum--Douglas K-cycle. The \emph{holonomy} over such a
D-brane background with flat Ramond--Ramond flux given by
\beqa
\xi=[E_0]-[E_1] \ \in \ \K^{-1}(M;\bbt)\cong\Hom_{\scrAb}\big(\K_{\rm odd}(M) \,,\,
\bbt \big) \ ,
\eeqa
where $E_0,E_1$ are complex vector bundles on $M$ of equal rank, is
then defined by the virtual K-chain
\beq
\big( \,
\widetilde{W}\,,\, \widetilde{f^*E_0}\,,\, \widetilde{f}\ \big) - \big( \,
\widetilde{W}\,,\, \widetilde{f^*E_1}\,,\, \widetilde{f}\ \big) \ .
\label{virtualKchain}\eeq
Unlike the couplings to D-branes, these couplings do not define spin$^c$
bordism invariants.

\subsubsection*{Topologically trivial Ramond--Ramond fields}

The torus of topologically trivial fields
$\Omega(M;K^\bullet)^{j}/\Omega_\IZ(M;K^\bullet)^{j}$ consists of globally defined
differential forms on $M$, i.e. $[\check C]=[C]$ with
$C\in\Omega(M;K^\bullet)^j$; these are the Ramond--Ramond potentials that are normally considered in
the physics literature. The field strength of such a potential is
\beqa
F\big( [C]\big)=\dd C
\eeqa
and the gauge invariance is
\beqa
C \ \longmapsto \ C+\dd\xi \qquad \mbox{with} \quad
\xi\in\Omega(M;K^\bullet)^{j-1}\, \big/\, \Omega_\IZ(M;K^\bullet)^{j-1} \ .
\eeqa

\subsubsection*{Properties}

\begin{enumerate}
\item The differential K-theory is 2-periodic:
\beqa
\check\K{}^{j+2}(M) \cong \check\K{}^j(M) \ .
\eeqa
\item There is a graded-commutative ring multiplication
\beqa
\smile \,:\, \check\K{}^j(M)\otimes \check\K{}^{j'}(M) \ \longrightarrow \
\check\K{}^{j+j'}(M)
\eeqa
which acts on topologically trivial fields $C$ as
\beqa
[C]\otimes [\check C{}'\,] \ \longmapsto \ \big[C\wedge F\big([\check
C{}'\,]\big) \big] \ .
\eeqa
\item A \emph{$\check\K$-orientation} or \emph{smooth K-orientation} on a manifold $M$ is a
  choice of spin$^c$ structure together with a riemannian structure
  and a compatible smooth connection on the spin$^c$ bundle. For
  $M$ $\check\K$-oriented there is an integration map
\beqa
\int^{\check\K}\!\!\!\!\!\int_M \,:\, \check\K{}^j(M) \ \longrightarrow \
\check\K{}^{j-m}(\pt)
\eeqa
where $m=\dim(M)=10$. For topologically trivial fields $C$ it is given
by
\beqa
\int^{\check\K}\!\!\!\!\!\int_M \, [C] =u^{\lfloor m/2\rfloor}\, \Big[\, \int_M\,
C\wedge\widehat{A}(M)\, \Big]_{u^0} \quad \mbox{mod} \ \IZ \ .
\eeqa
The integration map commutes with the field strength map, if the
integration of forms is defined appropriately using the orientation curvature, in the sense
that the curvature of $\int^{\check\K}\!\!\!\int_M \, [\check C]$ is $\int_M\,
\widehat{A}(M)\wedge F([\check C])$ by the Riemann--Roch theorem.
\end{enumerate}

\subsubsection*{Examples}

Using the exact sequences from (\ref{diffKexactseqs}) one can work out
differential K-theory groups explicitly in a number of basic examples.
\begin{itemize}
\item When $M$ is a point one has
\beqa
\check\K{}^0(\pt)=\IZ \qquad \mbox{and} \qquad \check\K{}^{-1}(\pt) =
\bbt \ .
\eeqa
\item Let $R$ be a linear vector space over $\IC$. Then $R$ is
contractible, so one has $\RH(R;K^\bullet)^{-1} = 0 = \K^{-1}(R)$ and
$\K^0(R)=\IZ$. Whence the group of Type~IIA Ramond--Ramond potentials on
$R$ is given by
\beqa
\check\K{}^0(R)=\Omega_\IZ(R;K^\bullet)^0 \ \longrightarrow \
\IZ \, \oplus \, \Omega(R;K^\bullet)^{-1} \, \big/\, \Omega_\IZ(R;K^\bullet)^{-1} \ .
\eeqa
It naturally contains those fields which trivialize the Ramond--Ramond
currents sourced by the stable D0-branes of the Type~IIA theory,
corresponding to characteristic classes $[c]\in\K^0(R)=\IZ$. Since $R$ is contractible, the potential is determined in positive
degree by a
globally-defined differential form $C$ of curvature $F([C])= \dd C$, with
the gauge invariance $C\mapsto C+\dd\xi$. The arrow is then the
natural map which associates to the field strength $F$ the
corresponding globally well-defined Ramond--Ramond potential $C$. The group of Type~IIB
Ramond--Ramond potentials on $R$ is on the other hand given by
\beqa
\check\K{}^{-1}(R)= \Omega(R;K^\bullet)^{0} \, \big/\, \Omega_\IZ(R; K^\bullet)^{0} \ .
\eeqa
In this case there is no extension as the Type~IIB theory has no
stable D0-branes, and hence the Ramond--Ramond fields are determined
entirely by the potentials $C$ which are globally defined differential
forms of even degree.
\end{itemize}

\subsection{Models for differential K-theory}

Let us now look at some explicit models for the differential K-theory
groups, stressing the virtues and drawbacks of each approach from a
physical perspective. For brevity we consider only the degree zero
groups $\check\K{}^0(M)$ pertinent to the Type~IIA string theory.

\subsubsection*{Differential function spaces}

The foundational construction of differential K-theory is found
in~\cite{Hopkins:2002rd}; this approach is based on classifying maps
for complex K-theory. Let $\Fred$ denote the algebra of Fredholm
operators on a separable Hilbert space. Taking the index bundle of a
homotopy class of maps $c:M\to\Fred$ determines an isomorphism
\beq
[M,\Fred] \ \xrightarrow{\Index} \ \K^0(M) \ .
\label{FredK0}\eeq
The cocycles for $\check\K{}^0(M)$ are triples $(c,h,\omega)$, where
$c:M\to\Fred$ represents a K-theory class via the isomorphism
(\ref{FredK0}), $h\in C(M;K^\bullet)^{-1}$ is a cochain on $M$, and
$\omega\in\Omega_{\rm cl}(M;K^\bullet)^0$ is a closed differential form on $M$ of even
degree such that
$c^*{\tt u}-\delta h=\omega$ where the cocycle ${\tt u}\in Z(\Fred;K^\bullet)^0$
represents the Chern character of the universal vector bundle. One then
defines an equivalence relation by declaring two cocycles to be
equivalent ``up to homotopy''. Although powerful because of its generality,
this approach is not very useful for constructing the extra geometrical
ingredients required in gauge
theory; a graded-commutative product structure in this model is
described in~\cite{Upmeier}.

\subsubsection*{Chan--Paton gauge fields}

A more geometric approach is given
in~\cite{Freed:2000ta}, which may be thought of as equipping the
Chan--Paton vector bundles of background D-branes that wrap the whole spacetime $M$ with
connections, and then constructing Ramond--Ramond potentials in an
analogous way to that described in \S\ref{DBranes} In this model one represents classes in
$\check\K{}^0(M)$ by pairs $(E,\nabla)$, where $E\to M$ is a hermitian
vector bundle and $\nabla$ is a unitary connection on $E$. The cocycles are
then triples $(f,\eta,\omega)$, where $f:M\to BU$ is a classifying
map for $E$, $\omega=\ch(\nabla)$ is the Chern--Weil representative of
the Chern characteristic class $\ch([E])$, and $\eta$ is a Chern--Simons form
with the property
\beqa
\dd\eta=f^*\omega_{BU}-\omega \ ,
\eeqa
with $\omega_{BU}=\ch(\nabla_{BU})$ the Chern characteristic class of
the universal bundle over $BU$ with universal connection
$\nabla_{BU}$. A refinement of the index theorem in this model can be
found in~\cite{FreedLott}.

\subsubsection*{Configuration space and gauge transformations of Ramond--Ramond fields}

As we have explained, locality of quantum field theory requires the
use of cocycles rather than isomorphism classes. We will regard the
Ramond--Ramond gauge fields as objects in a certain monoidal category
$\check\scrK_M$ whose morphisms are gauge transformations. A
construction of a category $\check\scrK_M$ of differential K-cocycles is sketched
in~\cite{Freed:2006yc} whose isomorphism classes are precisely the
differential K-theory classes in the group
$\check\K{}^0(M)$, and which is a \emph{groupoid}; below we elaborate on
this description. We identify gauge
fields up to isomorphism of cocycles. The monoidal structure on
$\check\scrK_M$ is generated by the sum of cocycles.

The objects of the
category $\check\scrK_M$ are triples $\check C= (E,\nabla,C)$, where
$E\to M$ is a $\IZ_2$-graded hermitian
vector bundle, $\nabla$ is a unitary connection on $E$, and
$C\in\Omega(M;K^\bullet)^{-1}$ is a differential form on $M$ of odd
degree. The field strength of such an object is given by~\cite{FreedLott}
\beqa
F(E,\nabla, C)=\ch(\nabla)+\dd  C \ ,
\eeqa
where
\beqa
\ch(\nabla)= \Tr \exp\big(-\mbox{$\frac1{2\pi\ii}$}\,
u^{-1}\otimes\nabla^2\, \big) \ \in \
\Omega_{\rm cl}(M;K^\bullet)^0
\eeqa
is the corresponding Chern character form which gives the curvature
contribution from background D-branes. In the topologically
trivial case we take $E=\emptyset$ to be the empty vector bundle and the differential form $C$ is
what is normally called the Ramond--Ramond field (or more precisely
potential) of Type~IIA string theory. The monoidal structure on objects is given by
\beqa
\hat C + \hat C{}'= (E,\nabla,C)+(E',\nabla',C'\,)= (E\oplus
E',\nabla\oplus \nabla',C+C'\,) \ ,
\eeqa
and a zero object $\hat C=0$ is represented by $(E\oplus E^{\rm
  op},\nabla\oplus\nabla,0)$ for a $\IZ_2$-graded vector bundle $E$,
where $E^{\rm op}$ denotes the bundle $E$ with the opposite grading.

A pre-morphism $g: \check C_0\longrightarrow\check C_1$ between two
objects $\check C_0=(E_0,\nabla_0, C_0)$ and 
$\check C_1=(E_1,\nabla_1, C_1)$ is a triple
$\big(\,\widetilde{G}\,,\,\widetilde{\nabla} \,,\, { \lambda}\ \big)$,
where $\big(\,\widetilde{G}\,,\,\widetilde{\nabla} \, \big)$ is a $\IZ_2$-graded
vector bundle with connection on
$M\times I$, $I=[0,1]$ such that
\beqa
\big(\,\widetilde{G}\,,\,\widetilde{\nabla} \,
\big)\Big|_{M\times 0}= (E_0,\nabla_0) \qquad \mbox{and} \qquad
\big(\,\widetilde{G}\,,\,\widetilde{\nabla} \,
\big)\Big|_{M\times 1}= (E_1,\nabla_1) \ ,
\eeqa
with $\widetilde{F}=F\big(\,\widetilde{G}\,,\,\widetilde{\nabla} \,,\,
{ \lambda}\, \big)$ constant along $t\in[0,1]$,
i.e. $i_{\partial/{\partial t}}\widetilde{F}=0$, and
$\lambda\in\Omega(M;K^\bullet)^{0}$ such that
\beqa
 C_1=  C_0+ \CS(\nabla_0,\nabla_1) +\dd\lambda
\eeqa
is given by a transgression form obtained from a canonically defined
Chern--Simons class~\cite{Lott}
\beqa
\CS(\nabla_0,\nabla_1)=\int_0^1\, \ch\big(\, \widetilde{\nabla}\, \big) \ \in \
\Omega(M;K^\bullet)^{-1}\,\big/\, {\rm im}(\dd)
\eeqa
with the property
\beqa
\dd\CS(\nabla_0,\nabla_1)=\ch(\nabla_0)-\ch(\nabla_1) \ .
\eeqa
The relations are $E_2=E_1+E_3$ whenever there is a short exact
sequence of vector bundles
\beqa
0 \ \longrightarrow \ E_1 \ \xrightarrow{ \ i \ } \ E_2 \ \xrightarrow{
  \ j \ } \ E_3 \ \longrightarrow \ 0 \ .
\eeqa
Choosing a splitting $s:E_3\to E_2$ gives an isomorphism $i\oplus
s:E_1\oplus E_3\to E_2$. If $\nabla^{E_i}$ is a connection on $E_i\to
M$, then the corresponding relative Chern--Simons form
\beqa
\CS\big(\nabla^{E_1}\,,\,\nabla^{E_2}
\,,\,\nabla^{E_3}\big):=\CS\big((i\oplus s)^*\nabla^{E_2}\,,\,
\nabla^{E_1}\oplus\nabla^{E_3}\big) \ \in \
\Omega(M;K^\bullet)^{-1}\,\big/\, {\rm im}(\dd)
\eeqa
is independent of the choice of splitting morphism $s$ and satisfies
\beqa
\dd\CS\big(\nabla^{E_1}\,,\,\nabla^{E_2}
\,,\,\nabla^{E_3}\big) = \ch\big(\nabla^{E_2}\big) -
\ch\big(\nabla^{E_1}\big)- \ch\big(\nabla^{E_3}\big) \ .
\eeqa
We then set $C_2=C_1+C_3-\CS\big(\nabla^{E_1}\,,\,\nabla^{E_2}
\,,\,\nabla^{E_3}\big)$.

Two pre-morphisms
\beqa
(E_0,\nabla_0, C_0)\xrightarrow{(\,\widetilde{G}_0,\widetilde{\nabla}_0
  ,{ \lambda}_0 \, )}(E_1,\nabla_1, C_1) \qquad \mbox{and} \qquad
(E_0,\nabla_0, C_0)\xrightarrow{(\,\widetilde{G}_1 ,\widetilde{\nabla}_1
  , { \lambda}_1 \, )}(E_1,\nabla_1, C_1)
\eeqa
are said to
be equivalent if there exists a triple $\big(\mbf{G},
\mbf{\nabla} , \mbf{ \lambda})$, where $\big(\mbf{G},
\mbf{\nabla} )$ is a $\IZ_2$-graded vector bundle with connection on $M\times
I\times I$
such that
\beqa
(\mbf{G},
\mbf{\nabla}) \big|_{M\times
  0\times I}= \big(\,\widetilde{G}_0\,,\,\widetilde{\nabla}_0 \big) \qquad \mbox{and} \qquad (\mbf{G},
\mbf{\nabla}) \big|_{M\times
  I\times 0} = (E_0\times I,\nabla_0\times 1)
\eeqa
while
\beqa
(\mbf{G},
\mbf{\nabla})\big|_{M\times 1\times I} = 
\big(\,\widetilde{G}_1\,,\,\widetilde{\nabla}_1 \big) \qquad \mbox{and} \qquad (\mbf{G},
\mbf{\nabla}) \big|_{M\times
  I\times 1} = (E_1\times I,\nabla_1\times 1) \ ,
\eeqa
with $\mbf\lambda\in\Omega(M;K^\bullet)^{-1}/\,{\rm im}(d)$ such that
\beqa
\lambda_1= \lambda_0+\int_0^1\, \ch(\mbf\nabla)
+\dd\mbf\lambda \ .
\eeqa
This generates an equivalence relation. Note that because of the
geometric product structure of $(\mbf G,\mbf\nabla)$, one has $\dd
\CS\big(\, \widetilde\nabla_0\,,\,\widetilde\nabla_1
\, \big) =
\ch\big(\, \widetilde\nabla_0\, \big)-\ch\big(\,
\widetilde\nabla_1\, \big)$ and hence
\beqa
\int_0^1\, \ch\big(\, \widetilde{\nabla}_0\, \big) +\dd\lambda_0 = \int_0^1\,
\ch\big(\, \widetilde{\nabla}_1\, \big) +\dd\lambda_1 \ .
\eeqa
A \emph{gauge transformation} from $\check C_0$ to $\check C_1$ is an
equivalence class of pre-morphisms
$\big(\,\widetilde{G}\,,\,\widetilde{\nabla} \,,\, {
  \lambda}\, \big)$. We refer to the
shift of $C_0$ by $\dd\lambda$ as a \emph{small} gauge transformation,
while the shift by the transgression form is called a \emph{large} gauge
transformation. Gauge transformations leave the corresponding
Ramond--Ramond field strengths invariant, $F([{\check C_0}]) =F([{\check
  C_1}] )$.

Next we construct the composition of gauge transformations. The
composition of pre-morphisms 
\beqa
\check C_0=(E_0,\nabla_0, C_0) \ \xrightarrow{(\,\widetilde{G}_1,\widetilde{\nabla}_1
  , { \lambda}_1\, )}  \ \check C_1= (E_1,\nabla_1, C_1) 
\eeqa
and
\beqa
\check C_1= (E_1,\nabla_1, C_1) \ \xrightarrow{(\,\widetilde{G}_2 ,\widetilde{\nabla}_2
  , { \lambda}_2\, )} \ \check C_2= (E_2,\nabla_2, C_2)
\eeqa 
is the pre-morphism $\check
C_0\xrightarrow{(\,\widetilde{G},\widetilde{\nabla}
  , { \lambda} \, )}\check C_2$ defined as follows. There is a graded vector
bundle with connection on $M\times I$ defined by
\beqa
\big(\,\widetilde{G}\,'\,,\,\widetilde{\nabla}' \,
\big) := \big(\,\widetilde{G}_1 \oplus \widetilde{G}_2
\,,\,\widetilde{\nabla}_1\oplus \widetilde{\nabla}_2 \,
\big)
\eeqa
with
\beqa
C_2=C_0+\int_0^1\,\ch\big(\,\widetilde\nabla'\, \big)
+\dd(\lambda_1+\lambda_2) \ .
\eeqa
Via concatenation of
paths, we now set $\big(\,\widetilde{G}\,,\,\widetilde{\nabla} \, \big)\big|_{M\times t}$ equal to
$\big(\,\widetilde{G}_1\,,\,\widetilde{\nabla}_1 \, \big)\big|_{M\times (2t)}$ for  $0\leq t\leq\frac12$ and
to $\big(\,\widetilde{G}_2\,,\,\widetilde{\nabla}_2 \, \big)\big|_{M\times(2t-1)} $ for $\frac12\leq t\leq1$.
This is compatible with the required identity
\beqa
\CS(\nabla_0,\nabla_2)=
\int_0^1\, \ch\big(\,\widetilde\nabla\,
\big)=\CS(\nabla_0,\nabla_1)+\CS(\nabla_1,\nabla_2) \ .
\eeqa
Likewise, we use concatenation of paths to define the graded vector bundle with connection $(\mbf E,\nablatr)$ on
$M\times I\times I$ having the three boundary faces
\beqa
\big((E_0,\nabla_0)\,,\,(E_1,\nabla_1) \big) \ , \qquad 
\big((E_1,\nabla_1)\,,\,( E_2,\nabla_2) \big) \qquad  \mbox{and} \qquad
\big((E_2,\nabla_2)\,,\, (E_0,\nabla_0) \big) \ .
\eeqa
One shows that
$\dd\int_0^1\,
\ch(\nablatr)=\int_0^1\,\big(\ch(\,\widetilde\nabla_1\,)+\ch(\,\widetilde\nabla_2\,)
-\ch(\,\widetilde\nabla'\,) \big)$.
We set
\beqa
\lambda=\lambda_1+\lambda_2+\int_0^1\,
\ch(\nablatr) +\dd\kappa
\eeqa
for some $\kappa\in\Omega(M;K^\bullet)^{-1}/\,{\rm im}(\dd)$.

\begin{lemma}
The composition $\big(\,\widetilde{G}_1\,,\,\widetilde{\nabla}_1 \,,\,
{ \lambda}_1\, \big)\circ
\big(\,\widetilde{G}_2\,,\,\widetilde{\nabla}_2 \,,\, {
  \lambda}_2 \, \big):= \big(\,\widetilde{G}\,,\,\widetilde{\nabla} \,,\, { \lambda}\, \big)$
in
$\check\scrK_M$ is well-defined and associative.
\end{lemma}
\begin{proof}
We show that the composition is well-defined. Suppose, for
example, that the gauge transformation $g:\check C_0\longrightarrow\check C_1$ is
represented by equivalent pre-morphisms $\big(\,\widetilde{G}_1\,,\,\widetilde{\nabla}_1 \,,\,
{ \lambda}_1\, \big) $ and
$\big(\,\widetilde{G}^{\,\prime}_1\,,\,\widetilde{\nabla}^{\, \prime}_1 \,,\,
{ \lambda}^{\prime}_1\, \big) $, implemented by a triple
$(\mbf{G},
\mbf{\nabla} , \mbf{ \lambda }) $. Set $
\big(\,\widetilde{G}^{\,\prime}\,,\,\widetilde{\nabla}^{\,\prime}
\, \big) =\big(\,\widetilde{G}_1^{\, \prime}\,,\,\widetilde{\nabla}^{\, \prime}_1 \, \big)\circ
\big(\,\widetilde{G}_2\,,\,\widetilde{\nabla}_2 \, \big)$ and
$\lambda'=\lambda_1'+\lambda_2+\int_0^1\,
\ch(\,\nablatr'\,)+\dd(\kappa- \mbf\lambda)$, where the vector
bundle
with connection $(\mbf E',\nablatr'\,)$ on $M\times I\times I$ is
obtained from $(\mbf E,\nablatr)$ by replacing $\widetilde G_1$ with
$\widetilde G_1^{\,\prime}$. Then the condition
$\lambda_1^{\prime}= \lambda_1+\int_0^1\,\ch(\mbf\nabla)
+\dd\mbf\lambda$ immediately implies that $\big(\,\widetilde{G}\,,\,\widetilde{\nabla} \,,\, { \lambda}\, \big)$ and $
\big(\,\widetilde{G}^{\,\prime}\,,\,\widetilde{\nabla}^{\,\prime}
\,,\, { \lambda}^{\prime}\, \big)$ determine equivalent gauge transformations. We
leave the check of associativity to the interested reader.
\end{proof}

The composition law above defines a group structure on the set of all
gauge transformations on $\check\scrK_M$. The identity morphism is given by
$\unit_{(E,\nabla,C)}:(E,\nabla,C)\xrightarrow{(E\times I,\nabla\times1,0)}
(E,\nabla,C)$, with $\CS(\nabla,\nabla)=\int_0^1\, \ch(\nabla\times1)=0$. The inverse of a morphism
$(E,\nabla,C)\xrightarrow{(\,\widetilde{G},\widetilde\nabla,\lambda
  \ )}(E',\nabla',C'\,)$ is given by
$(E',\nabla',C'\,)\xrightarrow{(\,\widetilde{G}{}^{\,{\rm op}},\widetilde\nabla,-\lambda
  \, )} (E,\nabla,C)$. To see
this, we must show that the composition
$\big(\,\widetilde{G}\,'\,,\,\widetilde\nabla\,' \,,\,\lambda'
  \, \big)$ of these two morphisms is
equivalent to the identity $\unit_{(E,\nabla,C)}$. Consider the morphism
$(E',\nabla',C'\,)\xrightarrow{(\,\widetilde{G},\widetilde\nabla,\lambda
  \, )^{-1}} (E,\nabla,C)$ defined via path inversion
\beqa
\big(\,\widetilde{G}\,,\,\widetilde\nabla\,,\,\widetilde\lambda
  \ \big)^{-1}\Big|_{M\times t}:= \big(\,\widetilde{G}\,,\,\widetilde\nabla\,,\,\widetilde\lambda
  \ \big)\Big|_{M\times(1-t)}
\eeqa
for all $t\in I$. Note that $\big(\,\widetilde{G}\,'\,,\,\widetilde\nabla\,' \,,\,\widetilde\lambda\,'
  \, \big)\big|_{M\times t}$ is equal to  $ \big(\,\widetilde{G}\,,\,\widetilde\nabla\,,\,\widetilde\lambda
  \ \big)\big|_{M\times(2t)}$ for $0\leq t\leq\frac12$ and to $\big(\,\widetilde{G}\,,\,\widetilde\nabla\,,\,\widetilde\lambda
  \ \big)\big|_{M\times2(1-t)}$ for $\frac12\leq t\leq1$, so that
\beqa
\int_0^1\, \ch\big(\, \widetilde\nabla\,'\, \big) = 0 \ .
\eeqa
Furthermore, in this case we have
\beqa
\dd\int_0^1\, \ch(\nablatr) = \int_0^1\,\big( \ch(\, \widetilde
\nabla\, )-\ch(\, \widetilde\nabla\, )-\ch(\nabla\times 1)\big) = 0
\eeqa
and 
\beqa
\lambda-\lambda+\int_0^1\, \ch(\nablatr) = \int_0^1\, \ch(\nablatr) \ .
\eeqa
This shows that every morphism is invertible. 

This construction defines the \emph{category $\check\scrK_M$ of
  Ramond--Ramond fields} on the manifold $M$, which is a strictly
symmetric monoidal category; it defines an action groupoid and has a
physical interpretation in boundary string field theory, generalized
to incorporate superconnections~\cite{Freed:2006yc}. The objects are the gauge
fields $\check C=(E,\nabla,C)$, while the morphisms are the gauge
transformations $\big(\,\widetilde
G\,,\,\widetilde\nabla\,,\,\lambda\, \big)$ together with the group law
described above. The set of gauge equivalence classes is the group of
isomorphism classes which coincides with the differential K-theory
$\check\K{}^0(M)$. Moreover, we see that gauge transformations
generate the differential K-theory group $\check\K{}^1(M)$; this
follows from the description of $\check\K{}^1(M)$ given
in~\cite{Freed:2006yc} by integrating classes in
$\check\K{}^0(M\times S^1)$, which are trivial at each point of $S^1$,
over $S^1$.

Every object of the cocycle category of fields $\check\scrK_M$ is invertible: The inverse of
the gauge field $\check C=(E,\nabla,C)$ is $\check
C^{-1}=(E^{\rm op},\nabla,-C)$. Thus $\check\scrK_M$ is a Picard category. The zero
object (of the monoidal structure) has non-trivial
automorphisms; from the above construction it follows that the
corresponding automorphism group coincides with the group of flat
fields
\beqa
\Aut_{\check\scrK_M}(0)\cong\K^0(M;\torus) \ .
\eeqa

The categorical structure of the gauge theory configuration space may
be iterated to define higher categories $\check\scrK_M^k$ in a similar
manner. Let us describe explicitly the next member $\check\scrK_M^2$ in the
multi-categorical hierarchy. In the same way that we went from the
definition of the differential K-theory to the category
$\check\scrK_M=\check\scrK_M^1$, we now replace the equivalence relation on
(1-)morphisms by 2-morphisms. Then we lose the strict notion of
composition of 1-morphisms, as explained in~\S\ref{Confspace} Since every gauge transformation has an
inverse, we can consider the composition as a certain subset
$\Comp^1\subset \Hom_{\check\scrK^1_M}(\check C_0,\check C_1)\times
\Hom_{\check\scrK^1_M}(\check C_1,\check C_2)\times \Hom_{\check\scrK_M^1}(\check
C_2,\check C_0)$ such that $(g_1,g_2,g_3)\in \Comp^1$ if and only if $g_1\circ
g_2\circ g_3=\unit_{\check C_0}$. On the other hand, it is no longer true in general that $g_1$ and $g_2$ determine $g_3$ uniquely for the subset
$(g_1,g_2,g_3)\in \Comp^2\subset \Hom_{\check\scrK^2_M}(\check C_0,\check C_1)\times
\Hom_{\check\scrK^2_M}(\check C_1,\check C_2)\times \Hom_{\check\scrK_M^2}(\check
C_2,\check C_0)$.

We now explain some details of the construction. The objects of the
category $\check\scrK_M^2$ are the same as those for $\check\scrK_M^1$. Every
object $\check C$ has a canonical inverse $\check C^{-1}$. As in
$\check\scrK_M^1$, a morphism $\check C_0\longmapsto \check C_1$ is the same thing
as a morphism $\check C_0^{-1}+\check C_1\longmapsto 0$. Therefore, we need
only describe morphisms to the zero object. A morphism
$g:(E,\nabla,C)\longmapsto0$ is given by a triple $\big(\,\widetilde
G\,,\,\widetilde\nabla\,,\,\lambda \, \big)$, with the pair $\big(\,\widetilde
G\,,\,\widetilde\nabla \, \big)$ on $M\times I$ such that
$\big(\,\widetilde
G\,,\,\widetilde\nabla \,
\big)\big|_{M\times0}=(E,\nabla)$ and $\big(\,\widetilde
G\,,\,\widetilde\nabla \, \big)\big|_{M\times1}= 0$,
while $-C+\dd\lambda=\int_0^1\, \ch(\, \widetilde\nabla\, )$. Let
$g':(E,\nabla,C)\xrightarrow{(\,\widetilde
  G\,',\widetilde\nabla\,',\lambda'\, )} 0$ be another
morphism. A 2-morphism $\alpha:g\Longrightarrow g'$ is given by a triple
$(\mbf G,\mbf\nabla,\mbf\lambda)$, where $(\mbf G,\mbf\nabla)$ is
graded vector bundle with connection on $M\times I\times I$ such that
$(\mbf G,\mbf\nabla)\big|_{M\times 0\times I}= \big(\,\widetilde
G\,,\,\widetilde\nabla\, \big)$,  $(\mbf
G,\mbf\nabla,\mbf\lambda)\big|_{M\times I\times0} = (E\times
I,\nabla\times1)$ and $(\mbf
G,\mbf\nabla)\big|_{M\times 1\times I}= \big(\,\widetilde
G\,'\,,\,\widetilde\nabla\,' \, \big)$,  $(\mbf
G,\mbf\nabla)\big|_{M\times I\times1} = 0$, and $\mbf\lambda\in
\Omega(M;K^\bullet)^{-1}/\, {\rm im}(d)$ such that
$\lambda'= \lambda+\int_0^1\,\ch(\mbf\nabla)+\dd\mbf\lambda$. 

We declare a pair of 2-morphisms $\alpha_0,\alpha_1:g\Longrightarrow g'$ to be
equivalent if there is a triple $\big(\,\widetilde{\mbf
  G}\,,\,\widetilde{\mbf\nabla}\,,\, \widetilde{\mbf\lambda}\,\big)$,
where $\big(\,\widetilde{\mbf
  G}\,,\,\widetilde{\mbf\nabla}\, \big)$ is a $\IZ_2$-graded vector
bundle with connection
on $M\times I\times I\times I$ having the boundary values
\beqa
\big(\,\widetilde{\mbf G} \,,\,
\widetilde{\mbf\nabla} \, \big) \big|_{M\times
  0\times I\times I}= (\mbf{G}_0 ,\mbf{\nabla}_0) \quad &
\mbox{and} &
\quad \big(\,\widetilde{\mbf G}\,,\,
\widetilde{\mbf\nabla} \, \big) \big|_{M\times
  1\times I\times I}= (\mbf{G}_1 ,\mbf{\nabla}_1) \ , \nonumber\\[4pt]
\big(\,\widetilde{\mbf G}\,,\,
\widetilde{\mbf\nabla} \, \big) \big|_{M\times
  I\times 0\times I}= \big(\,
\widetilde{G}_0\times I \,,\,\widetilde{\nabla}_0\times1\, \big) \quad
& \mbox{and} &
\quad \big(\,\widetilde{\mbf{G}}\,,\,
\widetilde{\mbf\nabla} \, \big) \big|_{M\times
  I\times 1\times I}= \big(\, \widetilde{G}_1\times
I\,,\,\widetilde{\nabla}_1\times 1\, \big)  \ , \nonumber\\[4pt]
\big(\,\widetilde{\mbf G}\,,\,
\widetilde{\mbf\nabla} \, \big) \big|_{M\times
  I\times I\times 0}= (E\times I\times I,{\nabla}\times1\times1) \quad
& \mbox{and} &
\qquad \big(\,\widetilde{\mbf G}\,,\,
\widetilde{\mbf\nabla} \, \big) \big|_{M\times
  I\times I\times 1}= 0 \ ,
\eeqa
and $\widetilde{\mbf\lambda}\in\Omega(M;K^\bullet)^{0}/\, {\rm im}(d)$ such that
$\mbf\lambda_1-\mbf\lambda_0-\int_0^1\,
\ch(\,\widetilde{\mbf\nabla}\,) =\dd\widetilde{\mbf\lambda}$.

Composition of 2-morphisms is now defined in the same manner as
composition for 1-morphisms in the preceeding level of the
hierarchy. The set $\Comp^2\subset \Hom_{\check\scrK^2_M}(\check C_0+\check
C_1^{-1},0)\times \Hom_{\check\scrK^2_M}(\check C_1+\check C_2^{-1},0)\times
\Hom_{\check\scrK_M^2}(\check C_2+\check C_0^{-1},0)$ now consists of all
triples $(g_1,g_2,g_3)$ such that there exists a 2-morphism $g_1\circ g_2\circ
g_3\Longrightarrow \Symm$, where $\Symm$ represents the symmetric 2-morphisms. One
can show that all 2-morphisms are isomorphisms. The set of
isomorphism classes in $\check\scrK_M^2$ coincides with the set of morphisms
in $\check\scrK_M^1$. The composition of 2-morphisms is associative. The
composition of 1-morphisms is associative up to 2-morphisms. One
can also replace cylinders in the construction of $\check\scrK^k_M$ by
general bordisms. 

\subsubsection*{Geometric cocycles}

The index theory model of~\cite{BunkeSchick2007} readily
permits all constructions required in gauge theory, at the price of
introducing a very large configuration space, containing broad
classes of fields, some of which have no
interpretation in terms of D-branes wrapping cycles. Let $\pi:E\to M$
be a proper submersion with closed fibres and even-dimensional vertical bundle
$T^v\pi:=\ker(\dd\pi)$. Choose a
fibrewise riemannian metric on $T^v\pi$, and a complement $T^h\pi\subset TE$ which
defines a horizontal distribution. We pick an orientation on $T^v\pi$, and also a family of Dirac bundles
over $E$, i.e. a $\IZ_2$-graded hermitian vector bundle with
connection $(V,\nabla)$ on $E$ and Clifford multiplication
${\tt c} : T^v\pi\otimes V\to V$. In~\cite{Bunke2009,BunkeSchick2007} this collection of
data was subsumed into the notion of a geometric family $\bun$. A cocycle for a differential K-theory class in $\check\K{}^0(M)$ is a pair
$(\bun,\xi)$, where $\bun$ is a geometric family and $\xi \in
\Omega(M;K^\bullet)^{-1} /\, {\rm im}(\dd)$ is a class of differential
forms on $M$ of odd degree; the equivalence relations
on cocycles can be found in~\cite{BunkeSchick2007}. These classes are
extensions of those defined above in terms of vector bundles with
connection, to which they reduce when $\pi=\Id_M$ has fibre consisting
of just a point; within this model lie the well-developed and very
powerful techniques of local index theory, whose properties can be
used as a ``black box'' to efficiently carry out all constructions. 
The index of the family of Dirac operators $\Dirac(\bun)$ on a geometric
family $\bun$ over $M$ can be naturally considered as
an element of K-theory $\Index(\bun)\in \K^0(M)$; it defines the
characteristic class map $c:\check\K{}^0(M)\to \K^0(M)$ as
$$c\big([\bun,\xi] \big):=\Index(\bun) \ . $$
For a geometric family $\bun$, the local index form $\Omega(\bun)\in\Omega(M;K^\bullet)^0$~\cite{Bunke2009,BunkeSchick2007} is
the adiabatic limit of local traces of the heat kernel of a Bismut superconnection on the associated
Hilbert bundle $H(\bun)\to M$ with fibres $H_x:={\rm
  L}^2(E_x;V|_{E_x})$ for $x\in M$; it provides a canonical and explicit de~Rham
representative for the Chern character of the index of $\bun$ through
the index theorem for families~\cite{BGV} which reads $\ch\big(\Index(\bun)\big) = \big[\Omega(\bun)\big]$. 
The field strength map $F:\check\K{}^0(M)\to \Omega_\IZ(M;K^\bullet)^0$ is then given by 
$$F\big( [\bun,\xi] \big) :=\Omega(\bun)-\dd\xi \ . $$
A refinement of the Chern character homomorphism between differential
K-theory and differential cohomology in this model can be found
in~\cite{BunkeSchick2007}.

We will now explain how topologically non-trivial quantized Ramond--Ramond fields
naturally fit into this framework as cocycles in the absence of
D-brane sources, i.e. as induced solely by the closed string
background $(M,g)$. For this, we will give a physical interpretation of this description in
terms of brane-antibrane annihilation in boundary string field
theory by presenting a special class of cocycles which
exhibit the salient features of the construction of the
Ramond--Ramond field associated to a K-theory class in (\ref{RRfieldKclass}). In this setting a
{Ramond--Ramond field} on the manifold $M$ is taken to be a
pair $\check{C}=(\bun,-C)$, where $\bun$ is a
geometric family over $M$ and $C\in\Omega(M;K^\bullet)^{-1}$ (not taken modulo
${\rm im}(\dd)$). In the topologically trivial case one sets $\bun=\emptyset$
and the differential form $C$ is what is usually called the
Ramond--Ramond field of the Type~IIA theory.

Consider first the
Ramond--Ramond field $\check C=(\Vcal,0)$ for the geometric family
$\Vcal$ with underlying fibre bundle $\pi=\Id_M:M\to M$
having zero-dimensional vertical bundle. The $\IZ_2$-graded
bundle $V=V^+\oplus V^-\to M$ represents $\Index(\Vcal)=[V]:=[V^+]-[V^-]\in
\K^0(M)$. The geometric K-homology class
$$
[M,V^+,\Id_M]-[M,V^-,\Id_M]
$$
represents a brane-antibrane pair filling $M$ with Chan--Paton
bundles $(V^+,V^-)$~\cite{Reis:2005pp}. The connection $\nabla$ and
hermitian structure of the family of Dirac bundles $\Vcal$ is the extra dynamical
information on the D-branes encoded in the boundary string field
theory, which naturally defines an element of differential
K-theory. Every class in $\K^0(M)$ (and hence every D-brane on $M$)
can be realized via this construction as the index of a geometric
family~\cite{BunkeSchick2007}. Using the explicit expression
for the local index form $\Omega(\Vcal)$ in this
case~\cite{Bunke2009}, the corresponding Ramond--Ramond field strength is given by
\beqa
F(V) = \ch(V)\wedge\sqrt{\widehat{A}(M)} \ .
\eeqa
This construction thus reproduces the Moore--Witten
derivation~\cite{Moore2000} for the Ramond--Ramond field strength
(\ref{RRfieldKclass}) associated to the K-theory element
$[V]\in\K^0(M)$ classifying the background brane-antibrane system
wrapping $M$.

As in the case of geometric K-homology~\cite{Reis:2005pp}, the
inclusion of more general families $\check C=(\bun,-C)$ extends this
description to include non-representable brane-antibrane pairs filling
$M$ which are represented by the K-homology classes
\beqa
\big[E\,,\,\Index(\bun)\,,\,\pi\big] \ ,
\eeqa
where $\pi:E\to M$ is the underlying fibre bundle of $\bun$ and
$\Index(\bun)\in\K^0(M)$. This modifies the associated Ramond--Ramond
field strength to
\beqa
F\big(\Index(\bun)\big)=\Big(\, \int_{E/M}\, \ch(W)\wedge
\widehat{A}(T^v\pi) +\dd C \, \Big)\wedge \sqrt{\widehat{A}(M)} \ ,
\eeqa
where we have decomposed the family of Dirac bundles into the spinor
bundle of $T^v\pi$ as $V=\cals(T^v\pi)\otimes W$
for a twisting bundle $W\to E$ with metric and compatible connection.
Here the fibrewise integral is the curvature generated by the
(non-representable) background D-branes
given by the local index form
$\Omega(\bun)\in\Omega(M;K^\bullet)^0$, which depends only on the
geometric family $\bun$, while $\dd C$ is the contribution of a
topologically trivial Ramond--Ramond field, which as such is not
sourced by any D-branes.

\subsubsection*{Holonomy on D-branes}

In analogy to the model of differential cohomology provided by
Cheeger--Simons differential characters, which are $\uo$-valued
homomorphisms on the group of smooth cycles in $M$, one can define the
differential K-theory $\check\K{}^\bullet(M)$ as a group of
$\uo$-valued homomorphisms on the set of Baum--Douglas K-cycles for
geometric K-homology; these maps are called differential characters
for K-theory
in~\cite{Benamour} and are interpreted as holonomies on D-branes
in~\cite{Reis:2006th}. They are characterized by their restrictions to
boundaries of K-chains (\ref{Kchainbdry}), which are given by pairing
a certain differential form $\omega|_{\widetilde{W}}$ with the
index density $\ch\big(\,\widetilde{E}\,\big)\wedge\widehat{A}\big(\, \widetilde{W}\,\big)$.

As an explicit example, define the reduced eta-invariant
of a K-chain with boundary (\ref{Kchainbdry}) by
\beqa
\Xi\big(\,\widetilde{W}\,,\,\widetilde{E}\,,\,\widetilde{f}\
\big)=\mbox{$\frac12$} \,\Big(\dim\big({\mathcal{H}}^W_E \big)+
\eta\big(\Dirac_E^W\big) \Big) \ \in \ \IR/\IZ \ ,
\eeqa
where $\hil_E^W$ is the space of harmonic $E$-valued spinors on $W$,
and $\eta\big(\Dirac_E^W\big)$ is the spectral asymmetry of the
$E$-twisted Dirac operator $\Dirac_E^W$ on $W$ which is the meromorphic
continuation at $s=0$ of the absolutely convergent series
\beqa
\eta\big(s,\Dirac_E^M\big)=
\sum_{\lambda\in{\rm spec}^0(\Dirac_E^M)-0} \ 
\frac\lambda{|\lambda|^{s+1}}
\eeqa
for $s\in\IC$ with $\Re(s)\gg0$, with the sum taken over the spectrum of the
closure of $ \Dirac_E^M$ which is the bounded Fredholm operator
$\Dirac_E^M\, \big(1+(\Dirac_E^M)^2\big)^{-1/2}$. The map $\Xi$
respects disjoint union, direct sum and Baum--Douglas vector bundle
modification of K-chains, but \emph{not} spin$^c$ bordism~\cite{Benamour}. Then
the \emph{holonomy} of the flat D-brane background defined by
(\ref{virtualKchain}) is given by
\beqa
\Omega\big(\,\widetilde{W}\,,\,\widetilde{\xi}\,,\,
\widetilde{f}\,\big)=\exp\Big(2\pi\ii\big(\Xi(\,\widetilde{W},\widetilde{f^*E_0},
\widetilde{f}\,)-\Xi(\,\widetilde{W},\widetilde{f^*E_1},
\widetilde{f}\,)\big) \Big) \ \in \ \uo \ .
\eeqa

\subsection{Self-dual field theories\label{Selfduality}}

We will now formulate the self-duality property of Ramond--Ramond fields
more precisely. We do this first in the more general setting of
\S\ref{Gendiffcoh}

\begin{definition}
A generalized cohomology theory $\E^\bullet$ is \emph{Pontrjagin
  self-dual} if there exists a ``shift'' $s\in\IZ$ such that
\beqa
\E^\bullet(M)\cong \Hom_{\IZ}\big(\E_{\bullet+s}(M)\,,\,\IZ\big)
\eeqa
for all spaces $M$.
\label{Pontrjagindef}\end{definition}

In Def.~\ref{Pontrjagindef}, $\E_\bullet$ is the homology theory obtained from
$\E^\bullet$ through its \emph{spectrum} $\{\cale_k\}$ such that
$\E^k(M)$ is the set of homotopy classes of maps $M\to \cale_k$; there
is a homotopy equivalence between $\cale_k$ and the based loop space
$\Omega\cale_{k+1}$. The $\E$-homology is given by the directed limit
\beqa
\E_k(M)=\lim_{\stackrel{\scriptstyle \longrightarrow}{r}}\,
\pi_{k+r}\big(M_+\wedge\cale_k \big) \ ,
\eeqa
where $M_+=M\sqcup m_0$ is the one-point compactification of $M$ by
a fixed base point $m_0\in M$, and $X\wedge Y= X\times Y\big/(X\times
y_0\sqcup x_0\times Y)$ is the smash product of locally compact
spaces. Def.~\ref{Pontrjagindef} is equivalent to the statement that
for each $k\in\IZ$ the natural pairing of real vector spaces
\beq
\check\iota\,:\, E^{k-s}\otimes E^{-k} \ \longrightarrow \ \IR
\label{iotaEd}\eeq
is non-degenerate, where $E^k:=\E^k(\pt)\otimes_\IZ \IR$.

\begin{definition}
A \emph{self-dual} generalized abelian gauge theory on a compact
$n$-dimensional smoothly $\E$-oriented manifold $N$ consists of a
Pontrjagin self-dual multiplicative cohomology theory $\E^\bullet$
with shift $s\in\IZ$, and its associated configuration space of gauge fields
$\check\E{}^d(N)$ for some $d\in\IZ$, together with a natural isomorphism
\beq
\Psi^{-1}\,:\, \E^d(N) \ \xrightarrow{ \ \approx \ } \ \E^{n-s+1-d}(N) \ .
\label{phiEd}\eeq
\end{definition}

Let $N$ be an $\E$-oriented riemannian manifold of dimension $n$, and
$M=\IR\times N$ the associated lorentzian spacetime of signature
$(1,n)$. The \emph{self-duality equations} for the gauge field
$F\in\Omega(M;E^\bullet)^d$ read
\beq
\dd F=0 \qquad \mbox{and} \qquad \Psi^{-1}(F)=\check\iota(\star F) \ \in \
\Omega(M;E^\bullet)^{n-s+1-d} \ ,
\label{selfdualityeqsEd}\eeq
where the first equation is the Bianchi identity while the second
equation is the self-duality condition. Here
\beqa
\Psi^{-1}\,:\, \Omega^k(M;E^{d-k}) \ \longrightarrow \
\Omega^k(M;E^{n-s+1-d-k})
\eeqa
is induced by the isomorphism (\ref{phiEd}), the map
\beqa
\star\,:\, \Omega^{n+1-k}(M;E^{d+k-n-1}) \ \longrightarrow \
  \Omega^k\big(M\,; \, (E^{d+k-n-1})^*\big)
\eeqa
is the lorentzian Hodge duality operator, and
\beqa
\check\iota\,:\, \Omega^k\big(M\,;\, (E^{d+k-n-1})^*\big) \
\longrightarrow \ \Omega^k(M;E^{n-s+1-d-k})
\eeqa
is induced by the pairing (\ref{iotaEd}). The equations
(\ref{selfdualityeqsEd}) define a first order linear hyperbolic
differential equation, so a solution is determined at any fixed time
$t\in\IR$; whence the space of solutions is isomorphic to the real
vector space $\Omega_{\rm cl}(N;E^\bullet)^d$. 

Note that the \emph{classical flux} $[F]_{\rm dR}$
defines a map $\Omega_{\rm cl}(N;E^\bullet)^d\to
\RH(N;E^\bullet)^d$. In the semi-classical
theory with Dirac charge quantization, the gauge field is a geometric
representative of a class in $\check\E{}^d(M)$; the space of classical
solutions on $M$ is then the differential cohomology group
$\check\E{}^d(N)$.

Incorporating sources into a self-dual gauge theory further requires
an isomorphism between electric and magnetic currents $j_e$ and $j_m$, as well as a
quadratic refinement of the bilinear pairing $\check\E{}^d(N)\otimes
\check\E{}^{n-s+1-d}(N)\to \torus$ between the corresponding
differential cocycles $\check j_e$ and $\check j_m$.

\subsubsection*{Type~II Ramond--Ramond fields}

In our main application, we take $n=9$ and $\E^\bullet=\K^\bullet$ to
be complex K-theory, so that $s=0$. A Ramond--Ramond field on a
compact riemannian spin manifold $N$ has a gauge equivalence in the
differential K-theory $\check\K{}^j(N)$, where $j=0$ for Type~IIA
string theory and $j=-1$ for Type~IIB. The K-theory of a point is
given by the Laurent polynomial ring $\K^\bullet(\pt)\cong \IZ(u)$, where $u$ has degree two and
the dual involution maps $u\mapsto u^*:=u^{-1}$. The automorphism
$\Psi^{-1}$ is the Adams operation on K-theory which acts as
complex conjugation, with $u\mapsto -u$. The lift of the Adams operation
$\check\Psi{}^{-1}$ to differential K-theory is then given
by
\beq
\check\Psi{}^{-1}\big([\check C]\big) = u^\ell \ \overline{[\check C]}
\ ,
\label{Adamsop}\eeq
where $\ell=5$ for Type~IIA and $\ell=6$ for Type~IIB, and if $[\check
C]\in\check\K{}^j(N)$ is represented by a complex vector bundle $E\to
N$ with connection $\nabla$, then the class $\overline{[\check C]}$ is
represented by the complex conjugate bundle $\overline{E}\to N$ with
conjugate connection $\overline{\nabla}$; a model independent
construction of all Adams operations $\check\Psi{}^k$, $k\in\IZ$, is
given in~\cite{Bunke2009b}. The self-duality equations
for the field strengths $F\in\Omega\big(\IR\times N,K^\bullet \big)^j$ are
then as described in~\S\ref{RRGT}

\bigskip

\section{Quantization of generalized abelian gauge
  fields\label{HamQuant}}

\subsection{Quantum actions and partition functions\label{Pathquant}}

We begin with some general remarks about the approach to the
quantization of abelian gauge theories that
we shall pursue. Recall that a generalized abelian gauge field on a
manifold $M$ is an
object of a suitable groupoid $\check\scrE{}^d(M)$ whose isomorphism
class sits in the generalized differential cohomology group
$\check\E{}^d(M)$; this semi-classical quantization of the gauge
theory leads to integrality of coupling constants and secondary
invariants, and also to Dirac charge quantization. Once we have identified the configuration space $\check\scrE{}^d(M)$
of a generalized abelian gauge theory, in the functional integral
approach to quantization we ``integrate'' over the isomorphism classes
$\check\E{}^d(M)$ using a suitable translation invariant measure; such
a Haar-like measure exists at least formally for gaussian fields and is induced
by the riemannian metric on $M$. Let us briefly explain the meaning of such an
integration, illustrated through several explicit examples.

To set up the path integral of the gauge theory, we regard the set of local
fields $\scrF$ as a certain covariant functor from the (opposite) category of smooth
manifolds with suitable morphisms to the category of sets. Locality of the fields is
implemented by the requirement that the functor $\scrF$ satisfies a
Mayer--Vietoris sheaf property, i.e. there is a pullback square
\beqa
\xymatrix{
\scrF(U\cup V) \ \ar[d] \ar[r] & \ \scrF(V) \ar[d] \\
\scrF(U) \ \ar[r] & \ \scrF(U\cap V)
}
\eeqa
for any pair of open charts $U,V$. In most of our examples we take
$\scrF=\Omega^q$, with $q=0$ corresponding to scalar fields, $q=1$ to
gauge fields, and so on. In our gauge theory examples we can
generalize this to require that the fields be valued in the suitable
configuration space, which requires
replacing the category of sets such that one considers sheaves of
groupoids, higher groupoids,
or even $\infty$-groupoids; such is the case for e.g. double covers
whose target is the category of simplicial sets.

Let $\Bord_m$ be the bordism category of smooth $m$-manifolds; an
object of $\Bord_m$ is a closed $m-1$-manfiold $N$, while a morphism from
$N_0$ to $N_1$ is an $m$-manifold $M$ with boundary $\partial M=N_0
\sqcup N_1$ and composition defined by gluing. Given a collection of fields
$\scrF$, the bordism category $\Bord_m(\scrF)$ enriched by $\scrF$ has
the same objects, but its morphisms are extended to pairs $(M,\Phi)$ where
$\Phi\in\scrF(M)$. Let $\Vect_\IC$ be the category of complex vector
spaces with linear transformations. The partition function of our gauge theory is a monoidal
functor
\beq
\scrZ_\scrF \,:\, \Bord_m(\scrF) \ \longrightarrow \ \Vect_\IC \ ,
\label{BordQFT}\eeq
which sends disjoint unions to tensor products. Semi-classical
quantization corresponds to restricting $\scrZ_\scrF$ so that it takes
values in
invertible objects of $\Vect_\IC$; if $M$ is closed then
$\scrZ_\scrF(M,\Phi)\in\IC^\times$ and we write $\scrZ_\scrF(M,\Phi)=:\exp\ii
S_M[\Phi]$. See~\cite{FHLT} for further details and constructions.

Let us look at a simple example of a higher abelian gauge theory to
demonstrate the need for using cocycles as objects in a suitable
category in order to formulate the path integral of the quantum theory. In
the setting of \S\ref{GAGT}, the field content of a generic higher abelian
gauge theory is $\Phi=(g,\check j_e,\check j_m,F)$ where $g$ is the
metric of $M$, the differential cocycles $\check j_e$ and $\check j_m$
are smooth refinements of electric and magnetic current forms
$j_e=\dd\star F\in
\Omega^{n-p+2}(M)$ and $j_m=\dd F\in\Omega^{p+1}(M)$, and
$F\in\Omega^p(M)$. We take $p=1$, and set $j_e=\sum_{x\in W_e} \, q_e(x)\, \PD_M(x)$ where $q_e(x)\in\IR$ are electric
charges inserted at a collection of points $x\in W_e\subset M$. If $j_m=0$, then
$\dd F=0$ and $F$ can be refined to a differential cohomology class
$\lambda\in\check\RH{}^1(M)= \Omega^0(M;\uo)$, i.e. a smooth map
$\lambda:M\to S^1$. Then $F=\dd\log\lambda$. In the quantum gauge
theory, exponentiation of the action functional (\ref{GAGTaction}) gives
\beq
\exp \ii S_M[\lambda] = \exp\Big(-\frac\ii2 \, \int_M\,
\frac{\dd\lambda\wedge\star\dd\lambda}{\lambda^2} \, \Big) \ \prod_{x\in
  W_e}\, \lambda(x)^{q_e(x)} \ ,
\label{GAGTquaction}\eeq
which is well-defined and $\IC$-valued provided that electric charge is quantized,
$q_e(x)\in\IZ$. Suppose now that $j_m\neq0$. Since $j_m=\dd F$ is
trivialised, we can refine it to a class in the differential
cohomology $\check\RH{}^2(M)$, represented by a hermitian line bundle with
connection $(L,\nabla)$ on $M$. We now take $\lambda\in \Omega^0(M;L)$
and set $F=\dd_\nabla\log\lambda$ so that $j_m=\nabla^2$ is the
curvature of $\nabla$. Then the product in (\ref{GAGTquaction}) lies in the fibres $\bigotimes_{x\in
  W_e}\, (L_x)^{\otimes q_e(x)}$, and the quantum action
(\ref{GAGTquaction}) takes values in a line bundle (rather than
in~$\IC$). The line bundle is an obstruction to defining the path
integral and it represents an \emph{anomaly}. ``Anomaly cancellation'' corresponds to a trivialization of
this line bundle; see~\cite{Freed:2000ta} for further details,
and~\cite{Freed:2000ta,Freed:2000tt} for an extension to
Green--Schwarz anomaly cancellation in Type~I superstring theory.

More generally, the product in (\ref{GAGTquaction}) is replaced with
\beqa
\chi(W) = \exp\Big(\ii \int_W\, q_e\, A\big|_W\Big)
\eeqa
for a $p$-brane $W$. Locality requires that
$\chi\in\Hom_{\scrAb}(Z_{p-1}(M),\uo)$. The equations of motion imply that
\beqa
\chi(W'\,)= \chi(W) \, \exp\Big(\ii\int_B\, q_e\, F\Big)
\eeqa
if $B$ is a bordism between $W$ and $W'$. This means that $\chi$ is a
Cheeger--Simons differential character (Def.~\ref{CheegerSimonsdef}), and $q_e\,
F\in{\rm im}(\delta_1)$ where $\delta_1:\check\RH{}^p(M)\to
\Omega^p_\IZ(M)$ is the field strength map $\delta_1(\chi)=F_\chi$. When $p=0$ and the field strength $F$ is produced by a
magnetic brane as above, we immediately arrive at Dirac quantization
of charge $q_e\, q_m\in2\pi\, \IZ$, as argued from a different
perspective in \S\ref{Maxwellsemi} More generally, the charges live in
the lattice obtained from the image of integral cohomology in
(\ref{classchargegp}). In dimensions $m=n+1=4s+3$ with $p=2s+2$, one can
also add a Chern--Simons term
\beqa
\exp\big(\ii \langle\check A,\check A\rangle \big) \ .
\eeqa
For $s=0$ this term is well-defined by picking a spin structure on the
three-manifold $M$.

One can
also have charges in images of generalised cohomology theories. For example, in
Type~II superstring theory the D-brane $(W,E,f)$ carries a vector bundle
with connection $(E,\nabla)$. If $C$ is a globally-defined
Ramond--Ramond gauge potential, then the product in
(\ref{GAGTquaction}) is replaced with
\beqa
\exp\Big(\ii\int_W\, Q(W,E,f)\wedge C\big|_W\Big)
\eeqa
where $Q(W,E,f)$ is the charge vector (\ref{Dbranecharge}).

Let us finally consider an example from M-theory. The field content
$\Phi=(\sigma,g,A,\psi,C)$ of supergravity on an 11-dimensional spin
manifold $M$ consists of a spin structure $\sigma$ on $M$, a riemannian
metric $g\in\Omega^0(M;T^*M\otimes T^*M)$, a connection one-form $A$ on a
principal bundle over $M$, a twisted spinor field $\psi\in\Omega^0(M;T^*M\otimes \cals_M)$
where $\cals_M$ is the twisted spin bundle of $M$, and an abelian gauge
potential $C\in\Omega^3(M)$ with field strength $G=\dd C\in
\Omega^4(M)$. The relevant terms in the quantum supergravity action
are
\beq
\exp\Big(\ii\int_M\, G\wedge\star\, G+ 2\ii \int_M\, \big(C\wedge
G\wedge G-C\wedge I_8(g)\big) \Big)
\label{11daction}\eeq
where $I_8(g)=\frac1{48}\, \big(4p_2(M)-p_1(M)^{\wedge 3}\big)$ and
$p_k(M)\in \RH^{4k}(M;\IZ)$ are the Pontrjagin classes of the tangent
bundle of $M$. For topologically non-trivial fields, we refine the
three-form $C$ to a class $\check C\in\check\RH{}^4(M)$. Then the
topological terms in (\ref{11daction}) refine to $\exp\big(\frac\ii6\,
\int^{\check\RH{}}\!\!\!\int_M \, \check C\smile\check C\smile\check
C\big)$; making this term well-defined requires a cubic refinement of
the trilinear form $\check\RH{}^4(M)\times \check\RH{}^4(M)\times
\check\RH{}^4(M)\to \torus$ defined by it. Since $K(\IZ,4)=B\E_8$ the
charge $c\in\RH^4(M;\IZ)$ is an isomorphism class of a principal
$\E_8$-bundle over $M$~\cite{Diaconescu:2000wy} (up to approximation
on the skeleton of $M$). The groupoid of fields $\check\scrH^4(M)$
consists of cocycles $\check C=(P,\nabla, C)\in\check Z{}^4(M)$, where
$P\to M$ is an $\E_8$-bundle with connection $\nabla$ and
$ C\in\Omega^3(M)$; gauge transformations connect cocycles
$(P,\nabla, C)$ and $(P',\nabla', C'\,)$ with
$C'-C =\CS(\nabla,\nabla'\,)+F_\chi$ for some $\chi\in\check\RH{}^3(M)$. This implies that $G= F(\check
C)= \Tr\big(
F_\nabla\wedge F_\nabla \big) -\frac12\,\Tr\big(R(g)\wedge R(g)\big)
+\dd C$ is gauge-invariant, where
$F_\nabla=\nabla^2$ is the curvature of the connection $\nabla$ and
$R(g)$ is the curvature two-form of the metric $g$ with
$\Tr\big(R(g)\wedge R(g)\big)=\frac14\, p_1(M)$;
see~\cite{Witten:1996md,Diaconescu:2003bm,Freed:2004yc} for further
details.

The functor (\ref{BordQFT}) with values in an invertible quantum field
theory can be formally gotten by performing the functional integral
over the configuration groupoid $\check\scrH{}^p(M)$, with
$\pi_0\check\scrH{}^p(M)= \check\RH{}^p(M)$, of free higher
abelian gauge theory, which is studied in~\cite{Kelnhofer:2007jf} using techniques of covariant
quantization on compact closed manifolds: For $\Phi=(g,\check A)$ with
$g$ a riemannian metric on $M$ and $[\check A]\in\check\RH{}^p(M)$,
by~\cite[Thm.~4.5]{Kelnhofer:2007jf} the partition function is
rigorously defined by the formula
\beq
\scrZ_p(M) = \prod_{j=0}^{p-1}\, \bigg(\,
\frac{\det'\big(\dd^\dag\,\dd\big|_{\Omega^j(M)\cap{\rm
      im}(\dd^\dag)}\big)}{{\rm
    vol}\big({\tt harm}^j(M)\,\big/\,{\tt harm}_\IZ^j(M)\big)^2}\,
\bigg)^{\frac12\,(-1)^{p-j}} \ \Theta_p(M) \ \big|\Tor\,
\RH^p(M;\IZ)\big|
\label{partfnrig}\eeq
where 
$\Theta_p(M)=\sum_{f\in{\tt harm}^p_\IZ(M)}\,
\exp\big(-\frac12\,\int_M\, f\wedge\star f\big)$ is a Riemann
theta-function on the lattice of harmonic $p$-forms $f$ on $M$,
i.e. $\dd f=0=\dd\star f$, with
integer periods; it can be interpreted as a section of a line bundle
over $\check\RH{}^p(M)$, of the type that arises in Chern--Simons theory. The product in (\ref{partfnrig}) arises from a formal
gaussian integral over oscillator modes $F_0+\dd a$, with
$a\in\Omega^{p-1}(M)$ modulo the small gauge invariances
$a\mapsto a+\dd\varepsilon$, $\varepsilon\mapsto \varepsilon+\dd
\eta$, and so on; it can be interpreted as a sort of analytic
torsion~\cite{Belov:2006jd}, i.e. as a Quillen norm of a section of
some determinant
line bundle over the space of metrics on $M$. In~\cite{Kelnhofer:2007jf} it is shown that the
partition function (\ref{partfnrig}) exhibits an electric-magnetic
duality relation wherein $\scrZ_p(M)$ and $\scrZ_{n+1-p}(M)$ are
proportional to each other, which follows from Hodge theory and the
Poisson resummation formula. A functional integral
approach to the quantization of self-dual higher abelian gauge fields
is similarly described in~\cite{Belov:2006jd} using higher-dimensional
Chern--Simons theories, and further elucidated
in~\cite{Monnier:2010ww}; an analogous path integral quantization of Ramond--Ramond
gauge theory is carried out in~\cite{Belov:2006xj}.

\subsection{Hamiltonian quantization\label{Quant}}

Another approach to constructing the functor (\ref{BordQFT}) into an
invertible quantum field theory is to categorify the partition
function to the Hilbert space of the quantum field theory; as the
partition function is generally valued in a line bundle, the Hilbert
space is thus valued in a gerbe. In the remainder of this article we
will explain how to construct this Hilbert space; we set $\check
j_e=0=\check j_m$ for the rest of the section.

Consider again the spacetime manifold $M=\IR\times N$ where $N$
is a compact $\E$-oriented $n$-dimensional riemannian manifold. As we demonstrate below, the configuration space of a
free generalized abelian gauge theory on $M$ is the generalized
differential cohomology group $\check\E{}^d(N)$. Heuristically, the general principles of hamiltonian
quantization suggest that the
\emph{Hilbert space} of the quantum field theory on which the fields
act as operators is the space $\hil={\rm L}^2\big(\check\E{}^d(N)\big)$ of
square integrable functions on the manifold $\check\E{}^d(N)$ with
respect to a suitable measure. The problem, however, is that the
differential cohomology $\check\E{}^d(N)$ is an infinite-dimensional
vector space, so it is tricky to define measures on it. Instead, we
will approach the problem of quantization from a group theory
perspective, and appeal to the representation theory of the Heisenberg
group. This will identify the quantum Hilbert space as a
representation of a certain Heisenberg extension of $\check\E{}^d(N)$~\cite{Freed:2006ya,Freed:2006yc}.

Recall the classical definition of a Heisenberg group $\Heis(V,\omega)$
associated to a symplectic vector space $(V,\omega)$: It is a central extension of the translation group $V$ by the
circle group $\uo=\{z\in\IC \ | \ |z|=1\}$. Topologically
$\Heis(V,\omega)\cong V\times\uo$ with a twisted multiplication
\beqa
(v_1,z_1)\cdot (v_2,z_2) = \big(v_1+v_2\,,\,\e^{\pi \ii
  \omega(v_1,v_2)}\, z_1\, z_2 \big) \ .
\eeqa

The idea behind the relevance of the Heisenberg group in quantization is
as follows. Let $\cG$ be a topological abelian Lie group with Haar measure,
and $\widehat{\cG}$ the Pontrjagin dual group of characters $\chi:\cG\to
\uo$. The groups $\cG$ and $\widehat{\cG}$ both act on the Hilbert space
$\hil:= {\rm L}^2(\cG)$, respectively as the translation or ``momentum'' operators
\beqa
(T_h\psi)(g)=\psi(g+h)
\eeqa
and as the multiplication or ``position'' operators
\beqa
(M_\chi\psi)(g)=\chi(g)\, \psi(g)
\eeqa
for $\psi\in\hil$, $h,g\in \cG$, and $\chi\in\widehat{\cG}$. The Hilbert
space $\hil$ is not a representation of $\widetilde{\cG}=\cG\times
\widehat{\cG}$, since
\beq
T_h\circ M_\chi=\chi(h)\, M_\chi\circ T_h \ .
\label{ThMchi}\eeq
But the commutation relations (\ref{ThMchi}) can be thought of as
originating from a suitable cocycle, and $\hil$ \emph{is} a representation of the Heisenberg group
$\Heis(\, \widetilde{\cG}\, )$ associated to $\widetilde{\cG}$, which is a certain central extension of
$\widetilde{\cG}$ by $\uo$; specifically, $\hil$ is the unique
Stone--von~Neumann representation of $\Heis(\, \widetilde{\cG}\,)$. We
will now proceed to define these concepts precisely.

\subsection{Heisenberg groups and their representations\label{Heisenberg}}

We will begin by collecting some general results concerning
central extensions of Lie groups and their representations,
following~\cite{Freed:2006ya}.

\begin{definition}
Let $\mathcal{G}$ be an abelian Lie group.
A \emph{generalized Heisenberg group} is a Lie
group ${\tt Heis}(\mathcal{G})$ which sits inside the exact sequence
\begin{displaymath}
1~\longrightarrow~\uo ~\longrightarrow~
{\tt Heis}(\mathcal{G})~\longrightarrow~\mathcal{G}~
\longrightarrow~{0} \ ,
\end{displaymath} 
with the circle group $\uo$ contained in the center ${{Z}_{{\tt
 Heis}(\mathcal{G})}}$. We will further require that the group
manifold of ${\tt Heis}(\mathcal{G})$ is a smooth, locally trivial
circle bundle over $\mathcal{G}$. This is guaranteed by assuming the
group $\mathcal{G}$ fits inside the exact sequence~\cite{Freed:2006ya}
\begin{displaymath}
\xymatrix{1 \ \ar[r]& \ \pi_{1}(\cG) \ \ar[r]& \ \mathfrak{g} \
  \ar[r]^{\exp}& \ \cG \ \ar[r]& \ \pi_{0}(\cG) \ \ar[r]& \ 0}
\end{displaymath}
where $\mathfrak{g}$ is the Lie algebra of $\cG$, and $\exp$ is the
exponential map. A generalized Heisenberg group is
said to be \emph{maximally noncommutative} if ${Z}_{{\tt
  Heis}(\mathcal{G})}=\uo$. A maximally noncommutative
generalized Heisenberg group is simply called a Heisenberg group.
\end{definition}

Any smooth map $c:\mathcal{G}\times\mathcal{G}\to\uo$ satisfying
the cocycle condition 
\begin{displaymath}
c(g_{1},g_{2})\,c(g_{1}+g_{2},g_{3})=c(g_{1},g_{2}+g_{3})\,c(g_{2},g_{3})
\end{displaymath}
for all ${g_{1},g_{2},g_{3}\in\mathcal{G}}$ defines a generalized
Heisenberg group denoted ${\tt Heis}(\mathcal{G};c)$. As a manifold
${\tt Heis}(\mathcal{G};c)$ is topologically the product
$\mathcal{G}\times\uo$, while the multiplication is defined by
\begin{equation}
(g_{1},z_{1})\cdot(g_{2},z_{2}):=
\big(g_{1}+g_{2}\,,\,c(g_{1},g_{2})\,z_{1}\,z_{2}\big)
\label{Heismult}\end{equation}
for all ${g_{1},g_{2}\in\mathcal{G}}$ and
$z_{1},z_{2}\in\uo$. From the group cocycle $c$ we construct the
\emph{commutator map}%% 
\beq
s\,:\,\mathcal{G}\times\mathcal{G}~\longrightarrow~\uo
\label{commmap}\eeq
defined by
\begin{equation}\label{commcocycle}
s(g_{1},g_{2}):=\frac{c(g_{1},g_{2})}{c(g_{2},g_{1})}
\end{equation}
for all ${g_{1},g_{2}\in\mathcal{G}}$. The commutator map $s$ enjoys
the following properties:
\begin{enumerate}
\item From the definition (\ref{Heismult}) of the group operation in
  $\Heis(\cG;c)$ and (\ref{commcocycle}) it follows that the group commutator is given by
\beq
\big[(g_1,z_1) \,,\, (g_2,z_2)\big]= \big(0\,,\, s(g_1,g_2)\big) \ .
\label{Heisgroupcomm}\eeq
\item $s$ is alternating:
\beqa
s(g,g)=1 \ .
\eeqa
\item $s$ is bimultiplicative:
\beqa
s(g_1+g_2,h)=s(g_1,h)\, s(g_2,h) \qquad \mbox{and} \qquad
s(g,h_1+h_2)= s(g,h_1)\, s(g,h_2) \ .
\eeqa
\item $s$ is skew-symmetric:
\beqa
s(g,h) = s(h,g)^{-1} \ .
\eeqa
\end{enumerate}

Given a smooth map $f:\mathcal{G}\to\uo$, consider the cocycle
$\tilde{c}$ for the group $\mathcal{G}$ defined by
\begin{displaymath}
\tilde{c}(g_{1},g_{2}):=\dfrac{f(g_{1}\,g_{2})}{f(g_{1})\,f(g_{2})}\,
c(g_{1},g_{2}) \ .
\end{displaymath}
We say that $\tilde{c}$ and $c$ differ by a
\emph{coboundary}. The map%% 
\begin{displaymath}
(g,z)~\longmapsto~\big(g\,,\,f(g)\,z\big)
\end{displaymath}
induces an isomorphism
\begin{displaymath}
{\tt Heis}(\mathcal{G};c)~\xrightarrow{ \ \approx \ }~
{\tt Heis}(\mathcal{G};\tilde{c}\,) \ .
\end{displaymath}
It follows easily that%%
\begin{displaymath}
\tilde{s}(g_{1},g_{2})=s(g_{1},g_{2})
\end{displaymath}
for all ${g_{1},g_{2}}\in\mathcal{G}$. A complete characterization of
generalized Heisenberg groups is given in~\cite{Freed:2006ya}.

\begin{proposition}\label{class} Given an abelian Lie group
  $\mathcal{G}$, any generalized Heisenberg group ${\tt
    Heis}(\mathcal{G})$ is of the form ${\tt Heis}(\mathcal{G};c)$ for
  some cocycle $c$, and ${\tt Heis}(\mathcal{G})$ is uniquely
  determined up to isomorphism by its commutator map
  (\ref{commmap}). Conversely, every alternating and bimultiplicative
  map $s:\mathcal{G}\times\mathcal{G}\to\uo$ uniquely
  determines a Heisenberg group ${\tt Heis}(\mathcal{G})$ up to
  isomorphism.
\end{proposition}

Denote by $\mathscr{E}_{\mathcal{G}}$ the category of central
extensions of $\mathcal{G}$ by the circle group $\uo$ with the usual
morphisms. Let $\mathscr{C}_{\mathcal{G}}$ be the category whose
objects are bimultiplicative maps $\psi:\cG\times \cG\to \uo$
(and hence automatically satisfy the cocycle condition), and whose morphisms
are \emph{quadratic} maps $f:\cG\to\uo$, sending $\psi$ to the
map $\psi\, \psi_{f}$ where
\begin{displaymath}
\psi_{f}(g_{1},g_{2})=\frac{f(g_{1}\, g_{2})}{f(g_{1})\, f(g_{2})}
\ .
\end{displaymath}
The bimultiplicativity of the coboundary $\psi_{f}$ is precisely what is meant by $f$
being quadratic. Then there is a functor
\begin{displaymath}
\mathscr{C}_{\mathcal{G}}~\longrightarrow~\mathscr{E}_{\mathcal{G}}
\end{displaymath}
which assigns a central extension to a cocycle. By
Prop.~\ref{class}, this defines a
natural equivalence of categories.

Suppose that we weaken the definition of the commutator map $s$ so that it is skew-symmetric but not alternating. This implies 
\begin{displaymath}
s(g,g)^{2}=1
\end{displaymath}
for all $g\in\mathcal{G}$, and so the group $\mathcal{G}$ acquires a
natural $\IZ_2$-grading given by the homomorphism $g\mapsto
s(g,g)\in\IZ_2$. A \emph{graded (generalized) Heisenberg
  group} is a central extension ${\tt Heis}(\mathcal{G})$ of
$\mathcal{G}$ by $\uo$ which is at the same time a
$\IZ_2$-graded group; the maps in the central extension are
$\IZ_2$-graded homomorphisms, with $\uo$ regarded as trivially graded. A graded central extension naturally
determines a commutator map which is skew-symmetric but not
alternating~\cite{Freed:2006ya}. Note that to every such skew-symmetric bimultiplicative map $s$ one can assign an alternating bimultiplicative map $\widetilde s$ by defining 
\beq
\widetilde{s}(g_{1},g_{2}):=s(g_1,g_2) \ \exp\big(-\pi \ii
\varepsilon(g_{1})\, \varepsilon(g_{2}) \big)
\label{stildeeps}\eeq
where $\varepsilon(g)$ is defined (modulo 2) through $s(g,g)=\exp\big(
\pi
  \ii \varepsilon(g) \big)$.

\begin{proposition}\label{class2} Every graded central extension of an
  abelian Lie group $\mathcal{G}$ by $\uo$ is determined
  uniquely up to isomorphism by its graded commutator map
  (\ref{commmap}). Conversely, every skew-symmetric and bimultiplicative map
  $s:\mathcal{G}\times\mathcal{G}\to\uo$ uniquely determines
  such a graded central extension up to isomorphism.
\end{proposition}

\subsubsection*{Generalized Stone--von~Neumann theorem}

One of the main results in the theory of generalized Heisenberg groups
is the fact that the irreducible unitary representations are uniquely
determined. This is essentially an extension of the Stone--von~Neumann
theorem. Consider the group 
\begin{equation}\label{center}
Z_{\mathcal{G};c}:=\big\{g\in\mathcal{G}\:\big| \:s(g,h)=1\quad
\forall h\in\mathcal{G}\big\} \ .
\end{equation}
The center $Z_{{\tt Heis}(\mathcal{G};c)}$ of ${\tt
  Heis}(\mathcal{G};c)$ sits in the exact sequence
\begin{displaymath}
1~\longrightarrow~\uo~\xrightarrow{ \ i \ }~
Z_{{\tt Heis}(\mathcal{G};c)}~\longrightarrow~{Z_{\mathcal{G};c}}~
\longrightarrow~0
\end{displaymath}
where $i$ is the inclusion. In any irreducible unitary representation $\rho$ of ${\tt
  Heis}(\mathcal{G};c)$, by Schur's lemma the center acts by scalar
multiplication as elements of the circle group $\uo$. Since the representations
we are considering satisfy $(\rho\circ i)(z)=z\, \Id$ for all $z\in\uo$, it follows that this sequence splits non-canonically via a homomorphism $\chi:Z_{{\tt
    Heis}(\mathcal{G};c)}\to\uo$.

\begin{proposition}\label{representation}
Any irreducible unitary representation of a 
generalized Heisenberg group of finite
dimension ${\tt Heis}(\mathcal{G};c)$ for which
$\uo \subseteq Z_{{\tt Heis}(\mathcal{G};c)}$ acts by the identity
character is uniquely
determined up to isomorphism by a splitting homomorphism
$\chi:Z_{{\tt Heis}(\mathcal{G};c)}\to\uo$. Conversely, any
such homomorphism $\chi$ gives rise to such an irreducible unitary representation of $\Heis(\cG;c)$.
\end{proposition}

\begin{corollary}
If the commutator map $s$
is \emph{non-degenerate}, i.e. the group (\ref{center}) is the trivial
group $Z_{\mathcal{G};c}=0$, then up to isomorphism there is a unique irreducible unitary
representation of ${\tt Heis}(\mathcal{G};c)$ for which the center $Z_{\Heis(\cG;c)}=\{(0,z) \ |
\ z\in\uo\}$
acts by scalar multiplication.
\end{corollary}

\subsubsection*{Examples}

\begin{itemize}
\item Let us return to the example $\widetilde{\cG}=\cG\times \widehat{\cG}$
  and $\hil= {\rm L}^2(\cG)$ from \S\ref{Quant} The Heisenberg group
  extending $\widetilde{\cG}$ has a representation
\beqa
\Heis(\widetilde{\cG}\,) \ \longrightarrow \ {\rm GL}(\hil) \ , \qquad
\big((g,\chi) \,,\, z) \ \longmapsto \ z\, T_g\circ M_\chi \ .
\eeqa
The cocycle in this case is given by
\beqa
c\big((g_1,\chi_1)\,,\, (g_2,\chi_2)\big) = \frac{1}{\chi_1(g_2)} \ ,
\eeqa
and its antisymmetrization gives the commutator map
\beqa
s\big((g_1,\chi_1)\,,\, (g_2,\chi_2)\big) =
\frac{\chi_2(g_1)}{\chi_1(g_2)} \ .
\eeqa
\item Let $\cG=\IR$, which we regard as parametrizing ``coordinate''
  operators $\e^{\ii p\, \hat x}$. The Pontrjagin dual
  $\widehat{\cG}=\widehat{\IR}\cong \IR$ may then be regarded as parametrizing ``momentum'' operators $\e^{\ii x\, \hat p}$, so that
  $\widetilde{\cG}=\IR\times\IR$ is ``phase space''.  More precisely, any character on
  $\cG$ in this case is of the form $\chi(p)=\e^{\ii x\, p}$ for some
  $x\in\IR$. Then the commutator map $s$ gives the canonical
  symplectic pairing on phase space, and the uniqueness result of
  Prop.~\ref{representation} is the usual Stone--von~Neumann theorem of
  quantum mechanics expressing uniqueness of the irreducible
  Schr\"odinger representation of the Heisenberg commutation relations.
\end{itemize}

\subsubsection*{Polarization}

Prop.~\ref{representation} can be generalized to
infinite-dimensional abelian groups which are
\emph{polarized}~\cite{Freed:2006ya}. 

\begin{definition}
A \emph{polarization} of an abelian Lie group $\mathcal{G}$ is an action of the real line $\mathbb{R}$ on the Lie 
algebra $\mathfrak{g}$ of $\mathcal{G}$ via operators
$\{u_{t}\}_{t\in\R}$ which preserve the Lie bracket and decompose
the complexification $\mathfrak{g}_{\mathbb{C}}$ into a countable
direct sum of finite-dimensional subspaces $\mathfrak{g}_{\lambda}$,
$\lambda\in\mathbb{R}$, such that $u_{t}$ for each $t$ acts on
$\mathfrak{g}_{\lambda}$ as multiplication by
$\e^{\ii\lambda\,{t}}$. If $\cG$ is a polarized group, then a
unitary representation of the Heisenberg group ${\tt
  Heis}(\mathcal{G})$ on a Hilbert space $\mathcal{H}$ is said to be
of \emph{positive energy} if there is a unitary action of the real
line on $\mathcal{H}$ by operators $U_t=\exp(\ii t\, H)$,
$t\in\R$, which intertwine with the action of ${\tt
  Heis}(\mathcal{G})$ on $\mathcal{H}$ such that the operator $H$
has discrete non-negative spectrum.
\end{definition}

In physical applications, the
one-parameter family $\{u_{t}\}$ is typically given by hamiltonian flow on phase
space while $\{{U}_{t}\}$ determines the time evolution of the
associated quantum theory, hence the positive energy condition.

\begin{proposition}\label{representationinf}
For a polarized generalized Heisenberg group ${\tt
  Heis}(\mathcal{G};c)$, any irreducible unitary representation of
positive energy for which
$\uo \subseteq Z_{{\tt Heis}(\mathcal{G};c)}$ acts by the identity
character is uniquely determined up to isomorphism by a
splitting homomorphism $\chi:Z_{{\tt
    Heis}(\mathcal{G};c)}\to\uo$. Conversely, any such
homomorphism $\chi$ gives rise to such an irreducible unitary representation of
$\Heis(\cG;c)$ of positive energy.
\end{proposition}

If ${\tt Heis}(\mathcal{G})$ is a $\IZ_2$-graded
generalized Heisenberg group, then the quantum Hilbert space
$\mathcal{H}$ automatically acquires a $\IZ_2$-grading as
well.

\subsection{Quantization of free generalized abelian gauge theory\label{Quantgenab}}

We will now explain the setting for the quantization of
generalized abelian gauge theories following~\cite{Freed:2006ya,Freed:2006yc}. We
consider hamiltonian quantization of an abelian gauge theory with
semi-classical configuration space $\check\E{}^d(M)$ given by a smooth refinement $\check\E{}^\bullet$ of a generalized cohomology theory $\E^\bullet$ on the spacetime $M=\IR\times N$ with metric of indefinite signature, where $N$ is a compact oriented riemannian manifold. For $[\check{A}]\in\check\E{}^d(M)$, the classical equations of motion for the gauge theory are
\begin{equation}\label{maxwell}
\dd F([\check{A}])=0 \qquad \mbox{and} \qquad
\dd\star F([\check{A}])=0
\end{equation}
where $F:\check\E{}^d(M)\to \Omega_\IZ(M;E^\bullet)^d$ is the field strength map (curvature), with $E^\bullet=\E^\bullet(\pt)\otimes_\IZ\IR$, and $\star$ denotes the Hodge duality operator on $M$. Of course, the first equation is automatically satisfied. We are interested in the \emph{space of classical solutions}, i.e. the subspace $\mathcal{M}\subset{\check\E{}^d(M)}$ of gauge fields $[\check A]$ which solve the equations (\ref{maxwell}).

For this, we decompose $F([\check{A}])$ on $M$ as
\begin{equation}\label{decomposition}
F([\check{A}])=B(t)-E(t)\wedge{\dd t}
\end{equation}
where $t$ is the time coordinate on $\IR$, and $B(t)$ and $E(t)$ for
each $t\in\IR$ are $d$- and $d-1$-forms on $N$, respectively. (The
notation stems from the fact that when $\E^\bullet=\RH^\bullet$ is
ordinary cohomology and $d=2$, i.e. in Maxwell theory, the forms $B(t)$ and $E(t)$ are the
magnetic and electric fields, respectively.) Then (\ref{maxwell}) can be rewritten as
\begin{equation}\label{hamiltonian}
\dfrac{\partial}{\partial{t}}B=-\tilde \dd E \qquad \mbox{and} \qquad
\dfrac{\partial}{\partial{t}}\, \tilde\star E=\tilde \dd\, \tilde\star B \ ,
\end{equation}
where $\tilde\star$ denotes the Hodge duality operator and $\tilde \dd$
the exterior derivative on $N$. The Cauchy data for these
first order linear elliptic differential equations are the values of
$B(t)$ and $E(t)$ at a given initial time $t=t_0$; the corresponding
solutions uniquely determine $F([\check{A}])$ on $M$ through~(\ref{decomposition}).

In particular, we can identify the solution space $\mathcal{M}$ with the
tangent bundle $T\check\E{}^d(N)$ in the following way. First, notice
that $T\check\E{}^d(N)$ can be trivialized as $\check\E{}^d(N)\times
\Omega(N;E^\bullet)^{d-1} /\, {\rm im}(\dd)$. Consider the map 
\begin{displaymath}
i_{t_0}\,:\, N \ \longrightarrow \ \IR\times{N} \ , \qquad i_{t_0}(x)=(t_0,x) \ .
\end{displaymath}
It induces a map from $\mathcal{M}$ to $T\check\E{}^d(N)$ by assigning
to $[\check{A}]$ the pair $(i^{*}_{t_0}[\check{A}],E(t_0))$, where
$E(t)$ is determined by the decomposition (\ref{decomposition}). The
inverse map is obtained by assigning to the pair $([\check{B}],E)$ the
unique element $[\check{A}]\in\check\E{}^{d}(M)$ such that
$i_{t_0}^{*}[\check{A}]=[\check{B}]$ and
$F([\check{A}])=B(t)-E(t)\wedge \dd t$, where $B(t)$ and $E(t)$ are
obtained from (\ref{hamiltonian}) with initial conditions given by
$B(t_0)=B$ and $\tilde\star E(t_0)=\tilde\star E$. The uniqueness of
$[\check{A}]$ is assured by the fact that the map $i_{t_0}$ induces
the isomorphism $\E^{d-1}(M;\bbt)\cong \E^{d-1}(N;\bbt)$ of cohomology
classes (flat fields) in the kernel of the field strength
transformation.

With the identification of the space of classical solutions as
$\mathcal{M}=T\check\E{}^d(N)$, the standard hamiltonian quantization scheme
suggests that the quantum Hilbert space of the generalized abelian
gauge theory is given heuristically by the ``space of ${\rm L}^2$-functions
on $\cG=\check\E{}^d(N)$''. A more precise definition is given
in~\cite{Freed:2006ya,Freed:2006yc}, where it is proposed that the
quantum (projective) Hilbert space $\hil$ of a generalized abelian gauge theory with
configuration space a group $\cG$ is an
irreducible representation of the generalized
Heisenberg group
\beqa
\Heis\big(\cG\times\widehat{\cG}\, \big) \ ,
\eeqa
where $\widehat{\cG}=\Hom_{\scrAb}(\cG,\uo)$ is the group of
characters of $\cG$ in the category $\scrAb$ of abelian groups. The case of self-dual gauge theories is somewhat simpler, as then the phase
space $\mathcal{G}\times\widehat{\mathcal{G}}$ degenerates to the
configuration space; in this case, due to the self-duality equations (\ref{selfdualityeqsEd}),
the space of classical solutions on $M$ may be identified with the
diagonal subgroup $\cG\cong\widehat{\cG}$ and the quantization is
carried out using the Heisenberg group ${\tt
  Heis}(\mathcal{G})$ itself. This technique can be applied to any abelian group
$\mathcal{G}$ based on a smooth refinement of a Pontrjagin self-dual
generalized cohomology theory $\E^\bullet$~\cite[App.~B]{Freed:2006ya}; one quantizes the Poisson
manifold $\cG$ in this case.

Isomorphism classes of Heisenberg group extensions are determined by
maps $s:\widetilde{\cG}\times\widetilde{\cG}\to \uo$ which are
skew-symmetric, alternating and bimultiplicative, where
$\widetilde{\cG}=\cG\times\widehat{\cG}$. For any generalized differential cohomology theory one can define a
pairing
\beq
\check\E{}^d(N)\otimes \check\E{}^{n-s+1-d}(N)
\ \longrightarrow \ \bbt \ , \qquad \big( [\check A]\,,\, [\check
A'\,] \big) \ \longmapsto \ \int^{\check\E}\!\!\!\!\!\int_N \, [\check A]\smile [\check
A'\,] \ .
\label{genEpairing}\eeq
The perfectness of the pairing (\ref{genEpairing}) is a feature of any
generalized differential cohomology theory $\check\E{}^\bullet$, defined as explained
in \S\ref{Gendiffcoh}, for which $\E^\bullet$ is a Pontrjagin
self-dual generalized cohomology theory, defined as in
\S\ref{Selfduality} The proof makes use of the fact that for such
theories there is a perfect pairing
\beqa
\E^d(N)\otimes \E^{n-s-d}(N;\bbt) \ \longrightarrow \ \bbt \ ,
\eeqa
and that the $\IR/\IZ$ cohomology $\E^\bullet(N;\bbt)$ appears as the
kernel of the field strength map; see~\cite[App.~B]{Freed:2006ya} for
details. The commutator map $s$ may then be constructed by exponentiating the
pairing~(\ref{genEpairing}).

The equations (\ref{maxwell}) can be obtained as the variational
equations for the action functional
\begin{equation}\label{lag}
{S}([\check{A}]):=-\frac{1}{2}\, \int_{M}\, F([\check{A}])\wedge\star
F([\check{A}]) \ ,
\end{equation}
and the hamiltonian derived from (\ref{lag}) is given by
\begin{equation}\label{ham}
{H}(t):=\frac{1}{2}\, \int_{N}\,
\big(B\wedge\tilde\star B+E\wedge\tilde\star E \big) \ .
\end{equation}
The hamiltonian (\ref{ham}) is a non-negative function defined on the cotangent bundle $T^{*}\cG$. At the identity,
$T_0^*\cG$ is the dual $\mathfrak{g}^*$ of the Lie algebra
$\mathfrak{g}$ of $\cG$. The canonical hamiltonian flow $\IR\to T^*\cG$
yields a family of maps $u_t:\mathfrak{g}^*\to \mathfrak{g}^*$ which
is an action of $\IR$ on $\mathfrak{g}^*$. If $\widehat\cG$ is the
Pontrjagin dual of $\cG$ obtained through a non-degenerate pairing
$\cG\otimes\cG \to\bbt$, then we obtain a family of operators acting
on the Lie algebra $\mathfrak{g}\oplus \mathfrak{g}^*$ of
$\cG\times\widehat\cG$ which satisfies all the properties of a
polarization.
Quite generally, the choice of polarization also appears in K\"ahler
quantization, wherein the K\"ahler form is given by the differential
of the antisymmetric pairing; in this instance though one should
clarify the origin of a suitable pre-quantum line bundle on the
configuration groupoid. 

In the self-dual case, this polarization does \emph{not}
induce a polarization on the diagonal subgroup. See~\cite{Freed:2006yc} for
a way to relate the self-dual gauge theory to a non-self-dual gauge
theory in dimensions $\dim(M)=4k+2$, $k\in\IN$; see also~\cite{Belov:2006jd,Belov:2006xj,Monnier:2010ww} where the
complete pre-quantization data is specified. The problems with
formulating self-dual higher abelian gauge theories which are both
local and covariant go back to e.g.~\cite{Henneaux:1988gg}, see also~\cite{Dunne:1989hv}.

In the following we will explicitly work out the cases where
$\E^\bullet=\RH^\bullet$ is ordinary cohomology and where
$\E^\bullet=\K^\bullet$ is complex K-theory.

\subsection{Quantization of higher abelian gauge theory\label{QuantHAGT}}

We will first apply this formalism to the Cheeger--Simons
groups. Let us begin by giving the heuristic argument using canonical
quantization of free fields. Let
$(\check A,\Pi)$ denote local coordinates on the phase space
$T^*\check\RH{}^p(N)= \check\RH{}^p(N)\times \Omega_{\rm cl}^{n-p+1}(N)$ where
$\Pi=\star F\big|_N$ with $[\star F]_{\rm dR}\in\RH^{n-p+1}(N;\IR)$; we think of the ``conjugate momentum'' $\Pi$ as
the functional derivative operator
$-\ii \hbar\, \frac\delta{\delta\check A}$ which generates translations
on the configuration space $\check\RH{}^p(N)$.  The
differential cohomology $\check\RH{}^p(N)$ is an abelian Lie group
with a translation-invariant measure induced formally on $T\check\RH{}^p(N)$ by
the riemannian metric on $N$. By quantizing $F,\Pi$
to operators $\hat F,\hat\Pi$ acting on the Hilbert space $\hil:={\rm
  L}^2\big(\check\RH{}^p(N)\big)$ we get the Heisenberg commutation relations
\beqa
\Big[\,\int_N\, \omega_1\wedge \hat F\,,\, \int_N\, \omega_2\wedge
\hat\Pi\, \Big] = \Big(\ii\hbar \int_N\, \omega_1\wedge\dd\omega_2
\Big) \ \Id_\hil
\eeqa
for any pair of differential forms $\omega_1,\omega_2$. The right-hand
side of these relations is just the pairing on globally defined forms
in differential cohomology.

Now let us make this argument more precise. Using Pontrjagin duality
\beqa
\check\RH{}^p(N)\times
\check\RH{}^{n-p+1}(N) \ \longrightarrow \ \torus
\eeqa
for $N$ compact and oriented, we set
\beq
\widetilde{\cG} = \check\RH{}^p(N)\times \widehat{\check\RH{}^p(N)}
\cong \check\RH{}^p(N)\times \check\RH{}^{n-p+1}(N)
\label{AGTphasespace}\eeq
which is the phase space of an abelian gauge theory from \S\ref{GAGT}
As the space of classical solutions on $M=\IR\times N$ in this
case is the tangent bundle on the Cheeger--Simons group $
\check\RH{}^p(N)$, the Hilbert space is heuristically $\hil={\rm
  L}^2\big(\check\RH{}^p(N)\big)$. Since in this case the abelian
group $\cG=\check\RH{}^p(N)$ is infinite-dimensional, the quantization
must be specified by a polarization. In the
hamiltonian formalism, a natural polarization is given by the
energy operator $H=\ii\frac\partial{\partial t}$; then the
complexifcation of the space of classical solutions $\mathcal{M}$ is a sum of
subspaces of
positive and negative energy solutions. In this case the quantum
(projective) Hilbert space $\hil$ is the unique irreducible
representation of the associated Heisenberg group which is compatible
with the polarization. If $N$ is non-compact, then this discussion
needs to be modified using some (conjectural) analog of ${\rm
  L^2}$-cohomology for the Cheeger--Simons groups.

In the present case the commutator map $s:\widetilde{\cG}\times\widetilde{\cG}\to \uo$ can be constructed using the
Pontrjagin--Poincar\'e duality property of ordinary differential
cohomology to define the pairing
\beqa
\langle-,-\rangle \,:\, \check\RH{}^p(N)\otimes \check\RH{}^{n-p+1}(N)
\ \longrightarrow \ \bbt=\IR/\IZ \ , \qquad \big\langle [\check A]\,,\, [\check
A'\,] \big\rangle=\int^{\check\RH}\!\!\!\!\!\int_N \, [\check A]\smile [\check
A'\,] \ .
\eeqa
This pairing is \emph{perfect}, i.e. it induces an isomorphism
\beqa
\check\RH{}^p(N)\cong\Hom_{\scrAb}\big(\check\RH{}^{n-p+1}(N)\,,\,\bbt\big)
\ ,
\eeqa
so that every homomorphism
$\check\RH{}^p(N) \to\bbt$ is given by pairing with an element of
$\check\RH{}^{n-p+1}(N)$. Let us sketch how to
understand this isomorphism. Using the universal coefficient theorem
and Poincar\'e duality, one shows that the pairing
\beqa
\RH^p(N;\bbt)\otimes \RH^{n-p}(N) \ \longrightarrow \ \bbt \ , \qquad
(\alpha,\alpha'\,) \ \longmapsto \ \big\langle \alpha\smile\alpha'
\,,\, [N]\big\rangle
\eeqa
is perfect, where $[N]$ denotes the fundamental class of the manifold
$N$. Moreover, one can define a pairing
\beqa
\Omega_\IZ^p(N)\otimes\Omega^{n-p}(N)\,\big/\, \Omega^{n-p}_\IZ(N) \
\longrightarrow \ \bbt \ , \qquad \big(F\,,\,[A]\big) \ \longmapsto \
\int_N\, F\wedge A \quad {\rm mod} \ \IZ
\eeqa
which is well-defined and perfect. Then by (\ref{CSexactseqs}) there is a commutative diagram
\beqa
\xymatrix{
0 \ \ar[r] & \ \RH^{p-1}(N;\bbt) \ \ar[r]\ar[d]^\approx & \ \check\RH{}^p(N) \
\ar[r]^F \ar[d] & \ \Omega_\IZ^p(N) \ \ar[r]\ar[d]^\approx & \ 0 \\
0 \ \ar[r] & \ \widehat{\RH^{n-p+1}(N;\bbt)} \ \ar[r] & \ \widehat{\check\RH{}^{n-p+1}(N)} \
\ar[r] & \ \widehat{\Omega^{n-p}(N)\,\big/\, \Omega_\IZ^{n-p}(N)} \ \ar[r] & \ 0 
}
\eeqa
where $\widehat{G}:=\Hom_{\scrAb}(G,\bbt)$ and the vertical maps are
the morphisms induced by the above pairings. It then follows that the
middle arrow is an isomorphism as well.

By using this pairing we may define a non-degenerate commutator map
$s:\widetilde{\cG}\times \widetilde{\cG}\to \uo$ by
\beqa
s\big(([\check A_1], [\check
A_1'\,]) \,,\,([\check A_2], [\check
A_2'\,]) \big) = \exp\Big( 2\pi\ii \big(\langle [\check A_2], [\check
A_1'\,]\rangle - \langle [\check A_1], [\check
A_2'\,]\rangle \big) \Big)
\eeqa
which defines a central extension of the group (\ref{AGTphasespace}). In this case, the map
$\varepsilon$ from (\ref{stildeeps}) is given in terms of Wu classes which are
polynomials in the Stiefel--Whitney classes of the tangent bundle of
$N$~\cite{Freed:2006yc}. After choosing a polarization, there thus
exists a unique (up to isomorphism) irreducible Stone--von~Neumann representation of
the Heisenberg group $\Heis(\widetilde{\cG}\,)$ on which the central
elements $(0,z)$ are realised as multiplication by $z\in\uo$; we
identify this representation with the quantum Hilbert space $\calH= {\rm
  L}^2\big(\check\RH{}^p(N)\big)$ of the higher abelian gauge theory.

\subsection{Quantization of Ramond--Ramond gauge theory\label{QuantRR}}

We will now consider the hamiltonian quantization of Ramond--Ramond
gauge fields in the absence of D-brane sources, i.e. as induced solely
by the closed string background.
In the free Ramond--Ramond gauge theory on $M=\IR\times N$ we have to contend with
self-duality. In this case we set
\beqa
\widetilde{\cG}=\cG=\check\K{}^j(N) \ ,
\eeqa
and the commutator map $s$ is obtained by restriction from
$\check\K{}^j(N)\times \check\K{}^j(N)$ to its ``diagonal'' subgroup,
in a sense that we now explain.

By composing the cup product with the integration map on differential
K-theory, we define an intersection form 
\beqa
(-,-) \,:\, \check\K{}^j(N)\otimes \check\K{}^j(N) \ \xrightarrow{ \
  \smile \ } \ \check\K{}^0(N) \ \xrightarrow{ \ \int^{\check\K}\!\!\!\int_N \ }
\ \check\K{}^{-1}(\pt)\cong \bbt \ ,
\eeqa
where we have used Bott periodicity and the assumption that $n=\dim(N)$ is odd.
Explicitly, for Ramond--Ramond potentials $[\check C]$ and $[\check
C'\,]$ in complementary degrees, i.e. $\deg[\check C]+\deg[\check
C'\,] =n+1$, one has
\beq
\big( [\check C]\,,\, [\check C'\,] \big):= \int^{\check\K}\!\!\!\!\!\int_N \,
[\check C]\smile [\check C'\, ] \ \in \ \bbt \ .
\label{intformdiffK}\eeq
By the general properties of Pontrjagin self-dual generalized
cohomology theories, this pairing is \emph{perfect}, i.e. it induces
an isomorphism
\beqa
\check\K{}^j(N)\cong \Hom_{\scrAb}\big(\check\K{}^j(N)\,,\, \bbt\big)
\ .
\eeqa
But it is not necessarily \emph{antisymmetric}, due to
graded-commutativity of the cup product, e.g. in even degree $j=0$
this pairing is symmetric. To this end we use the Adams operation
(\ref{Adamsop}) (lifting the complex conjugation map $\Psi^{-1}$ of
\S\ref{Selfduality}) to define a new pairing by
\beq
\big\langle [\check C]\,,\, [\check C'\,] \big\rangle:= \Big[\,
\int^{\check\K}\!\!\!\!\!\int_N \,
[\check C]\smile \check\Psi^{-1}[\check C'\, ] \, \Big]_{u^0} \ \in \
\bbt \ ,
\label{diffKpairing}\eeq
where again we regard differential forms on $N$ as elements of the
graded vector space $\Omega(N;K^\bullet)^\bullet$ with $K^\bullet=\IR(u)$.

\begin{theorem}
The
pairing $\langle-,-\rangle:\check\K{}^j(N)\otimes \check\K{}^j(N) \to
\bbt$ is non-degenerate, and it is antisymmetric in dimensions $n\equiv1 \ {\rm mod} \ 4$.
\label{pairingthm}\end{theorem}
\begin{proof}
The operator $u^{-\ell}\,\check\Psi^{-1}$ is an involution, and hence
an isomorphism, and so since $(-,-)$ is non-degenerate, it follows
that the pairing $\langle-,-\rangle$ is also non-degenerate.
For illustration, we will prove antisymmetry for topologically trivial Ramond--Ramond
fields in Type~IIA string theory, i.e. potentials $[\check
C]=[C]\in\check\K{}^0(N)$ which can be represented
by odd degree differential forms $C\in\Omega(N;K^\bullet )^{-1}$
(modulo exact forms). Setting $n=2k+1$, the pairing is then given
modulo $\IZ$ by
\beqa
\big\langle [C]\,,\, [C'\,] \big\rangle &=& \Big[\, \int_N \,
C\wedge \dd \Psi^{-1}(C'\,) \, \Big]_{u^0} \quad {\rm mod} \ \IZ
\\[4pt] &=& \Big[\,
\int_N\,\Big(\,\sum_{j=0}^k \, u^{-j-1}\otimes C_{2j+1}\,\Big)\wedge\Big(\,
u^{k+1}\,\sum_{l=0}^k\,
(-1)^{l}\,u^{-l-1}\otimes \dd C'_{2l+1}\,\Big)\,\Big]_{u^0} \nonumber\\[4pt]
&=&(-1)^{k+1}\,\sum_{j=0}^k\,(-1)^j \ \int_N \,C_{2j+1}\wedge
 \dd C'_{2(k-j)-1} \nonumber\\[4pt] &=&(-1)^{k+1}\,
\sum_{j=0}^k\,(-1)^j \ \int_N \,C'_{2(k-j)-1} \wedge
 \dd C_{2j+1} \nonumber\\[4pt]
&=& (-1)^{k+1}\, \Big[\, \int_N \,
C' \wedge \dd \Psi^{-1}(C) \, \Big]_{u^0} \ = \ (-1)^{k+1}\,
\big\langle [C'\,]\,,\, [C] \big\rangle \ ,
\eeqa
where $C_p,C_p' \in\Omega^p(N)$; we have used the fact that $N$ has no
boundary and that the de~Rham differential is a skew-derivation on
forms. It follows that the pairing $\langle-,-\rangle$ is
antisymmetric only when $k$ is an even integer, i.e. when
$n= \dim(N)=1+4r$ with $r\in\mathbb{N}$. In the general case, the
proof given in~\cite{Freed:2006ya} involves the definition of an
orthogonal version of differential K-theory, and its value on a point.
\end{proof}

Thm.~\ref{pairingthm} applies in particular to the case relevant
to the physical Type~II superstring theory, where $N$ is a
nine-dimensional spin manifold. Using it we can now define
the Type~IIA/IIB commutator map
\beqa
s\,:\, \check\K{}^j(N)\times \check\K{}^j(N) \ \longrightarrow \ \uo \ ,
\qquad s\big( [\check C]\,,\, [\check C'\,] \big) = \exp\Big(2\pi\ii
\big\langle [\check C]\,,\, [\check C'\,] \big\rangle \Big)
\eeqa
in even/odd degree. It is bilinear, skew-symmetric, and
non-degenerate, but it is \emph{not} necessarily alternating. One has
$s(g,g)^2=1$ for all $g\in\cG$, so that $s(g,g)=\exp(\pi\ii
\varepsilon(g))$, where $\varepsilon:\cG\to\IZ_2$ is a group homomorphism
which defines a \emph{$\IZ_2$-grading}. An alternating commutator map
$\widetilde{s}$ is then defined by
\beq
\widetilde{s}\big( [\check C]\,,\, [\check C'\,] \big) =
\exp\Big(2\pi\ii \big(
\big\langle [\check C]\,,\, [\check C'\,]\big\rangle
-\mbox{$\frac12$}\, \varepsilon[\check C]\, \varepsilon[\check C'\,] \big)
\Big) \ .
\label{tildecommmap}\eeq
This definition does not depend on the chosen lifts. The degree
\beqa
\varepsilon[\check C] = 2\, \big\langle [\check C]\,,\, [\check C] \big\rangle =
2\, \Big[\, \int^{\check\K}\!\!\!\!\!\int_N \,
[\check C]\smile \check\Psi^{-1}[\check C] \, \Big]_{u^0} \ \in \ \IZ_2
\eeqa
depends only on the
characteristic class of $[\check C]$,
i.e. $\varepsilon\in\Hom_{\scrAb}(\K^j(N),\IZ_2)$; for $j=0$ it can be identified with the mod~2 index of the Dirac operator on
$N$ coupled to the virtual bundle $\xi\otimes\overline{\xi}$, where
$\xi=c([\check C]) $.

Using Prop.~\ref{class2} we may now define the Heisenberg group
extension $\Heis\big(\check\K{}^j(N)\big)$ of the
differential K-theory group $\cG=\check\K{}^j(N)$ associated to the
commutator map $s$, which is unique up to (non-canonical)
isomorphism. It is also
$\IZ_2$-graded with degree map
$\varepsilon:\Heis\big(\check\K{}^j(N)\big) \to \IZ_2$, with the maps
in the central extension $\IZ_2$-graded and the $\uo$ subgroup of even
degree. The quantum Hilbert space $\hil$ is also $\IZ_2$-graded, and the unique $\IZ_2$-graded irreducible
representation of $\Heis\big(\check\K{}^j(N)\big)$ (after a choice of
polarization) is compatible with the $\IZ_2$-grading on $\End_\IC(\hil)$.

\begin{definition}
The quantum Hilbert space of the Ramond--Ramond gauge theory
$\hil_{\rm RR}$ is the unique irreducible $\IZ_2$-graded unitary
representation of the Heisenberg group
$\Heis\big(\check\K{}^j(N)\big)$ of positive energy defined by the commutator map
(\ref{tildecommmap}) which is compatible with the polarization
discussed in \S\ref{Quantgenab}, with the property that the central subgroup $\uo$ acts by
scalar multiplication, where $j=0/-1$ for the Type~IIA/IIB string theory respectively.
\label{hilRRdef}\end{definition}

The $\IZ_2$-grading implies, in particular, that the quantum Hilbert
space $\hil_{\rm RR}$ contains both bosonic and fermionic states. Let
\beqa
\cO\,: \, \Heis\big(\check\K{}^j(N)\big) \ \longrightarrow \ 
\End_\IC(\hil_{\rm RR})
\eeqa
denote the irreducible representation of Def.~\ref{hilRRdef}. Given classes
$[\check C], [\check C'\, ]\in\check\K{}^j(N)$, let
$\check\calc=([\check C],z)$ and $\check\calc{}'=([\check C'\,],z'\, )$ be
lifts to the Heisenberg group $\Heis\big(\check\K{}^j(N)\big)$. From
(\ref{Heisgroupcomm}), with the commutator understood as the graded
group commutator in $\Heis\big(\check\K{}^j(N)\big)$, together with
$\cO\big((0,z)\big)=z\, \Id_{\hil_{\rm RR}}$ for all $z\in\uo$, it follows that
the commutation relations among the corresponding unitary operators on the
(infinite-dimensional) quantum Hilbert space $\hil_{\rm RR}$ of the
Ramond--Ramond gauge theory are given by
\beq
\big[\cO(\check\calc\, )\,,\, \cO(\check\calc\,'\,)\big] =
\widetilde{s}\big( [\check C]\,,\, [\check C'\,] \big) \
\Id_{\hil_{\rm RR}}
\label{quantumcomms}\eeq
in $\End_\IC(\hil_{\rm RR})$.

\subsection{Noncommutative quantum flux sectors}

The quantum flux sectors of the free abelian gauge theory can be
described as follows. A state $\psi\in{\rm
  L}^2\big(\check\RH{}^p(N)\big)$ of definite electric flux $
E\in\RH^{n-p+1}(N;\IZ)$ is an eigenstate of translation by flat
fields, i.e.
\beqa
\psi(\check A+\phi) = \exp\Big(2\pi \ii \int^{\RH{}}\!\!\!\!\!\!\int_N
\, E\smile \phi\Big) \
\psi(\check A) \qquad \mbox{for} \quad \phi\in \RH^{p-1}(N;\torus) \ .
\eeqa
This defines a decomposition of the Hilbert space $\hil$ into electric
flux sectors labelled by $E\in\RH^{n-p+1}(N;\IZ)$. Note that an
analogous definition using electric fields $\check
E\in\check\RH{}^{n-p+1}(N)$ and arbitrary translations
$\check\phi\in\check\RH{}^{p-1}(N)$ would lead to wavefunctionals
$\psi$ which are neither compactly supported nor decaying; in
particular, $\check E_1$ is homotopic to $\check E_2$ if and only if
$\int^{\check\RH{}}\!\!\!\int_N\, \check\phi\smile\check
E_1=\int^{\check\RH{}}\!\!\!\int_N\, \check\phi\smile\check E_2$ for
$\phi\in\RH^{p-1}(N;\bbt)$, or equivalently if and only if
$\int^{\RH{}}\!\!\!\!\int_N\, \phi\smile c(\check E_1)= \int^{\RH{}}\!\!\!\!\int_N\, \phi\smile c(\check E_2)$. Using dual
flat fields, we also get a decomposition of $\hil$ into magnetic flux sectors
labelled by $B\in\RH^{p}(N;\IZ)$. However, one cannot simultaneously
decompose $\hil$ into both electric and magnetic flux sectors because of the
Heisenberg commutation relations
\beqa
\big[\calu_E(\eta_e)\,,\,\calu_B(\eta_m)\big] = \exp\Big(2\pi \ii
\int^{\RH{}}\!\!\!\!\!\!\int_N \, \eta_e\smile \beta(\eta_m)\Big)
\ \Id_\hil
\ ,
\eeqa
where $\calu_{E}:\RH^{p-1}(N;\torus)\to \hil$ and
$\calu_{B}:\RH^{n-p}(N;\torus)\to \hil$ are quantization
maps corresponding to the electric and magnetic gradings of the Hilbert
space $\hil$, and $\beta$ is the Bockstein homomorphism (see
(\ref{CSflatseq})). Hence non-trivial commutators are related to
torsion in the cohomology $\RH^\bullet(N;\IZ)$; this means that the
Hilbert space $\hil$ can be simultaneously graded by electric and
magnetic fluxes only modulo torsion. See~\cite{Freed:2006ya,Freed:2006yc} for
further details. 

For $p=1$ and $M=S^1$, the generalized abelian gauge theory is the
theory of a periodic scalar field $g:S^1\to S^1$; in this case the
magnetic flux is the winding number and the electric flux is the
momentum of the field $g$.

For $p=2$ and $N=L_k=S^3/\IZ_k$ a lens space, there is only torsion in the
relevant cohomology groups $\RH^2(L_k;\IZ)=\IZ_k$ and
$\RH^1(L_k;\bbt)=\IZ_k$. The quantum Hilbert space is the unique
finite-dimensional irreducible representation of the Heisenberg group
extension
\beqa
1 \ \longrightarrow \ \IZ_k \ \longrightarrow \ {\tt
  Heis}(\IZ_k\times\IZ_k) \ \longrightarrow \ \IZ_k\times\IZ_k \
\longrightarrow \ 0 \ ,
\eeqa
and hence one cannot simultaneously measure
electric and magnetic flux in this case. A potential experimental test of this phenomenon is
described in~\cite{Kitaev:2007ed}: Although any embedded codimension
zero three-manifold in $\IR^3$ has torsion-free cohomology, it might
be possible to find a configuration where the effective space is only
immersed, with a line of double points, by using Josephson junctions.

Now let us look more closely at the quantum flux sectors of the
Ramond--Ramond gauge theory. As first pointed out in~\cite{Moore2000}, the quantization of \emph{flat}
Ramond--Ramond fields is of particular interest; these fluxes are described by classes in the K-module
theory which is isomorphic to the kernel of the curvature morphism
\beq
\K^{j-1}(N;\bbt)=\ker\Big(\check\K{}^j(N)\xrightarrow{ \ F \
}\Omega(N;K^\bullet)^j\Big) \ \xrightarrow{ \ \check \imath \ } \
\check\K{}^j(N) \ ,
\label{flatgroup}\eeq
where $\check \imath$ denotes the embedding and we assume $n=\dim(N)=4k+1$
for some $k\in\IN$. The Chern character $\ch:\K^j(N)\to
\RH(N;K^\bullet)^j$ becomes an isomorphism after tensoring over
$\IR$; its kernel coincides with the image of the connecting
homomorphism $\beta:\K^{j-1}(N;\bbt)\to \K^j(N)$, so that ${\rm
  im}(\beta)=\Tor\, \K^j(N)\subseteq\K^j(N)$ is the torsion
subgroup. We will show that the quantum commutators
(\ref{quantumcomms}) restrict non-degenerately to the flat
Ramond--Ramond fields if and only if the K-theory $\K^j(N)$ has
non-trivial torsion subgroup (and the real cohomology vanishes in the
opposite parity). 
On general grounds, the non-degenerate pairing on the torsion group
arises from the fact that there exists a non-degenerate pairing
between the group of components of the flat part and the torsion
subgroup of topological K-theory. Note that this sector of the
Ramond--Ramond gauge theory is ``topological'', in the sense that the
corresponding hamiltonian vanishes and there is no time evolution. In
this case the Heisenberg group $\Heis\big(\K^{j-1}(N;\bbt)\big)$ is a
finite-dimensional torus extended by a finite abelian group which
plays the role of the (finite-dimensional) configuration space of
fields; by Prop.~\ref{representation} it is represented uniquely on a
$\IZ_2$-graded finite-dimensional Hilbert space $\hil$. These torsion
fluxes arise entirely from Dirac quantization, and the corresponding
quantum operators do not all commute in the quantization by the full
K-theory group $\K^j(N)$. 

\begin{proposition}
Suppose that $\Tor\, \K^j(N)=0$. Let $\omega,\omega'\in
\K^{j-1}(N;\bbt)$ be classes of flat fields with lifts
$\tilde\omega,\tilde\omega'$ to the Heisenberg group
$\Heis\big(\K^{j-1}(N;\bbt)\big)$. Then
$$
\big[\Ocal(\tilde\omega)\,,\,\Ocal(\tilde\omega'\,)\big] =
\Id_{\hil_{\rm RR}} \ .
$$
\label{Ktors0thm}\end{proposition}
\begin{proof}
Since $\ch\circ c =[\, - \,]_{\rm dR}\circ F$, we have
\begin{equation*}
c\circ{\check \imath}(\omega) \ \in \ \Tor\, \K^j(N)
\end{equation*}
for $\omega\in \K^{j-1}(N;\bbt)$. One also has
\begin{displaymath}
\check\imath(\omega)\smile{[\check C]}=\check \imath\big(\omega \,
\dot{\smile} \, c([\check C])\big)
\end{displaymath} 
for $\omega \in \K^{j-1}(N;\bbt)$ and $[\check C]\in\check\K{}^{j'}(N)$, where
\begin{equation*}
\dot{\smile}\,:\, \K^{j-1}(N;\bbt)\otimes{\K}^{j'}
(N) ~ \longrightarrow~ \K^{j+j'-1}(N;\bbt)
\end{equation*}
denotes the restriction of the cup product. Then
\begin{eqnarray*}
\big\langle\,  \check\imath(\omega)\,,\, \check\imath(\omega'\,) \,
\big\rangle &=& \Big[\,\int^{\check
  \K}\!\!\!\!\!\int_{N} \, \check
\imath(\omega)\smile\check\Psi^{-1}\big(\, \check\imath(\omega'\,)  \, \big)\,
\Big]_{u^0}\\[4pt]
&=&\Big[\,\int^{\check
  \K}\!\!\!\!\!\int_{N} \, \check\imath(\omega)\smile
\check\imath \big(\Psi^{-1}(\omega'\,)\big) \, \Big]_{u^0} \\[4pt]
&=&\Big[\,\int^{\check
  \K}\!\!\!\!\!\int_{N} \, \check\imath\Big(\omega\,
\dot{\smile}\, c\circ\check{\imath}\big(\Psi^{-1}(\omega'\,)\big)\Big) \,
\Big]_{u^0} \ = \ 0 \ ,
\end{eqnarray*}
where we have used the fact that the Adams operation $\check\Psi^{-1}$
commutes with the embedding $\check\imath$ and the last equality
follows from $\Tor\,\K^j(N)=0$.
\end{proof}

The converse of Prop.~\ref{Ktors0thm} gives an explicit criterion and
formula for the non-trivial quantum commutators of flat fields in the
Ramond--Ramond gauge theory; its crux is the fact that the pairing (\ref{diffKpairing})
factors to a non-degenerate pairing on the torsion part of topological K-theory.

\begin{proposition}
If $\Tor\, \K^j(N)\neq0$, then there exist classes $\omega,\omega'\in\K^{j-1}(N;\bbt)$ such that
\beqa
\big[\cO(\tilde \omega)\,,\, \cO(\tilde\omega'\,)\big] \neq
\Id_{\hil_{\rm RR}} \ .
\eeqa
\label{torsthm}\end{proposition}
\begin{proof}
We show that the pairing $\langle-,-\rangle$ on
differential K-theory restricts non-degenerately to the subring
$\K^{j-1}(N;\bbt)\subset \check\K{}^j(N)$. Since $\K^{j-1}(N;\bbt)\cong\Hom_{\scrAb}(\K^j(N),\bbt)$, there is a non-degenerate pairing
\beqa
(-,-)_{\K} \,:\, \K^{j-1}(N;\bbt) \otimes \K^j(N) \ \longrightarrow \
\bbt
\eeqa
given by a formula like that in (\ref{intformdiffK}). Combined with
the pairing
\beqa
\Omega_\IZ(N;K^\bullet)^j\otimes \big( \Omega(N;K^\bullet)^{j-1}\,
\big/ \, \Omega_\IZ(N;K^\bullet)^{j-1}\big) \ \longrightarrow \ \bbt \
, \qquad \big(F\,,\,[C]\big) \ \longmapsto \ \Big[\, \int_N\, F\wedge
C\, \Big]_{u^0}
\eeqa
and the two exact sequences (\ref{diffKexactseqs}) of differential
K-theory, one proves explicitly that the pairing (\ref{intformdiffK})
is a perfect pairing. Recall from \S\ref{DiffK} that
the group of flat fields (\ref{flatgroup}) sits in the long exact sequence
\beqa
\cdots \ \longrightarrow \ \K^{j-1}(N) \ \xrightarrow{ \ \ch \ } \ \RH(N;K^\bullet)^{j-1} \
\longrightarrow \ \K^{j-1}(N;\bbt) \ \xrightarrow{ \ \delta \ } \
\K^j(N) \ \xrightarrow{ \ \ch \ } \ \RH(N;K^\bullet)^j \
\longrightarrow \ \cdots
\eeqa
which induces the short exact sequence
\beq
0 \ \longrightarrow \ \RH(N;K^\bullet)^{j-1} \, \big/ \,
\ch\big(\K^{j-1}(N)\big) \ \longrightarrow \ \K^{j-1}(N;\bbt) \
\xrightarrow{ \ \beta \ } \ \Tor\, \K^j(N) \ \longrightarrow 0 \ .
\label{flatshortexact}\eeq
There is also a short exact sequence
\begin{eqnarray*}
0~\longrightarrow~ \Hom_{\scrAb}\big(\RH(N;K^\bullet)^j\,,\,\bbt\big)
&\longrightarrow& \Hom_{\scrAb}\big(\K^j(N)\,,\,\bbt\big)~
\longrightarrow \\ &\longrightarrow&  \Hom_{\scrAb}\big(\Tor\,\K^j(N)\,,\,\bbt\big)~
\longrightarrow~ 0
\end{eqnarray*}
obtained by applying the exact contravariant functor
$\Hom_{\scrAb}\big(-,\bbt)$ to the short exact sequence induced by taking the kernel of the Chern character homomorphism
$\ch:\K^j(N)\to \RH(N;K^\bullet)^j$. We then obtain a commutative diagram with exact
horizontal sequences given by
$$
{\scriptsize\xymatrix{
0~
\ar[r]&~\dfrac{\RH(N;K^\bullet)^{j-1}}{{\ch}\big(\K^{j-1}(N)
  \big)}~ \ar[r]\ar[d]& ~\K^{j-1}(N;\bbt)~
\ar[r]^{\beta}\ar[d]^{\approx}& ~\Tor\, \K^j(N) ~ \ar[r]\ar[d]& ~0\\
0~ \ar[r]& \ \Hom_{\scrAb}\big(\RH(N;K^\bullet)^j \,,\,\bbt\big)~ \ar[r]&
~\Hom_{\scrAb}\big(\K^j(N) \,,\,\bbt \big)~
\ar[r]& ~\Hom_{\scrAb}\big(\Tor\, \K^j(N) \,,\,\bbt \big)~ \ar[r]&
~0
}}
$$
where the vertical morphism on the left is given by composing wedge product,
integration, and reduction modulo $\IZ$, and in the middle isomorphism
we have used Pontrjagin duality of the $\IR/\IZ$ K-theory
$\K^{j-1}(N;\bbt)$. By Pontrjagin duality on $K^\bullet$-valued cohomology
$\RH(N;K^\bullet)^j$ it follows that the left vertical morphism is an isomorphism,
and hence so is the right vertical morphism. Denote by
\beq
(-,-)_{\Tor} \,:\, \Tor\, \K^j(N) \otimes\Tor\, \K^j(N) \
\longrightarrow \ \bbt
\label{torsionpairing}\eeq
the \emph{torsion pairing} associated to these morphisms. Then the
calculation in the proof of Prop.~\ref{Ktors0thm} shows that
$$
\big\langle \, \check\imath(\omega)\,,\, \check\imath(\omega'\,)
\big\rangle=\big(\omega\,,\,c(\,\check\imath(\omega'\,)\, ) \big)_{\K}
\ ,
$$
and since $\ch\circ
c=[-]_{\rm dR}\circ F$, we have $c(\, \check i(\omega'\,))\in \Tor\, \K^j(N)$
and it follows that
$$
\big\langle \, \check\imath(\omega)\,,\, \check\imath(\omega'\,)
\big\rangle =\big(\beta(\omega)\,,\,c(\, \check\imath(\omega'\,) \,)
\big)_{\Tor} \ ,
$$
where we have used the fact that the kernel torus $\ker(\beta)$ of the Bockstein
homomorphism has trivial cup product with the torsion elements in $\Tor\,
\K^j(N)$. Since the pairing $(-,-)_{\Tor}$ is
non-degenerate, and $\Tor\, \K^j(N) \neq0$, we can always find classes
$\omega$ and $\omega'$ in $\K^{j-1}(N;\bbt)$ such that $\langle \,
\check\imath(\omega), \check\imath(\omega'\,)\rangle \neq0$. 
\end{proof}

An explicit formula for the torsion pairing (\ref{torsionpairing})
can be written as follows. Let $E\to N$ be a
complex vector bundle and $\rank(E)$ the trivial vector bundle over
$N$ of the same rank as $E$. Then the K-theory class
$\xi=[E]-[\rank(E)]$ is torsion, so there exits an integer $k$ and an
isomorphism on K-theory $\psi:E^{\oplus k}\to k\, [\rank(E)]$. Let
$\nabla$ be a connection on $E$, and $\nabla_0$ the trivial connection
with vanishing holonomy. Then the K-theory integral of the class $\xi$
can be expressed as
\beqa
\int^{\K}\!\!\!\!\!\int_N \,
\xi=\eta(\Dirac_\nabla)-\eta(\Dirac_{\nabla_0})- \frac1k\, \int_N\,
\CS\big(\psi^*(\nabla^{\oplus k})\,,\, \nabla_0^{\oplus k}\big) \wedge
\widehat{A}(N) \quad {\rm mod} \ \IZ \ ,
\eeqa
where $\eta(\Dirac_\nabla)$ is the spectral asymmetry of the Dirac operator
on $N$ coupled to the bundle $E$.

The flat fluxes play a crucial role in the grading on the quantum
Hilbert space $\mathcal{H}_{\rm RR}$ of the
Ramond--Ramond gauge theory into topological
sectors~\cite{Freed:2006ya,Freed:2006yc}. As explained
by~\cite{Moore2000,Freed:2006ya,Freed:2006yc}, the subgroup of
differential K-theory comprinsing flat cocycles is the group of
unbroken gauge symmetries of the Ramond--Ramond gauge theory. The
equivalence classes comprising shifts of cycles by flat fields define
the topological classes of Ramond--Ramond fluxes. Recall that the characteristic classes
$\xi\in\K^j(N)$ label the connected components of
$\check\K{}^j(N)$. Hence there is a natural grading of the quantum Hilbert
space $\hil_{\rm RR}$ of the Ramond--Ramond gauge theory into topological sectors by the K-theory group $\K^j(N)$ modulo
torsion; it is induced by diagonalising the translation action by the
flat Ramond--Ramond fields. The group of components of this subgroup,
which is isomorphic to the torsion part of the topological
K-theory, can shift the Hilbert space gradings. In that case, the
grading can only be defined modulo these torsion subgroups. By Prop.~\ref{torsthm}, if the topological K-theory $\K^j(N)$
of the space $N$ has non-trivial torsion subgroup, then elements of the
subgroup of flat Ramond--Ramond fields do not commute in the Heisenberg extension of the
differential K-theory group and quantize to
operators which do not all
commute among themselves; this expresses an uncertainty
principle which asserts that the K-theory class of a Ramond--Ramond
field cannot be measured. Explicit examples exhibiting this phenomenon
can be found in~\cite[Ex.~2.21]{Freed:2006ya}
and~\cite[\S5]{Freed:2006yc}. Flat Ramond--Ramond fluxes also give rise to novel effects in certain flux
compactifications of Type~IIA string
theory~\cite[\S4.5]{dBDHKMMS}. If $N$ is any smooth manifold of dimension
$4k+1$ with finite
abelian fundamental group, then the K-theory
$\K^0(N)$ contains a non-trivial torsion subgroup. A simple example is the manifold
$N=L_k\times S^6$ where $L_k\cong S^3/\IZ_k$ is a three-dimensional
lens space; then $\Tor\, \K^0(N)=\IZ_k\oplus\IZ_k$.

\bigskip

\section*{Acknowledgments}

\noindent
This article is based partly on a series of lectures given by the
author at the
school on ``Higher Index Theory and Differential K-Theory'' which was
held at the Mathematisches Institut of Georg-August-Universit\"at in
G\"ottingen, Germany on October 18--22, 2010; the author thanks Alexander Kahle,
Thomas Schick and Alessandro Valentino for the invitation and
hospitality during the school. Some of the details in these notes
arose in the course of discussions and correspondence with Ulrich Bunke,
Thomas Schick and Alessandro Valentino, whom we warmly thank for the
collaboration. Some parts of this paper are
also influenced by discussions and presentations at the workshop on ``Differential
Cohomology'' which was held at the Simons Center for Geometry
and Physics in Stonybrook, New York on January 10--14, 2011; the
author warmly thanks Dan Freed, Greg Moore and Dennis Sullivan for
the invitation to participate. Portions of this paper were 
developed at the Erwin Schr\"odinger International Institute for Mathematical
Physics in Vienna, Austria during July 2012 under the auspices of the
programme ``K-Theory and Quantum Field Theory''; we thank
Alan Carey, Harald Grosse and Jouko Mickelsson for the invitation to
participate and present some of this material. Finally, the author would like to thank Andrey Bytsenko
for the invitation to contribute this article to the proceedings.
This work was supported in part by the
Consolidated Grant ST/J000310/1 from the UK Science and Technology
Facilities Council, and by Grant RPG-404 from the Leverhulme Trust.

\bigskip

%\newpage


\begin{thebibliography}{99}
%\addtolength{\itemsep}{-6pt}

%\cite{Baez}
\bibitem{Baez}
  J.C.~Baez and A.D.~Lauda,
  ``Higher-dimensional algebra V: 2-groups,''
  Theory Appl. Categ. {\bf 12} (2004) 423--491
  [arXiv:math.QA/0307200].

%\cite{Baez1}
\bibitem{Baez1}
J.C.~Baez and U.~Schreiber,
``Higher gauge theory,''
Contemp. Math. {\bf 431} (2007) 7--30 
[arXiv:math.DG/0511710].

%\cite{BaumDouglas}
\bibitem{BaumDouglas}
P.~Baum and R.G.~Douglas,
``K-homology and index theory,''
Proc. Symp. Pure Math. {\bf 38} (1982) 117--173.

%\cite{Belov:2006jd}
\bibitem{Belov:2006jd}
  D.M.~Belov and G.W.~Moore,
  ``Holographic action for the self-dual field,''
  arXiv:hep-th/0605038.
  %%CITATION = HEP-TH/0605038;%%

%\cite{Belov:2006xj}
\bibitem{Belov:2006xj}
  D.M.~Belov and G.W.~Moore,
  ``Type II actions from 11-dimensional Chern--Simons theories,''
  arXiv:hep-th/0611020.
  %%CITATION = HEP-TH/0611020;%%

%\cite{Benamour}
\bibitem{Benamour}
M.-T.~Benameur and M.~Maghfoul, 
``Differential characters
  in K-theory,''
Diff. Geom. Appl. {\bf 24} (2006) 417--432.

%\cite{BGV}
\bibitem{BGV}
N.~Berline, E.~Getzler and M.~Vergne,
\emph{Heat Kernels and Dirac Operators}
(Springer-Verlag, Berlin, 2004).

%\cite{Brylinski}
\bibitem{Brylinski}
J.-L.~Brylinski,
{\it Loop Spaces, Characteristic Classes and Geometric Quantization}
(Birkh\"auser, Boston, 2007).

%\cite{Brylinski2}
\bibitem{Brylinski2}
J.-L.~Brylinski and D.A.~McLaughlin,
``The geometry of degree $4$ characteristic classes I,''
Duke Math. J. {\bf 75} (1994) 105--138.

%\cite{Bunke2009}
\bibitem{Bunke2009}
U.~Bunke,
``Index theory, eta forms, and {D}eligne cohomology,''
Mem. Amer. Math. Soc. {\bf 198} (2009) 1--120
[arXiv:math.DG/0201112].

%\cite{Bunke2009b}
\bibitem{Bunke2009b}
U.~Bunke,
``Adams operations in smooth K-theory,''
Geom. Topol. {\bf 14} (2010) 2349--2381
[arXiv:0904.4355 [math.KT]].

%\cite{Bunke-rev}
\bibitem{Bunke-rev}
U.~Bunke,
``Differential cohomology,''
arXiv:1208.3961 [math.AT].

%\cite{BunkeSchick2007}
\bibitem{BunkeSchick2007}
  U.~Bunke and T.~Schick,
  ``Smooth K-theory,''
  Ast\'erisque {\bf 328} (2009) 45--135 
  [arXiv:0707.0046 [math.KT]].
  %%CITATION = ARXIV:0707.0046;%%

%\cite{BunkeSchick2009}
\bibitem{BunkeSchick2009}
  U.~Bunke and T.~Schick,
  ``Uniqueness of smooth extensions of generalized cohomology theories,''
  J. Topol. {\bf 3} (2010) 110--156 
  [arXiv:0901.4423 [math.DG]].
  %%CITATION = ARXIV:0901.4423;%%

%\cite{bunke-20092}
\bibitem{bunke-20092}
U.~Bunke and T.~Schick,
``Differential orbifold K-theory,''
arXiv:0905.4181 [math.KT].

%\cite{Bunke:2010mq}
\bibitem{Bunke:2010mq}
  U.~Bunke and T.~Schick,
  ``Differential K-theory. A survey,''
  Springer Proc. Math. {\bf 17} (2012) 303--358
  [arXiv:1011.6663 [math.KT]].
  %%CITATION = ARXIV:1011.6663;%%

%\cite{Carey:2002xp}
\bibitem{Carey:2002xp}
  A.L.~Carey, S.~Johnson and M.K.~Murray,
  ``Holonomy on D-branes,''
  J. Geom. Phys. {\bf 52} (2004) 186--216
  [arXiv:hep-th/0204199].
  %%CITATION = HEP-TH/0204199;%%

%\cite{Carey}
\bibitem{Carey}
A.L.~Carey, J.~Mickelsson and B.-L.~Wang,
``Differential twisted K-theory and applications,''
J. Geom. Phys. {\bf 59} (2009) 632--653
[arXiv:0708.3114 [math.KT]].

%\cite{Alexander:1985aa}
\bibitem{Alexander:1985aa}
J.~Cheeger and J.~Simons,
``Differential characters and geometric invariants,''
Lect. Notes Math. {\bf 1167} (1985) 50--80.

%\cite{dBDHKMMS}
\bibitem{dBDHKMMS}
  J.~de Boer, R.~Dijkgraaf, K.~Hori, A.~Keurentjes, J.~Morgan, D.R.~Morrison and S.~Sethi,
  ``Triples, fluxes and strings,''
  Adv.\ Theor.\ Math.\ Phys.\  {\bf 4} (2002) 995--1186
  [arXiv:hep-th/0103170].
  %%CITATION = HEP-TH/0103170;%%

%\cite{Diaconescu:2000wy}
\bibitem{Diaconescu:2000wy}
  D.-E.~Diaconescu, G.W.~Moore and E.~Witten,
  ``$\E_8$ gauge theory and a derivation of K-theory from M-theory,''
  Adv.\ Theor.\ Math.\ Phys.\  {\bf 6} (2003) 1031--1134
  [arXiv:hep-th/0005090].
  %%CITATION = HEP-TH/0005090;%%

%\cite{Diaconescu:2003bm}
\bibitem{Diaconescu:2003bm}
  D.-E.~Diaconescu, G.W.~Moore and D.S.~Freed,
  ``The M-theory three-form and $\E_8$ gauge theory,''
  arXiv:hep-th/0312069.
  %%CITATION = HEP-TH/0312069;%%

%\cite{DFM}
\bibitem{DFM}
  J.~Distler, D.S.~Freed and G.W.~Moore,
  ``Orientifold pr\'ecis,''
  arXiv:0906.0795 [hep-th].
  %%CITATION = ARXIV:0906.0795;%%

%\cite{DFM2}
\bibitem{DFM2}
  J.~Distler, D.S.~Freed and G.W.~Moore,
``Spin structures and superstrings,''
  arXiv:1007.4581 [hep-th].
  %%CITATION = ARXIV:1007.4581;%%

%\cite{Dunne:1989hv}
\bibitem{Dunne:1989hv}
  G.V.~Dunne, R.~Jackiw and C.A.~Trugenberger,
  ``Topological (Chern--Simons) quantum mechanics,''
  Phys.\ Rev.\ D {\bf 41} (1990) 661--666.
  %%CITATION = PHRVA,D41,661;%%

%\cite{Freed:2000ta}
\bibitem{Freed:2000ta}
  D.S.~Freed,
  ``Dirac charge quantization and generalized differential cohomology,''
  {Surv. Diff. Geom.} {\bf VII} (2000) 129--194
  [arXiv:hep-th/0011220].
  %%CITATION = HEP-TH/0011220;%%

%\cite{Freed:2002qp}
\bibitem{Freed:2002qp}
  D.S.~Freed,
  ``K-theory in quantum field theory,''
  Current Develop. Math. {\bf 2001} (2002) 41--87
  [arXiv:math-ph/0206031].
  %%CITATION = MATH-PH/0206031;%%

%\cite{Freed:2000tt}
\bibitem{Freed:2000tt}
  D.S.~Freed and M.J.~Hopkins,
  ``On Ramond--Ramond fields and K-theory,''
  J. High Energy Phys. {\bf 0005} (2000) 044
  [arXiv:hep-th/0002027].
  %%CITATION = HEP-TH/0002027;%%

%\cite{FreedLott}
\bibitem{FreedLott}
  D.S.~Freed and J.~Lott,
  ``An index theorem in differential K-theory,''
  Geom. Topol. {\bf 14} (2010) 903--966 
  [arXiv:0907.3508 [math.DG]].
  %%CITATION = ARXIV:0907.3508;%%

%\cite{Freed:2004yc}
\bibitem{Freed:2004yc}
  D.S.~Freed and G.W.~Moore,
  ``Setting the quantum integrand of M-theory,''
  Commun.\ Math.\ Phys.\  {\bf 263} (2006) 89--132
  [arXiv:hep-th/0409135].
  %%CITATION = HEP-TH/0409135;%%

%\cite{Freed:2006ya}
\bibitem{Freed:2006ya}
  D.S.~Freed, G.W.~Moore and G.~Segal,
  ``The uncertainty of fluxes,''
  Commun.\ Math.\ Phys.\  {\bf 271} (2007) 247--274
  [arXiv:hep-th/0605198].
  %%CITATION = HEP-TH/0605198;%%

%\cite{Freed:2006yc}
\bibitem{Freed:2006yc}
  D.S.~Freed, G.W.~Moore and G.~Segal,
  ``Heisenberg groups and noncommutative fluxes,''
  Ann. Phys.\  {\bf 322} (2007) 236--285
  [arXiv:hep-th/0605200].
  %%CITATION = HEP-TH/0605200;%%

%\cite{FHLT}
\bibitem{FHLT}
D.S.~Freed, M.J.~Hopkins, J.~Lurie and C.~Teleman,
``Topological quantum field theories from compact Lie groups,''
arXiv:0905.0731~[math.AT].

%\cite{Gawedzki:1987ak}
\bibitem{Gawedzki:1987ak}
  K.~Gaw\c{e}dzki,
  ``Topological actions in two-dimensional quantum field theories,''
in: {\sl Nonperturbative Quantum Field Theory}, eds. G.~'t~Hooft,
A.~Jaffe, G.~Mack, P.K.~Mitter and R.~Stora (Plenum Press, New York,
1988) 101--142.

%\cite{Kiyonori:2001aa}
\bibitem{Kiyonori:2001aa}
K.~Gomi and Y.~Terashima,
``Higher-dimensional parallel transports,''
Math. Res. Lett. {\bf 8} (2001) 25--33.

%\cite{Henneaux:1988gg}
\bibitem{Henneaux:1988gg}
  M.~Henneaux and C.~Teitelboim,
  ``Dynamics of chiral (self-dual) $p$-forms,''
  Phys.\ Lett.\ B {\bf 206} (1988) 650--654.
  %%CITATION = PHLTA,B206,650;%%

%\cite{Hopkins:2002rd}
\bibitem{Hopkins:2002rd}
  M.J.~Hopkins and I.M.~Singer,
  ``Quadratic functions in geometry, topology, and M-theory,''
  J.\ Diff.\ Geom.\  {\bf 70} (2005) 329--452
  [arXiv:math.AT/0211216].
  %%CITATION = MATH/0211216;%%

%\cite{Johnson}
\bibitem{Johnson}
S.~Johnson,
``Constructions with bundle gerbes,''
PhD Thesis, University of Adelaide, 2002 [arXiv:math.DG/0312175].

%\cite{Kahle}
\bibitem{Kahle}
A.~Kahle and A.~Valentino,
``T-duality and differential K-theory,''
arXiv:0912.2516 [math.KT].

%\cite{Kelnhofer:2007jf}
\bibitem{Kelnhofer:2007jf}
  G.~Kelnhofer,
  ``Functional integration and gauge ambiguities in generalized abelian gauge theories,''
  J.\ Geom.\ Phys.\  {\bf 59} (2009) 1017--1035
  [arXiv:0711.4085 [hep-th]].
  %%CITATION = ARXIV:0711.4085;%%

%\cite{Kitaev:2007ed}
\bibitem{Kitaev:2007ed}
  A.~Kitaev, G.W.~Moore and K.~Walker,
  ``Noncommuting flux sectors in a tabletop experiment,''
  arXiv:0706.3410 [hep-th].
  %%CITATION = ARXIV:0706.3410;%%

%\cite{Lott}
\bibitem{Lott}
J.~Lott,
``$\IR/\IZ$ index theory,''
Commun. Anal. Geom. {\bf 2} (1994) 279--311.

%\cite{Mathai:2005sw}
\bibitem{Mathai:2005sw}
  V.~Mathai and D.~Roberts,
  ``Yang--Mills theory for bundle gerbes,''
  J.\ Phys.\ A {\bf 39} (2006) 6039--6044
  [arXiv:hep-th/0509037].
  %%CITATION = HEP-TH/0509037;%%

%\cite{Minasian:1997mm}
\bibitem{Minasian:1997mm}
  R.~Minasian and G.W.~Moore,
  ``K-theory and Ramond--Ramond charge,''
  J. High Energy Phys. {\bf 9711} (1997) 002
  [arXiv:hep-th/9710230].
  %%CITATION = HEP-TH/9710230;%%

%\cite{Monnier:2010ww}
\bibitem{Monnier:2010ww}
  S.~Monnier,
  ``Geometric quantization and the metric dependence of the self-dual field theory,''
  Commun.\ Math.\ Phys.\  {\bf 314} (2012) 305--328
  [arXiv:1011.5890 [hep-th]].
  %%CITATION = ARXIV:1011.5890;%%

%\cite{Moore2000}
\bibitem{Moore2000}
  G.W.~Moore and E.~Witten.
  ``Self-duality, Ramond--Ramond fields and K-theory,''
  {J. High Energy Phys.} {\bf 0005} (2000) 032
  [arXiv:hep-th/9912279].
  %%CITATION = JHEPA,0005,032;%%

%\cite{Olsen:1999xx}
\bibitem{Olsen:1999xx}
  K.~Olsen and R.J.~Szabo,
  ``Constructing D-branes from K-theory,''
  Adv.\ Theor.\ Math.\ Phys.\  {\bf 3} (1999) 889--1025
  [arXiv:hep-th/9907140].
  %%CITATION = HEP-TH/9907140;%%

%\cite{Ortiz}
\bibitem{Ortiz}
  M.L.~Ortiz,
  ``Differential equivariant K-theory,''
  arXiv:0905.0476 [math.AT].
  %%CITATION = ARXIV:0905.0476;%%

%\cite{Reis:2005pp}
\bibitem{Reis:2005pp}
  R.M.G.~Reis and R.J.~Szabo,
  ``Geometric K-homology of flat D-branes,''
  Commun.\ Math.\ Phys.\  {\bf 266} (2006) 71--122
  [arXiv:hep-th/0507043].
  %%CITATION = HEP-TH/0507043;%%

%\cite{Reis:2006th}
\bibitem{Reis:2006th}
  R.M.G.~Reis, R.J.~Szabo and A.~Valentino,
  ``KO-homology and Type I string theory,''
  Rev.\ Math.\ Phys.\  {\bf 21} (2009) 1091--1143
  [arXiv:hep-th/0610177].
  %%CITATION = HEP-TH/0610177;%%

%\cite{Sharpe}
\bibitem{Sharpe}
  E.R.~Sharpe,
  ``String orbifolds and quotient stacks,''
  Nucl.\ Phys.\ B {\bf 627} (2002) 445--505
  [arXiv:hep-th/0102211].
  %%CITATION = HEP-TH/0102211;%%

%\cite{Stevenson-thesis}
\bibitem{Stevenson-thesis}
D.~Stevenson,
``The geometry of bundle gerbes,''
PhD Thesis, University of Adelaide, 2000
  [arXiv:math.DG/0004117].
  %%CITATION = MATH/0004117;%%

%\cite{Szabo:2002jv}
\bibitem{Szabo:2002jv}
  R.J.~Szabo,
  ``D-branes, tachyons and K-homology,''
  Mod.\ Phys.\ Lett.\ A {\bf 17} (2002) 2297--2316
  [arXiv:hep-th/0209210].
  %%CITATION = HEP-TH/0209210;%%

%\cite{Szabo:2008hx}
\bibitem{Szabo:2008hx}
  R.J.~Szabo,
  ``D-branes and bivariant K-theory,''
  arXiv:0809.3029 [hep-th].
  %%CITATION = ARXIV:0809.3029;%%

%\cite{SV}
\bibitem{SV}
  R.J.~Szabo and A.~Valentino,
  ``Ramond--Ramond fields, fractional branes and orbifold differential K-theory,''
  Commun.\ Math.\ Phys.\  {\bf 294} (2010) 647--702
  [arXiv:0710.2773 [hep-th]].
  %%CITATION = ARXIV:0710.2773;%%

%\cite{Upmeier}
\bibitem{Upmeier}
M.~Upmeier,
``Products in generalized differential cohomology,''
arXiv:1112.4173 [math.GT].

%\cite{Valentino:2008xd}
\bibitem{Valentino:2008xd}
  A.~Valentino,
  ``K-theory, D-branes and Ramond--Ramond fields,''
  PhD Thesis, Heriot--Watt University, 2008
  [arXiv:0812.0682 [hep-th]].
  %%CITATION = ARXIV:0812.0682;%%

%\cite{Witten:1996md}
\bibitem{Witten:1996md}
  E.~Witten,
  ``On flux quantization in M-theory and the effective action,''
  J.\ Geom.\ Phys.\  {\bf 22} (1997) 1--13
  [arXiv:hep-th/9609122].
  %%CITATION = HEP-TH/9609122;%%

%\cite{Witten:1998cd}
\bibitem{Witten:1998cd}
  E.~Witten,
  ``D-branes and K-theory,''
  J. High Energy Phys. {\bf 9812} (1998) 019
  [arXiv:hep-th/9810188].
  %%CITATION = HEP-TH/9810188;%%

\end{thebibliography}
\end{document}